%% file: main.tex
\documentclass[11pt]{article}
\usepackage[utf8]{inputenc}
\usepackage[nottoc]{tocbibind}
\usepackage[english]{babel}
\usepackage{csquotes}

\usepackage[
  backend=biber,
  style=numeric,   
  url=false,          
  doi=false,          
  isbn=false,         
  giveninits=true,    
  maxnames=99,        
  sorting=none        
]{biblatex}
\AtEveryBibitem{\clearfield{issn}}   
\DeclareFieldFormat{eprint:arxiv}{%
  \href{https://arxiv.org/abs/#1}{arXiv\addcolon\space#1}%
}
\DeclareFieldFormat{eprint:stackexchange}{%
  \href{#1}{#1}%
}
\renewbibmacro{in:}{}

\usepackage{url}
\usepackage{amsmath}
\usepackage{amssymb}
\usepackage[dvipsnames]{xcolor}
\usepackage{url}
\usepackage{graphicx}
\usepackage{parskip}
\usepackage{xfrac}
\graphicspath{{images/}}
\usepackage{parskip}
\usepackage{fancyhdr}
\usepackage{vmargin}
\usepackage{listings}
\usepackage{xcolor}
\usepackage{dsfont}
\usepackage{tabu}
\usepackage{enumitem}
\usepackage{pdflscape}
\usepackage{systeme}
\usepackage{ stmaryrd }
\usepackage{physics}
\usepackage{pdfpages}
\usepackage{float}
\usepackage{mathrsfs} 
\usepackage{tensor}
\usepackage{upgreek}
\setmarginsrb{3 cm}{0 cm}{3 cm}{2.5 cm}{1 cm}{1.5 cm}{1 cm}{1.5 cm}
\usepackage{tikzit}
\input{style.tikzstyles}
\usepackage{caption}
\usepackage{subcaption}
\usepackage{amsthm}
\usepackage[colorlinks,citecolor={blue},urlcolor={blue}]{hyperref}
\usepackage{varioref}
\usepackage[capitalise]{cleveref}
\usepackage{array}
\usepackage{makecell}
\usepackage{authblk}

\crefname{assumption}{assumption}{assumptions}
\crefname{notation}{notation}{notation}

\usetikzlibrary{patterns.meta}

\newtheoremstyle{break}
  {\topsep}{\topsep}%
  {\itshape}{}%
  {\bfseries}{}%
\theoremstyle{break}
\newtheorem{theorem}{Theorem}[section]
\newtheorem{definition}[theorem]{Definition}

\newtheorem{lemma}[theorem]{Lemma}
\newtheorem{proposition}[theorem]{Proposition}
\newtheorem{remark}[theorem]{Remark}
\newtheorem{notation}[theorem]{Notation}
\newtheorem{example}[theorem]{Example}

\newtheorem*{theorem-non}{Theorem}
\newtheorem*{corollary-non}{Corollary}
\newtheorem*{conjecture-non}{Conjecture}

\makeatletter
\renewcommand\@maketitle{%
  \newpage
  \null
  \begin{center}%
    {\LARGE \@title \par}%
    \vskip 1em%
    {\large
      \lineskip .5em%
      \begin{tabular}[t]{c}%
        \@author
      \end{tabular}\par}%
  \end{center}%
  \par
  \vskip 1.5em}
\makeatother

\allowdisplaybreaks
\numberwithin{equation}{section}

\title{Tensor network manifolds and Riemannian fundamental theorem for tensor networks}
\author[1,2]{Pablo~Páez Velasco\thanks{pablopaez@ucm.es}}

\affil[1]{\textit{Departamento de An\'{a}lisis Matemático y Matemática Aplicada, Facultad de Matemáticas, Universidad Complutense de Madrid, 28040 Madrid, Spain}}
\affil[2]{\textit{Instituto de Ciencias Matemáticas, 28049 Madrid, Spain}}

\newcommand{\bb}[1]{\mathbb{#1}}

\crefformat{equation}{Eq.~(#2#1#3)} 

\begin{document}

\input{TikzFigures/preamble}

\maketitle
\renewcommand{\thefootnote}{\arabic{footnote}}
\setcounter{footnote}{0}

\begin{abstract}
Tensor networks provide a powerful framework for efficiently representing high-dimensional data and many-body quantum states. Endowing tensor networks with a Riemannian manifold structure provides a natural setting for numerical optimization and analysis. A central feature of tensor networks is their gauge freedom, whose characterisation---captured by so-called fundamental theorems---underlies both their intrinsic structure and the design of numerical algorithms. In this work, we study the interaction between the Riemannian manifold structure and the gauge freedom for several families of tensor networks. Using group actions and Riemannian submersions, we establish a \textit{Riemannian fundamental theorem} for the tensor network families studied.
\end{abstract}

\newpage
\hypersetup{linkcolor=black}
\tableofcontents
\hypersetup{linkcolor=red}

\newpage
\input{Chapters/Intro}
\newpage
\input{Chapters/TNManifolds}
\newpage
\clearpage
\section*{Acknowledgements}
{\sloppy 
P.\,P.\,V. would like to thank Alberto Ruiz de Alarcón, Sofyan Iblisdir, David Pérez García, Ángela Capel, Angelo Lucia and Marco Castrillón López for their invaluable help in the development of this work, and for their careful revision of the manuscript. P.\,P.\,V. acknowledges support of the Spanish Ministry of Science and Innovation MCIN/AEI/10.13039/501100011033 (CEX2023-001347-S, CEX2019-000904-S, CEX2019-000904-S-21-2). This work has been funded by the Spanish Ministry of Science, Innovation and Universities MICIU/AEI/10.13039/501100011033 (CEX2023-001347-S, PID2023-146758NB-I00), Comunidad de Madrid (TEC-2024/COM-84-QUITEMAD-CM), Universidad Complutense de Madrid (FEI-EU-22-06),  and the Ministry for Digital Transformation and of Civil Service of the Spanish Government through the QUANTUM ENIA project call - Quantum Spain project, and by the European Union through the Recovery, Transformation and Resilience Plan - NextGenerationEU within the framework of the Digital Spain 2026 Agenda. 
\par}

\printbibliography[heading=bibintoc]
\newpage 

\appendix 
\newpage
\input{Chapters/Geometry}

\end{document}

%% file: style.tikzstyles
\tikzstyle{black dot}=[fill=black, draw=black, shape=circle, minimum size=0.1mm, inner sep=1mm]
\tikzstyle{Rectangle}=[fill=white, draw=black, shape=rectangle, minimum width=2.5cm, minimum height=1.5cm, transform shape]
\tikzstyle{White Ball}=[fill=white, draw=black, shape=circle, minimum size=8.5mm, inner sep=0, transform shape]
\tikzstyle{White Square}=[fill=white, draw=black, shape=rectangle, transform shape, minimum width=1cm, minimum height=1cm]
\tikzstyle{oval}=[fill=white, draw=black, shape=ellipse, minimum width=8cm, minimum height=1cm, transform shape]

\tikzstyle{linea negra}=[-, line width=0.2mm, outer sep=0mm, inner sep=0mm]
\tikzstyle{flecha negra}=[->, line width=0.2mm, outer sep=0mm, inner sep=0mm]
\tikzstyle{linea discontinua}=[-, dotted, draw={rgb,255: red,77; green,77; blue,77}, line width=0.25mm]
\tikzstyle{linea azul}=[-, draw={rgb,255: red,0; green,187; blue,255}, line width=0.5mm]
\tikzstyle{new edge style 0}=[-, draw={rgb,255: red,255; green,0; blue,4}, line width=0.5mm]
\tikzstyle{new edge style 1}=[-, dotted, line width=0.3mm, draw={rgb,255: red,255; green,101; blue,103}]

%% file: TikzFigures/preamble.tex
\definecolor{myred}{HTML}{E53935}
\definecolor{myblue}{HTML}{1E88E5}
\definecolor{mygreen}{HTML}{43A047}
\definecolor{myyellow}{HTML}{FDD835}
\definecolor{myorange}{HTML}{FB8C00}
\definecolor{mygold}{HTML}{F9A825}
\definecolor{mypurple}{HTML}{8E24AA}
\definecolor{mygray}{HTML}{BDBDBD}
\definecolor{mybrown}{HTML}{6D4C41}
\definecolor{mynavy}{HTML}{1A237E}
\definecolor{mypink}{HTML}{ffbfca}
\definecolor{myseagreen}{HTML}{26A69A}
\definecolor{myviolet}{HTML}{f07ef0}
\definecolor{mydarkblue}{HTML}{0D47A1}
\definecolor{mydarkcyan}{HTML}{E0FFFF}
\definecolor{darkgray}{rgb}{0.66, 0.66, 0.66}
\definecolor{mydarkgreen}{HTML}{1B5E20}
\definecolor{mydarkmagenta}{HTML}{AD1457}
\definecolor{mydarkorange}{HTML}{EF6C00}
\definecolor{lightblue}{rgb}{0.68, 0.85, 0.9}
\definecolor{lightcyan}{rgb}{0.88, 1.0, 1.0}
\definecolor{lightgray}{rgb}{0.83, 0.83, 0.83}
\definecolor{mylightgreen}{HTML}{81C784}
\definecolor{lightyellow}{rgb}{1.0, 1.0, 0.88}
\definecolor{myshadow}{rgb}{0.5, 0.5, 0.5}
\definecolor{pink}{rgb}{1.0, 0.75, 0.8}
\definecolor{violet}{rgb}{0.93, 0.51, 0.93}
\definecolor{myauxcolor}{RGB}{245, 255, 255}
\definecolor{myauxcolor2}{RGB}{224, 246, 255}
\definecolor{mylightgreen}{RGB}{193, 225, 159}
\tikzstyle{heavier} = [line width=0.8pt]
\usetikzlibrary{arrows.meta}

%% file: Chapters/Intro.tex
\section{Introduction}
Tensor networks have become a central tool for describing and simulating strongly correlated quantum many-body systems, providing efficient parametrizations of states whose full Hilbert-space description would otherwise be computationally intractable. Among the most prominent examples are matrix product states (MPS) and projected entangled pair states (PEPS), which play a central role in the field of quantum many-body systems and quantum information \cite{cirac2021matrix}.

Beyond their representational power, tensor networks are also closely tied to some of the most successful numerical methods in quantum physics. White's seminal density matrix renormalization group (DMRG) algorithm \cite{white1992density} was later understood as an optimization algorithm over the manifold of MPS \cite{dukelsky1998equivalence}. This viewpoint was the germ of several extensions, both in 1D and higher, creating a body of research which is the state of the art in several optimization problems \cite{xiang2023density,ran2020tensor,banuls2023tensor}. Since quantum circuits are just another family of tensor networks, even variational quantum algorithms \cite{cerezo2021variational} can be interpreted as tensor network optimization procedures over the manifold defined by their parameter space. This perspective naturally led to the study of tensor networks from the viewpoint of Riemannian geometry. Endowing tensor network families with a Riemannian manifold structure provides a natural framework for formulating optimization algorithms, defining gradients and geodesics, and analysing convergence properties \cite{haegeman2011time,haegeman2014geometry}. In particular, optimization and sampling algorithms on Riemannian manifolds can be used in this context \cite{capel2026rapidmixinggibbsmeasures}, allowing to import powerful techniques recently developed in a rapidly growing area of research, mainly motivated by its potential applications to machine learning \cite{li2023riemannian,cheng2022efficient,criscitiello2019efficiently,sun2019escaping}.

A fundamental feature of tensor networks is their gauge freedom: distinct choices of local tensors can lead to the same global tensor. The characterization of this gauge for a given family of tensor networks is known as the {\it fundamental theorem} for such family. As the name suggests, these results play a very important role both in the mathematical understanding of tensor network states, in their use to understand and classify quantum phases of matter \cite{cirac2021matrix}, and in the design of efficient numerical algorithms. Fundamental theorems are well understood in the case of MPS \cite{PerezGarcia2007} and have also been developed for some classes of PEPS (e.g. \cite{molnar2018normal,perez2010characterizing}).

The main goal of this work is to identify the gauge freedom of several tensor network families, with the aim of describing it as a nice group action on the manifold defined by their parameter space. We show that this group action induces a projection map from the original manifold to the associated orbit space, with the property of being a \textit{Riemannian submersion}. The points in the orbit space will then be in one-to-one correspondence with the state represented by a tensor network. In this sense, our work can be understood as a generalization of the geometric framework introduced in \cite{haegeman2014geometry}, from the Riemannian manifold perspective. 

It is important to remark that having a Riemannian submersion, as captured by the definition of {\it Quotient Tensor Network Manifold} (\cref{def:quotTN} below), is the key property to apply the Riemannian sampling and optimization algorithms developed in \cite{capel2026rapidmixinggibbsmeasures}. It is therefore crucial not only to characterize the gauge freedom of a given tensor network, but to do it via a nice group action which ensures a Riemannian submersion. This justifies the name ``Riemannian fundamental theorem". 

We study the following tensor network families:
\begin{itemize}
    \item depth-two quantum circuits in one or two dimensions, both with a fixed input or not,
    \item matrix product states,
    \item two-dimensional sequentially generated PEPS \cite{banuls2008sequentially},
    \item isometric PEPS \cite{wei2022isometric}.
\end{itemize}
For each of these families, we obtain a \emph{Riemannian fundamental theorem}, relating the geometry of the manifold structure to the gauge freedom of the tensor network representation.

We consider particularly relevant the cases of sequentially generated and isometric PEPS, two paradigmatic subfamilies of PEPS. Our results open the possibility to apply Riemannian optimization techniques with rigorous convergence guarantees \cite{capel2026rapidmixinggibbsmeasures} to those families.

Before stating the main results of this work, let us define \textit{tensor networks}, \textit{tensor network manifolds} and \textit{quotient tensor network manifolds}. To do so, we start by defining a  \textit{network}, which can be understood as a slight generalisation of a directed multigraph (See \cref{fig:graph-example}).
\begin{definition}[Network]
A \textup{network} is a directed multigraph $\mathcal{G} = (\bb{V}, \bb{E})$, where $\bb{V} = \{v_i\}_{i = 1}^{|\bb{V}|}$ is the set of (labelled) vertices and $\bb{E}$ is the multiset of edges of the graph. Unlike in usual directed multigraphs, the edges of a network may only have an inbound or outbound vertex, resulting in \textup{half edges}. 
\end{definition}

\begin{figure}
    \centering
    \input{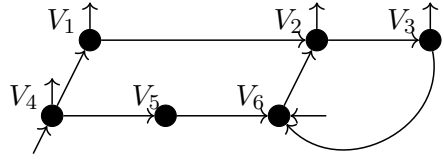}
    \caption{Visual representation of a network.}
    \label{fig:graph-example}
\end{figure}

\begin{notation}
Let $\mathcal{G}$ be a network. For every $v \in \bb{V}$ we denote by $v \leftarrow e$ and $v \rightarrow e$ the set of inbound and outbound edges of $v$, respectively. We also denote by $\bb{E}_c$ the set of \textit{connected} edges of the graph, i.e. the set of edges such that they have both an inbound and an outbound vertex.
\end{notation}

In particular, the sets of inbound and outbound edges of a given vertex of the network can contain both half edges and connected edges. 

Networks provide a graphical notation that enables us to visually represent mathematical objects and operations between them in a simple and intuitive way. In this context, we will use vertices to represent \textit{tensors} and edges to denote \textit{tensor contractions}. 

\begin{definition}[Tensor product and tensors]
Let $V_1, \dotsc, V_k$ be finite-dimensional complex vector spaces. We define their \textup{tensor product} as the set of multilinear maps defined on the product of their duals,
\[
V_1 \otimes \dotsm \otimes V_k := \{T : V_1^* \times \dotsm \times V_k^* \to \bb{C}: T \textup{ is multilinear}\}.
\]  
Every element $T \in V_1 \otimes \dotsm \otimes V_k$ is known as a \textup{tensor}. 
\end{definition}

Given two tensors, we can define their \textit{tensor product} as follows: 
\begin{definition}[Tensor product of two tensors]
Let $T \in V_1 \otimes \dotsm \otimes V_n$ and $\tilde T \in \tilde{V}_1 \otimes \dotsm \otimes \tilde V_m$ be two tensors, their \textup{tensor product} is defined as the element 
\[
T \otimes \tilde T \in V_1 \otimes \dotsm \otimes V_n \otimes \tilde{V}_1 \otimes \dotsm \otimes \tilde V_m
\]
acting as
\[
(T \otimes \tilde T)(v_1, \dotsc, v_n, \tilde v_1, \dotsc, \tilde v_m) := T(v_1, \dotsc, v_n) \tilde T(\tilde v_1, \dotsc, \tilde v_m),
\]
for every $(v_1, \dotsc, v_n) \in V_1^* \times \dotsm \times V_n^*$ and every $(\tilde v_1, \dotsc, \tilde v_m) \in \tilde V_1^* \times \dotsm \times \tilde V_m^*$. 
\end{definition}

Furthermore, there is a natural map defined on the tensor product of a vector space and its adjoint, known as the \textit{contraction map}:
\begin{definition}[Contraction map]
Given a finite-dimensional vector space $V$ and its adjoint space $V^*$, consider a basis for $V$, $\{e_i\}_{i = 1}^n$, and its dual basis for $V^*$, $\{\phi^i\}_{i = 1}^n$. The \textup{contraction map} associated with $V$ and $V^*$ is defined as the bilinear map $\theta : V \otimes V^* \to \bb{C}$, such that 
\[
\theta(e_i \otimes \phi^j) =  \delta_i^j,
\]
for every $i, j \in \{1, \dotsc, n\}$, where $\delta_i^j$ is the Kronecker delta. This way, $\theta$ is defined on $V \otimes V^*$ as the linear extension of the above map. 
\end{definition}

With the above definitions in mind, let us present the correspondence between tensors and networks in detail. In order to associate a tensor to each vertex of the network, we start by associating a finite-dimensional complex vector space to each edge of the network, $\{V_e\}_{e \in \bb{E}}$. This way, to every $v \in \bb{V}$ we can associate a tensor $T_v$ in the tensor product of the vector spaces corresponding to its inbound and outbound edges, 
\begin{equation}
\label{eq:tensorvertex}
T_{v} \in \bigotimes_{v \rightarrow e} V_e \otimes \bigotimes_{v \leftarrow e} V_e^*.     
\end{equation}
Lastly, as we mentioned earlier, the connected edges of the network are used to denote tensor contractions. In particular, an edge $e$ joining the vertices $v$ and $w$ indicates that the tensors $T_v$ and $T_w$ are contracted over their shared vector space $V_e$. 

To illustrate the correspondence between vertices and tensors, as well as between edges and contractions, we provide the next two examples.

\begin{example}
Let $v \in \bb{C}^n$ be a vector---which can be understood as a linear map $(\bb{C}^n)^* \to \bb{C}$---and let $A \in \bb{C}^m \otimes (\bb{C}^n)^*$ be a matrix. Consider their product $Av$, which can be written as the contraction (in Einstein notation)
\[
Av = A^i_j v_i \in \bb{C}^m.
\]
We can represent the vector $Av$ using a network with two vertices and two edges as follows:
\[
\input{TikzFigures/example-contraction1.tikz}.
\]

Let us now consider three unitary matrices $U_1, U_2, U_3 \in \textup{U}(n^2)$, where $n > 1$,  and let $\mathds{1}$ denote the identity matrix of dimension $n$. In particular, note that $\mathds{1} \in \bb{C}^n \otimes (\bb{C}^n)^*$. Seeing each unitary $U_i$ as an element of $\bb{C}^n \otimes \bb{C}^n \otimes (\bb{C}^n)^*\otimes (\bb{C}^n)^*$ we can represent the matrix multiplication
\[
(\mathds{1} \otimes U_3 \otimes \mathds{1}) (U_1 \otimes U_2)
\]
using a network with three vertices and ten edges:
\[
\input{TikzFigures/example-contraction2.tikz}.
\]
\end{example}

Note that in the last example, we implicitly assumed that two neighbouring vertices with no edges linking them represent the tensor product of their associated tensors. This will be made more precise in the following definition. Let us define \textit{tensor networks} and \textit{families of tensor networks}.

\begin{definition}[Tensor network]
\label{def:TN}
Let $\mathcal{G} = (\bb{V}, \bb{E})$ be a fixed network and let $\{V_e\}_{e \in \bb{E}}$ be a fixed set of finite-dimensional complex vector spaces. A \textup{tensor network} is a triple $(\mathcal{G}, \{V_e\}_{e \in \bb{E}}, \{T_v\}_{v \in \bb{V}})$, where $\{T_v\}_{v \in \bb{V}}$ is the set of tensors that are in correspondence to the vertices of $\mathcal{G}$ (cf. \cref{eq:tensorvertex}).
\end{definition}

We often denote a tensor network simply by the pair $(\mathcal{G}, \{T_v\}_{v \in \bb{V}})$ whenever the vector space associated with each edge can be inferred from the context. 

Taking into account the contractions between the tensors represented by the edges of the network, each tensor network $(\mathcal{G}, \{V_e\}_{e \in \bb{E}}, \{T_v\}_{v \in \bb{V}})$ corresponds to the tensor
\[
\bigotimes_{e \in \bb{E}_c} \theta_e \Big[\bigotimes_{v \in \bb{V}} T_v\Big],    
\]
where $\theta_e$ denotes the contraction associated with $V_e$ and $V_e^*$---note that, since $e \in \bb{E}_c$, $e$ has an inbound and an outbound vertex in $\mathcal{G}$. 

Although networks can be understood as a visual (local) representation of a single (global) tensor via a tensor network,  we can also use them to represent \textit{families} of tensors. 

\begin{definition}[Tensor network family]
Let $\mathcal{G} = (\bb{V}, \bb{E})$ be a network. Consider a fixed set of vector spaces assigned to the edges $\{V_e\}_{e \in \bb{E}}$ and a set $\mathcal{A}$,
\[
\mathcal{A} \subset \prod_{v \in \bb{V}}\Big( \bigotimes_{v \rightarrow e} V_e \otimes \bigotimes_{v \leftarrow e} V_e^*\Big),
\]
representing the \textit{allowed} choices of tensors associated with each vertex of the network. We define the \textup{family of tensor networks with network $\mathcal{G}$ and set of allowed tensors $\mathcal{A}$} as the set of all possible tensor networks for which the tensors lay in $\mathcal{A}$, i.e. 
\[
\{(\mathcal{G}, \mathcal{T})\}_{\mathcal{T} \in \mathcal{A}}. 
\]
\end{definition}

Let us illustrate the above definition with two examples. 

\begin{example}
\label{exampleTNfamilies}
Probably one of the most well-known tensor network families are matrix product states (MPS). Let us consider a family of uniform MPSs with five sites and periodic boundary conditions, which can be represented by the network
\[
\input{TikzFigures/MPS.tikz}.
\]
Assuming that the dimension of the vector spaces associated with the contracting edges---also known as the bond dimension of the MPS---is $D$, and that the dimension of the vector spaces associated with the half-edges---also known as the physical dimension---is $d$, the set of allowed tensors is
\[
\mathcal{A} = \{(T_1, \dotsc, T_5) : T_i \in \bb{C}^d \otimes \bb{C}^D \otimes (\bb{C}^D)^* \text{, and } T_1 = T_2 = \dotsc = T_5\}.
\]
From the definition of $\mathcal{A}$, it is clear that the same tensor is assigned to every vertex of the network, resulting in a uniform MPS.

Next, consider the tensor network family of one-dimensional depth-two quantum circuits, with gates acting on two qudits, which can be represented by the network
\[
\input{TikzFigures/circuitosimple2.tikz}.
\]
In this case, as any quantum gate is given by a unitary matrix, the set of allowed tensors is
\[
\mathcal{A} = \{(T_1, \dotsc, T_{2k-1}) : T_i \in \textup{U}(d^2) \subset \bb{C}^d \otimes \bb{C}^d \otimes (\bb{C}^d)^*\otimes (\bb{C}^d)^*\},
\]
where $k$ is the number of gates in the bottom layer of the circuit, and each qudit is assumed to be a unit vector in $\bb{C}^d$. 
\end{example}

In the two tensor network families shown above, the set $\mathcal{A}$ is such that it can be endowed with a smooth manifold structure---and thus with a Riemannian metric. Whenever this is possible, we say that the tensor network family defines a \textit{tensor network manifold}. 

\begin{definition}[Tensor network manifolds]
A tensor network family defines a \textup{tensor network manifold} if the set $\mathcal{A}$ can be endowed with a smooth manifold structure.
\end{definition}

\begin{remark}
Note that, despite working with tensors acting on the product of complex vector spaces, the manifold structure considered in a tensor network manifold is that of real manifolds. This perspective is different to the one shown in \cite{haegeman2014geometry}, where the authors study tensor networks as complex manifolds. Moreover, since we only consider real manifolds in this work, the assumption that the tensors act on the product of complex vector spaces can be easily modified to consider real vector spaces.     
\end{remark}

As we mentioned earlier, the tensor network families shown in \cref{exampleTNfamilies} define tensor network manifolds. Indeed, the set of allowed tensors in the uniform MPS family described can be seen as the---real---manifold $\bb{C}^{d \times D^2}$, where $d$ is the physical dimension of the MPS and $D$ is its bond dimension. Moreover, the set of allowed tensors in the family of one-dimensional, depth-two quantum circuits corresponds to the manifold $\textup{U}(d^2)^{\times 2k-1}$, where each qudit is assumed to be a unit vector of $\bb{C}^d$ and $k$ is the number of gates in the bottom layer of the circuit.

Given a tensor network family defining a tensor network manifold $M$, there is a natural correspondence between the points in the manifold and the tensor networks in the family. Indeed, every $x \in M$ can be understood as a particular choice of tensors for every vertex of the network, 
\begin{equation}
\label{eq:definitionpointx}
x \equiv (T_1, \dotsc, T_{|\bb{V}|}).
\end{equation}

\begin{notation}
Consider a tensor network family $\{(\mathcal{G}, \mathcal{T})\}_{\mathcal{T} \in \mathcal{A}}$ defining a tensor network manifold $M$. For every $x \in M \simeq \mathcal{A}$, we denote by $\mathds{TN}_x$ the global tensor which corresponds to the choice of local tensors represented by $x$. 
\end{notation}

In particular, if $x \in M$ is given as in \cref{eq:definitionpointx}, $\mathds{TN}_x$ is defined as
\[
\mathds{TN}_x = 
\bigotimes_{e \in \bb{E}_c} \theta_e \Big[\bigotimes_{v \in \bb{V}} T_v\Big].
\]

\begin{remark}
\label{remarkgauge}
As we mentioned in the introduction, the gauge freedom of tensor networks implies that the correspondence between points in a tensor network manifold and the tensor represented by a tensor network may not be one-to-one, as two distinct points in a given tensor network manifold $x, y \in M$ can result in the same tensor, i.e. $\mathds{TN}_x = \mathds{TN}_y$. 
\end{remark}

To illustrate the above remark, let us consider the following example. 

\begin{example}
\label{eg:gaugecircuits}
Consider the tensor network family of a one-dimensional, depth-two circuit with three gates, 
\[
\input{TikzFigures/smallcircuit.tikz}, 
\]
for which its associated tensor network manifold is $\textup{U}(d^2)^{\times 3}$. Let
\[
x := (U, V, W) \in \textup{U}(d^2)^{\times 3},
\]
be some fixed point, where $U$ and $V$ correspond to the bottom left and bottom right unitaries of the circuit, respectively, and $W$ corresponds to the unitary in the top row. Define
\[
y := ((\mathds{1} \otimes X)U, (Y \otimes \mathds{1})V, W(X^\dagger \otimes Y^\dagger)) \in \textup{U}(d^2)^{\times 3},
\]
where $X, Y \in \textup{U}(d)$ are two fixed unitaries and $\mathds{1}$ is the identity matrix of size $d$. Clearly, if either $X$ or $Y$ are different from the identity matrix, $x$ and $y$ are different points in $\textup{U}(d^2)^{\times 3}$. Nevertheless, it holds that $\mathds{TN}_x = \mathds{TN}_y$. Indeed, $\mathds{TN}_x$ can be represented as
\[
\input{TikzFigures/smallcircuitx.tikz}
\]
and $\mathds{TN}_y$ can be represented as
\[
\input{TikzFigures/smallcircuity.tikz},
\]
showing that $\mathds{TN}_x = \mathds{TN}_y$.
\end{example}

In conclusion, \cref{remarkgauge} implies that the map $x \mapsto \mathds{TN}_x$ is in general not injective. In fact, symmetries such as the one shown in the previous remark are inherent to tensor networks, as tensor contractions introduce a gauge freedom. This shows the importance of having fundamental theorems for tensor networks, i.e. having explicit characterisations of the gauge symmetries of general tensor networks. 

Consider a tensor network family defining a tensor network manifold, and assume that some gauge freedom of the family can be expressed as an equivalence relation in the manifold, which can be \textit{quotiented out}. In this case, the resulting quotient space induces a \textit{less redundant} mapping from its points to the tensor networks of the family. Indeed, each point in the quotient space will correspond to the \textit{equivalence class} of a given tensor network, modulo the gauge considered. Motivated by this, we define \textit{quotient tensor network manifolds}.

\begin{definition}[Quotient tensor network manifold]
\label{def:quotTN}
Let $(M, g)$ be a tensor network manifold endowed with a Riemannian metric. We say that $(M, g)$ defines a \textup{quotient tensor network manifold} if there exists a Riemannian manifold $(B, h)$ and a surjective Riemannian submersion $\pi$ from $(M, g)$ to $(B, h)$.
\end{definition}

\begin{notation}
Given a tensor network manifold $(M, g)$ defining a quotient tensor network manifold $(B, h)$ with associated Riemannian submersion $\pi$, for every point $x \in M$, we will often use $[x]_B$ to denote $\pi(x)$. 
\end{notation}

It is often the case that the set of points in a given tensor network manifold $M$ which give rise to the same tensor can be related by a \textit{group action} on $M$. Group actions on manifolds are very useful, as they frequently allow to consider the quotient spaces they induce as manifolds (see \cref{sec:riemanniansubmersion}).

From the above definition it follows that a tensor network manifold can give rise to several different quotient tensor network manifolds. Indeed, in the following example we will define three quotient tensor network manifolds associated with the same tensor network manifold.

\begin{example}
\label{examplebeforethm}
Consider the tensor network family of depth-two quantum circuits with three gates,
\[
\input{TikzFigures/smallcircuit.tikz}.
\]
Recall that its associated tensor network manifold is 
\[
\textup{U}(d^2)^{\times 3}.
\]
Assume that we are interested in obtaining a manifold for which each point allows us to describe a circuit of the family up to a global phase. To do so, we consider the action of $\textup{U}(1)^{\times 3}$ on $\textup{U}(d^2)^{\times 3}$ that describes the multiplication of any of the gates of the circuit by a complex phase, 
\begin{align*}
\textup{U}(1)^{\times 3} \times \textup{U}(d^2)^{\times 3} &\to \textup{U}(d^2)^{\times 3}\\
((\lambda_1, \lambda_2, \lambda_3),(U, V, W)) &\mapsto (\lambda_1U, \lambda_2V, \lambda_3W).
\end{align*}
In this case, the quotient space $B_1 := \textup{U}(d^2)^{\times 3} / \textup{U}(1)^{\times 3}$---which corresponds to $\textup{PU}(d^2)^{\times 3}$, where $\textup{PU}$ is the projective unitary group (see \cref{sec:SecLieGroups})---is such that every point in $B_1$ represents a quantum circuit up to a global phase.

We also showed in \cref{eg:gaugecircuits} that multiplying over the edges of any tensor network of this family by unitaries and their inverses had no effect on the circuit it represents. With this in mind, we can consider the action of $\textup{U}(d)^{\times 2}$ defined as
\begin{align}
\label{eqgaugenonproj}
\begin{split}
\textup{U}(d)^{\times 2} \times \textup{U}(d^2)^{\times 3} &\to \textup{U}(d^2)^{\times 3}\\
((X, Y), (U, V, W)) &\mapsto ((\mathds{1} \otimes X)U, (Y \otimes \mathds{1})V, W(X^\dagger \otimes Y^\dagger)),
\end{split}
\end{align}
where $U$ and $V$ denote the bottom-left and bottom-right gates of the circuit, respectively, and $W$ denotes the top gate of the circuit. This action can be represented as
\[
\input{TikzFigures/circuitex1.tikz},
\]
and encodes the gauge freedom shown in \cref{eg:gaugecircuits}. This way, the quotient space $B_2 := \textup{U}(d^2)^{\times 3} / \textup{U}(d)^{\times 2}$ is such that every point $x \in B_2$ corresponds to the equivalence class $x = [(U, V, W)]_{B_2}$, with $(U, V, W) \in \textup{U}(d^2)^{\times 3}$ where 
\[
[(U, V, W)]_{B_2} := \{((\mathds{1} \otimes X)U, (Y \otimes \mathds{1})V, W(X^\dagger \otimes Y^\dagger))\, :\, X, Y \in \textup{U}(d)\},
\]
and so every $x \in B_2$ represents a tensor network of the family \textit{up to} the gauge freedom identified in \cref{eg:gaugecircuits}. 

Lastly, going back to the quotient tensor network manifold $B_1 = \textup{PU}(d^2)^{\times 3}$, which results after taking the quotient with respect to multiplying each gate by a complex phase, we can also consider the analogous of the action defined in \cref{eqgaugenonproj}. This way, we define the action of $\textup{PU}(d)^{\times 2}$ on $\textup{PU}(d^2)^{\times 3}$ as
\begin{align*}
\textup{PU}(d)^{\times 2} \times \textup{PU}(d^2)^{\times 3} &\to \textup{PU}(d^2)^{\times 3}\\
(([X], [Y]), ([U], [V], [W])) &\mapsto ([(\mathds{1} \otimes X)U], [(Y \otimes \mathds{1})V], [W(X^\dagger \otimes Y^\dagger])).
\end{align*}
This action induces the quotient space $B_3 := \textup{PU}(d^2)^{\times 3} / \textup{PU}(d)^{\times 2}$. If we consider the composition of the quotient maps $\textup{U}(d^2)^{\times 3} \to B_1$ and $B_1 \to B_3$, we can see $B_3$ as a quotient tensor network manifold which uniquely characterises each circuit up to both unitary multiplication over the contracting edges, and up to a phase. In this case, the quotient tensor network manifold $B_3$ arises naturally as a composition of two projections, rather than as the quotient space associated with a single group action on the original tensor network manifold $\textup{U}(d^2)^{\times 3}$. See \cref{sectionnofixedinput} for more details. 
\end{example}

\subsection{Main results}
\label{sec:mainresults}

As we mentioned in the introduction, the main goal of this work is to obtain quotient tensor network manifolds which uniquely characterise \textit{almost every} tensor represented by a tensor network, for several tensor network families. We will consider tensors (given by tensor networks) which represent either quantum states or quantum circuits. In both cases, it is common to consider equality up to a global phase. In the case of states since they cannot be distinguished by an observable. In the case of unitaries since global phases disappear under the adjoint action. For this reason, we obtain quotient tensor network manifolds that characterise states and circuits both uniquely and up to a phase.

\begin{theorem}[Riemannian fundamental theorem]
\label{mainthm2}
The following tensor network families define a tensor network manifold: 
\begin{itemize}
    \item depth-two quantum circuits in one or two dimensions, both with a fixed input or not, 
    \item matrix product states,
    \item two-dimensional sequentially generated PEPS,
    \item isometric PEPS.
\end{itemize}
Let $M$ be the tensor network manifold associated with any of the above tensor network families. There exists a suitable Riemannian metric $g$ on $M$ so that $(M, g)$ defines a quotient tensor network manifold $(B, h)$ such that, for every $x, y$ on a open and dense subset of $M$, it holds that
\[
\exists \lambda \in \textup{U}(1) \text{ s.t. } \mathds{TN}_x = \lambda \mathds{TN}_y \iff [x]_B = [y]_B.
\]

Furthermore, let $(M, g)$ the tensor network manifold associated with the tensor network family of depth-two quantum circuits in one dimension, either with a fixed input or not, or matrix product states. It defines a quotient tensor network manifold $(B', h')$ such that, for every $x, y$ on a open and dense subset of $M$, it holds that
\[
\mathds{TN}_x = \mathds{TN}_y \iff [x]_{B'} = [y]_{B'}.
\]
\end{theorem}

In the following sections, we will prove this result for each tensor network family considered. The proof strategy will be the same in all cases. First we will identify the gauge freedom. Then, we will ensure that the group action representing the gauge freedom induces an orbit space admitting a smooth manifold structure. For that, we will show that the action is smooth, free and isometric. These properties also guarantee that the projection map associated with the quotient is a Riemannian submersion. See \cref{sec:riemanniansubmersion} for a short introduction to Riemannian submersions, and \cref{sec:SecLieGroups,sec:manifolds} for a brief introduction on the manifolds considered in this work. 

%% file: TikzFigures/example-contraction1.tikz
\begin{tikzpicture}[scale=2]
    \draw (0.141, 0.282) rectangle (0.423, 0);
    \draw (0, 1.27) rectangle (0.564, 0.706);
    \draw [->] (0.282, 0.282) -- (0.282, 0.706);
    \draw [->] (0.282, 1.27) -- (0.282, 1.693);
    \node[anchor=center] at (0.282, 0.988) {$A$};
    \node[anchor=center] at (0.282, 0.141) {$v$};
\end{tikzpicture}

%% file: Chapters/TNManifolds.tex
\section{Depth-two quantum circuits in one dimension with a fixed input}
\label{onedimcircuitsfixed}

Let us first study the tensor network family of depth-two quantum circuits with a fixed input in one dimension, with $k$ gates in the bottom layer. We assume that each gate acts on two neighbouring qudits, which are assumed to be unit vectors in $\bb{C}^{d}$. These circuits can be represented using the network
\[
\input{TikzFigures/circuitosimple.tikz}.
\]
Without loss of generality, the input state of the circuits is assumed to be $\ket{0}^{\otimes 2k} \in (\bb{C}^{d})^{\times 2k}$.

To avoid redundancies, as the input of the circuit is fixed, we can replace the tensors that correspond to the gates in the bottom layer by their output state. This way, we represent the output state of the circuit using the following tensor network\footnote{Observe that in \cref{circuitodepth2} we have omitted the arrows to simplify the notation. Unless explicitly written otherwise, the arrows in the tensor networks that appear throughout this work will be assumed to always point upwards.},
\begin{equation}
\label{circuitodepth2}
\input{TikzFigures/fig2.tikz},
\end{equation}
where each tensor $U_i$ corresponds to a two-qudit gate in the top layer of the circuit and each tensor $\Lambda_i$ corresponds to the output state of a gate in the bottom layer. Therefore, $\Lambda_1, \dots, \Lambda_k \in \bb{S}^{2d^2-1}$, and $U_1, \dotsc, U_{k-1} \in \textup{U}(d^2)$. Using this tensor network description, the tensor network manifold associated with the tensor network family of depth-two quantum circuits with a fixed input in one dimension, with $k$ gates in the bottom layer is 
\[
M_{\textup{1Df}} := (\bb{S}^{2d^2-1})^{\times k} \times \textup{U}(d^2)^{\times k-1}.
\]
We consider $M_{\textup{1Df}}$ endowed with its \textit{natural metric} $g_{\textup{1Df}} := g_{\textup{round}}^{\oplus k} \oplus g_{\textup{bi}}^{\oplus k-1}$ (cf. \cref{sec:manifolds}).

We will obtain two quotient tensor network manifolds defined by $M_{\textup{1Df}}$ and such that they characterise the output state of the circuit. We will prove that one gives a one-to-one correspondence between points in the quotient tensor network manifold and the state represented by a tensor network of the family, and the other gives a one-to-one correspondence up to a complex phase.
\begin{theorem}
\label{thm:quotTN1DFI}
Let $M_{\textup{1Df}}$ be the tensor network manifold associated with the tensor network family of one-dimensional, depth-two circuits with a fixed input and $k$ gates in the bottom layer. When endowed with its natural metric $g_{\textup{1Df}}$, it defines two quotient tensor network manifolds, $(B_{\textup{1Df}}, h_{\textup{1Df}})$ and $(B'_{\textup{1Df}}, h'_{\textup{1Df}})$ such that, for every $x, y$ in a open and dense subset of $M_{\mathit{1df}}$, it holds that
\[
\mathds{TN}_x = \mathds{TN}_y \iff [x]_{B_{\textup{1Df}}} = [y]_{B_{\textup{1Df}}},
\]
and
\[
\exists \lambda \in \textup{U}(1) \text{ s.t. } \mathds{TN}_x = \lambda \mathds{TN}_y \iff [x]_{B'_{\textup{1Df}}} = [y]_{B'_{\textup{1Df}}}.
\]
\end{theorem}

\subsection{Characterising the gauge}

Let us first identify the gauge freedom of depth-two quantum circuits in one dimension with a fixed input, both exactly and up to a phase. 
\begin{theorem}
\label{theoremgauge1}
Consider the tensor network family of depth-two quantum circuits in one dimension with $k$ gates in the bottom layer and a fixed input, as shown in \cref{circuitodepth2}. Consider two tensor networks in this family representing the states $\ket{U}$ and $\ket{V}$,
\[
\ket{U} \equiv (\Lambda_1, \dotsc, \Lambda_k, U_1, \dotsc, U_{k-1}),\quad \ket{V} \equiv (\Sigma_1, \dotsc, \Sigma_k, V_1, \dotsc, V_{k-1}).
\]
Assume that the tensors $\Lambda_1, \dotsc, \Lambda_k$ are full-rank when seen as linear maps from $\bb{C}^{d}$ to $\bb{C}^d$, 
\[
\begin{tikzpicture}[scale=1.5]
    \filldraw[fill=white] 
    (0.706, 0.141) ellipse[x radius=-0.141, y radius=0.141];
    \draw[->] (0, 0.141) -- (0.564, 0.141);
    \draw[->] (0.847, 0.141) -- (1.411, 0.141);
\end{tikzpicture}.
\]
Assume that there exists some phase $\theta \in [0, 2\pi)$ for which $\ket{U} = e^{i\theta} \ket{V}$. Then, for every $j \in \{1, \dots, k-1\}$ there exist some unitary matrices $X_j, Y_j \in \textup{U}(d)$ such that
\[
V_j \propto U_j (X^\dagger_j \otimes Y^\dagger_j)
\]
and 
\begin{align*}
\Sigma_1 &\propto (\mathds{1} \otimes X_1) \Lambda_1,\\
\Sigma_j &\propto (Y_{j-1} \otimes X_j) \Lambda_j,\ \forall j \in \{2, \dotsc, k-1\},\\
\Sigma_{k} &\propto (Y_{k-1} \otimes \mathds{1})\Lambda_k,    
\end{align*}
where $\propto$ denotes equality up to a complex phase. Furthermore, if $\theta = 0$, the above equations hold with equality. 
\end{theorem}

\begin{remark}
\label{rmk:generic1}
The condition of the bottom states of the tensor network associated with $\ket{U}$ being \textit{full-rank} is verified in an open and dense set of the tensor network manifold.
\end{remark}

Before proving \cref{theoremgauge1}, we state an auxiliary proposition relating unitary matrices and the tensor product. 
\begin{lemma}[{\cite[Theorem 23]{broxson2006kronecker}}]
\label{productounitarias}
Let $U$, $V$ be two matrices for which $U \otimes V$ is unitary. Then, there exists some $\mu > 0$ such that $\mu U$ and $\frac{1}{\mu}V$ are unitary.
\end{lemma}

\begin{proof}[Proof of \cref{theoremgauge1}]
Consider that the same state, up to a global phase $\theta$, is represented by two tensor networks of the family, 
\begin{equation}
\label{eq:firsteq1dcircuit}
\input{TikzFigures/fig3.tikz},
\end{equation}
where by assumption $\Lambda_1, \dots, \Lambda_k$ are full-rank. We used $\propto$ to denote that both states are equal up to a phase. 
 
Let us fix one site $i \in \{1, \dotsc, k-1\}$ in the above equality. Multiplying site $i$ by $V_i^\dagger$ and every site $j \neq i$ by $U_j^\dagger$, we obtain
\begin{equation*}
\input{TikzFigures/fig4.tikz},
\end{equation*}
where the vertical red dotted line highlights the $i$-th site in both tensor networks. Since the states $\Lambda_1, \dotsc, \Lambda_k$ are full-rank by assumption, we can invert them to obtain
\begin{equation*}
\input{TikzFigures/fig5.tikz}.
\end{equation*}
Tracing at both sides of the $i$-th site, we see that
\begin{equation}
\label{eq:equalitiestrace}
\input{TikzFigures/fig6.tikz}.
\end{equation}
If we now define
\[
\input{TikzFigures/defunitaries.tikz}
\]
we can rewrite \cref{eq:equalitiestrace} as
\begin{equation*}
\input{TikzFigures/fig7.tikz},
\end{equation*}
where $X_i$ and $Y_i$ can be assumed to be unitary without loss of generality, as their tensor product is unitary (cf. \cref{productounitarias}).

Following the same reasoning for every site, we can conclude that for every $j \in \{1, \dotsc, k-1\}$ there exists a phase $\varphi_j \in [0, 2\pi)$ and two unitary matrices $X_j, Y_j \in \textup{U}(d)$ such that 
\begin{equation}
\label{sym1}
U_j = e^{i\varphi_j} V_j (X_j \otimes Y_j).
\end{equation}
This allows us to rewrite \cref{eq:firsteq1dcircuit} as
\begin{equation*}
\input{TikzFigures/fig8.tikz},
\end{equation*}
and so we can conclude that there exist some phases $\phi_1, \dotsc, \phi_k \in [0, 2\pi)$, such that
\begin{align}
\begin{split}
\label{sym2}
\Sigma_1 &= e^{i\phi_1}(\mathds{1} \otimes X_1) \Lambda_1,\\
\Sigma_j &= e^{i\phi_j}(Y_{j-1} \otimes X_j) \Lambda_j,\ \forall j \in \{2, \dotsc, k-1\},\\
\Sigma_{k} &= e^{i\phi_k}(Y_{k-1} \otimes \mathds{1})\Lambda_k.
\end{split}
\end{align}

In the case when the global phase $\theta$ is zero, the proof can be adapted by \textit{absorbing} the phases from \cref{sym2} into the unitary matrices $X_j$ and $Y_j$. Indeed, note that whenever $\theta = 0$, the phases $\varphi_j$ from \cref{sym1} are all zero. This way, we can conclude that \cref{sym2} holds for some phases $\phi_i$ such that $\sum_{j= 1}^k \phi_j \in  2\pi \bb{Z}$. Thus, we can define 
\begin{align*}
\tilde X_j := e^{i\sum_{m = 1}^j \phi_m} X_j \in \textup{U}(d),\\
\tilde Y_j := e^{-i \sum_{m = 1}^j \phi_m} Y_j \in \textup{U}(d),
\end{align*}
for every $j \in \{1, \dotsc, k-1\}$ to conclude that 
\[
U_j = V_j (\tilde X_j \otimes \tilde Y_j) \big(= V_j (X_j \otimes Y_j)\big),
\]
for every $j \in \{1, \dots, k-1\}$ and 
\begin{align*}
\Sigma_1 &= (\mathds{1} \otimes \tilde X_1) \Lambda_1,\\
\Sigma_j &= (\tilde Y_{j-1} \otimes \tilde X_j) \Lambda_j,\ \forall j \in \{2, \dotsc, k-1\},\\
\Sigma_{k} &= (\tilde Y_{k-1} \otimes \mathds{1})\Lambda_k,    
\end{align*}
finishing the proof. 
\end{proof}

\subsection{Quotient tensor network manifold structure}
\label{sec:1DFIasquot}

Having described the gauge freedom of one-dimensional, depth-two quantum circuits with a fixed input, both up to a phase and exactly, let us now turn our attention to the group action description of such a gauge. Recall that the tensor network manifold associated with this tensor network family is
\[
M_{\textup{1Df}} = (\bb{S}^{2d^2-1})^{\times k} \times \textup{U}(d^2)^{\times k-1}.
\]
\cref{theoremgauge1} clearly motivates the definition of the following group action of $\textup{U}(d)^{\times 2(k-1)}$ on $M_{\textup{1Df}}$ describing the gauge freedom. 

\begin{definition}[Gauge action for 1D circuits with fixed input]
\label{def:gauge1DFIP}
We define the action $\alpha_1$ of $\textup{U}(d)^{\times 2(k-1)}$ on $M_{\textup{1Df}}$ as
\[
\alpha_1 : \textup{U}(d)^{\times 2(k-1)} \times M_{\textup{1Df}} \to M_{\textup{1Df}},\quad (X, (\Lambda, V)) \mapsto (X\Lambda, VX^\dagger),
\]
where 
\begin{align*}
X &:= (X_1, \dotsc, X_{2(k-1)}) \in \textup{U}(d)^{2(k-1)},\\
\Lambda &:= (\Lambda_1, \dotsc, \Lambda_k) \in (\bb{S}^{2d^2-1})^{\times k},\\
V &:= (V_1, \dotsc, V_{k-1}) \in \textup{U}(d^2)^{\times k-1},
\end{align*}
and 
\begin{align*}
X\Lambda &:= ((\mathds{1} \otimes X_1) \Lambda_1, (X_2 \otimes X_3) \Lambda_2, \dotsc, (X_{2(k-2)} \otimes X_{2k-3}) \Lambda_{k-1} , (X_{2(k-1)} \otimes \mathds{1}) \Lambda_k),\\
VX^\dagger &:= (V_1 (X^\dagger_1 \otimes X^\dagger_2), \dotsc, V_{k-1}(X^\dagger_{2k-3} \otimes X^\dagger_{2(k-1)})).
\end{align*}
\end{definition}

Considering the quotient of $M_{\textup{1Df}}$ with respect to $\alpha_1$ we will obtain a quotient tensor network manifold which gives a one-to-one correspondence between points in the quotient manifold and the state represented by a tensor network of the family. 

Going back to the statement of \cref{theoremgauge1}, recall that we also studied the case when two tensor networks
\[
x = (\Lambda_1, \dotsc, \Lambda_{k}, U_1, \dotsc, U_{k-1}),\quad 
y = (\Sigma_1, \dotsc, \Sigma_{k}, V_1, \dotsc, V_{k-1}),
\]
represent the same state up to a global phase, $\mathds{TN}_x = e^{i\theta}\mathds{TN}_y$ with $\theta \in (0, 2\pi)$. In this case, the gauge is identified only up to a phase. For this reason, the result obtained can be written naturally using projective spaces. Indeed, as $\bb{S}^{2d-1} / \textup{U}(1) \simeq \bb{CP}^{d-1}$ and $\textup{U}(d) / \textup{U}(1) \simeq \textup{PU}(d)$ (cf. \cref{sec:manifolds}), if we consider
\[
N_{\textup{1Df}} := (\bb{CP}^{d^2-1})^{\times k} \times \textup{PU}(d^2)^{\times k-1},
\]
each point in this manifold represents a tensor network in the tensor network family of one-dimensional, depth-two quantum circuit with a fixed input up to a phase. Rewriting the result obtained in \cref{theoremgauge1} using this picture, we showed that for every $x, y$ in an open and dense subset of $M_{\textup{1Df}}$ if $\mathds{TN}_x = e^{i\theta}\mathds{TN}_y$ for some phase $\theta \in (0, 2\pi)$, then there exist some unitary matrices $X_j, Y_j \in \textup{U}(d)$, $j \in \{1, \dotsc, k-1\}$ such that
\[
[V_j] = [U_j (X^{\dagger}_j \otimes Y^\dagger_j)] \in \textup{PU}(d^2),
\]
for every $j \in \{1, \dotsc, k-1\}$ and
\begin{align*}
\bb{CP}^{d^2-1} \ni [\Sigma_1] &= [(\mathds{1} \otimes X_1) \Lambda_1],\\ 
\bb{CP}^{d^2-1} \ni [\Sigma_i] &= [( Y_i \otimes X_i) \Lambda_i],\ \forall i \in \{2, \dotsc, k-1\},\\ \bb{CP}^{d^2-1} \ni [\Sigma_{k}] &= [(Y_{k-1} \otimes \mathds{1})\Lambda_k].
\end{align*}

As in the case of the tensor network manifold $M_{\textup{1Df}}$, we can also endow $N_{\textup{1Df}}$ with what we call as its \textit{natural metric}, which is defined as $\tilde g_{\textup{1Df}} := g_{\textup{FS}}^{\oplus k} \oplus g_{\textup{pu}}^{\oplus k-1}$ (cf. \cref{sec:manifolds}). In particular, the metric considered on $N_{\textup{1Df}}$ is such that the projection $(M, g_{\textup{1Df}}) \to (N_{\textup{1Df}}, \tilde g_{\textup{1Df}})$ is a Riemannian submersion (cf. \cref{sec:manifolds}), and so $N_{\textup{1Df}}$ is a quotient tensor network manifold. 

Motivated by the above interpretation of \cref{theoremgauge1}, we define a group action on $N_\textup{1Df}$ describing the gauge freedom.

\begin{definition}[Gauge action for 1D circuits with fixed input, up to phase]
\label{def:gauge1DFI}
We define the action $\alpha_2$ of $\textup{PU}(d)^{\times 2(k-1)}$ on $N_{\textup{1Df}}$ as
\begin{align*}
\alpha_2 :\textup{PU}(d)^{\times 2(k-1)} \times N_{\textup{1Df}} \to N_{\textup{1Df}},\quad (X, (\Lambda, V)) \mapsto (X\Lambda, VX^\dagger),
\end{align*}
where 
\begin{align*}
X &:= ([X_1], \dotsc, [X_{2(k-1)}]) \in \textup{PU}(d)^{\times 2(k-1)},\\
\Lambda &:= ([\Lambda_1], \dotsc, [\Lambda_k]) \in (\bb{CP}^{d^2-1})^{\times k},\\
V &:= ([V_1], \dotsc, [V_{k-1}]) \in \textup{PU}(d^2)^{\times k-1},
\end{align*}
and 
\begin{align*}
X\Lambda &:= ([(\mathds{1} \otimes X_1) \Lambda_1], [(X_2 \otimes X_3) \Lambda_2], \dotsc, [(X_{2(k-2)} \otimes X_{2k-3}) \Lambda_{k-1}] , [(X_{2(k-1)} \otimes \mathds{1}) \Lambda_k]),\\
VX^\dagger &:= ([V_1 (X^\dagger_1 \otimes X^\dagger_2)], \dotsc, [V_{k-1}(X^\dagger_{2k-3} \otimes X^\dagger_{2(k-1)})]).
\end{align*}
\end{definition}

The following two propositions guarantee the existence of a smooth manifold structure and a Riemannian metric on the orbit spaces that result after considering the quotients by the actions $\alpha_1$ and $\alpha_2$ given in \cref{def:gauge1DFIP,def:gauge1DFI}, which in turn represent the gauge found in \cref{theoremgauge1}. The Riemannian metrics on the quotient tensor network manifolds are such that the projection maps are Riemannian submersions. For this, we will prove that $\alpha_1$ and $\alpha_2$ are smooth, free and isometric actions, in which case the result will simply follow from \cref{existenceRiemannianSubmersionMetrics}. See \cref{fig:diag1} for a visual representation of the setting considered. 
\begin{figure}
    \centering
    \input{TikzFigures/diag1.tikz}
    \caption{Sketch of the three quotient tensor network manifolds associated with one-dimensional, depth-two quantum circuits with a fixed input. Each arrow represents a Riemannian submersion induced by a group action.}
    \label{fig:diag1}
\end{figure} 

\begin{proposition}
\label{circuitswithphase}
Let $M_{\textup{1Df}}$ be endowed with its natural metric $g_{\textup{1Df}}$. Consider the action $\alpha_1$ of $\textup{U}(d)^{\times 2(k-1)}$ on $M_{\textup{1Df}}$ as given in \cref{def:gauge1DFIP}. There is a smooth manifold structure and a Riemannian metric $h_{\textup{1Df}}$ on $B_{\textup{1Df}} := M_{\textup{1Df}} / \textup{U}(d)^{\times 2(k-1)}$ for which the projection $(M_{\textup{1Df}}, g_{\textup{1Df}}) \to (B_{\textup{1Df}}, h_{\textup{1Df}})$ is a Riemannian submersion. 
\end{proposition}
\begin{proof}
First, it is known that matrix multiplication is a smooth map, and so $\alpha_1$ is smooth. Now, as $M_{\textup{1Df}}$ is endowed with its natural metric $g_{\textup{1Df}}$, which is the product of the round metric and the bi-invariant metric, it follows from \cref{prop:isomsphere,prop:isomunitaries} that $\alpha_1$ is an isometric action, as it is defined as unitary matrix multiplication in every component. 

In order to conclude the proof, it only remains to show that the action is free. For this, let $V \in \textup{U}(d^2)^{\times k-1}$, $\Lambda \in (\bb{S}^{2d^2-1})^{\times k}$ and assume that there exists an element $X \in \textup{U}(d)^{\times 2(k-1)}$ such that
\[
VX^\dagger = V,\ \textup{and}\ X\Lambda = \Lambda,
\]
where we are using the same notation as in \cref{def:gauge1DFIP}. Then, for every $i \in \{1, \dotsc, k-1\}$ it holds that 
\[
V_i (X^\dagger_{2i-1} \otimes X^\dagger_{2i}) = V_i,
\]
which in particular implies that for every $i \in \{1, \dotsc, k-1\}$ there exists some $\lambda_i \in \textup{U}(1)$ such that 
\begin{equation}
\label{almostfree}
X_{2i-1} = \lambda_i \mathds{1},\ X_{2i} = \frac{1}{\lambda_i} \mathds{1}.
\end{equation}
Moreover, from the assumption that
\[
X\Lambda = \Lambda,
\]
we know that
\begin{align*}
(\mathds{1} \otimes X_1)\Lambda_1 &= \Lambda_1,
\\(X_{2i} \otimes X_{2i+1})\Lambda_i &= \Lambda_i,\quad \forall i \in \{2, \dotsc k-1\},\\
(X_{2(k-1)} \otimes \mathds{1})\Lambda_k &= \Lambda_k.
\end{align*}
These equalities, along with \cref{almostfree}, automatically imply that $\lambda_i = 1$ for every $i \in \{1, \dotsc, k-1\}$, and so the action is free. 
\end{proof}

\begin{proposition}
\label{thm:action1Dcircuitsfixedinputphase}
Let $N_{\textup{1Df}}$ be endowed with its natural metric $\tilde g_{\textup{1Df}}$. Consider the action $\alpha_2$ of $\textup{PU}(d)^{\times 2(k-1)}$ on $N_{\textup{1Df}}$ as given in \cref{def:gauge1DFI}. There is a smooth manifold structure and a Riemannian metric $h'_{\textup{1Df}}$ on $B'_{\textup{1Df}} := N_{\textup{1Df}} / \textup{PU}(d)^{\times 2(k-1)}$ for which the projection $(N_{\textup{1Df}}, \tilde g_{\textup{1Df}}) \to (B'_{\textup{1Df}},  h'_{\textup{1Df}})$ is a Riemannian submersion. 
\end{proposition}
\begin{proof}
The action is smooth. Furthermore, it is isometric. This follows automatically from \cref{projectiveIsometry} as $\alpha_2$ corresponds to the \textit{projected version} of $\alpha_1$, which was shown to be an isometric action in the previous proof. 

The freeness of the action follows automatically from the previous proof. Indeed, if we assume that for every $i \in \{1, \dotsc, k-1\}$ it holds that 
\[
[V_i (X^\dagger_{2i-1} \otimes X^\dagger_{2i})] = [V_i],
\]
then for every $i \in \{1, \dotsc, k-1\}$,
\[
[(X^\dagger_{2i-1} \otimes X^\dagger_{2i})] = [\mathds{1}], 
\]
and so $[X_i] = [\mathds{1}]$ for every $i\in \{1, \dotsc, 2(k-1)\}$.

\end{proof}

\subsection{Proof of Theorem \ref{thm:quotTN1DFI}}
\label{sec:proofThm1}

The proof of \cref{thm:quotTN1DFI} follows automatically from the results obtained in the two previous sub-sections, i.e. the gauge description obtained in \cref{theoremgauge1} and the guarantee of a smooth manifold structure and a Riemannian metric making the projections from $M_{\textup{1Df}}$ and $N_{\textup{1Df}}$ to $B_{\textup{1Df}}$ and $B'_{\textup{1Df}}$ Riemannian submersions obtained in \cref{circuitswithphase,thm:action1Dcircuitsfixedinputphase}. 

\begin{proof}[Proof of \cref{thm:quotTN1DFI}]
First, as we mentioned in \cref{rmk:generic1}, the full-rank assumption made in the statement of \cref{theoremgauge1}, holds on an open and dense subset of $M_{\textup{1Df}}$. Thus, using \cref{circuitswithphase} and \cref{theoremgauge1} we can conclude that $(M_{\textup{1Df}}, g_{\textup{1Df}})$ defines the quotient tensor network manifold $(B_{\textup{1Df}}, h_{\textup{1Df}})$ as defined in \cref{circuitswithphase}, and that for every $x, y$ in an open and dense subset of $M_{\textup{1Df}}$ it holds that 
\[
\mathds{TN}_x = \mathds{TN}_y \iff [x]_{B_{\textup{1Df}}} = [y]_{B_{\textup{1Df}}}.
\]

To obtain the result for $(B'_{\textup{1Df}}, h'_{\textup{1Df}})$, recall that the projection map $(M_{\textup{1Df}}, g_{\textup{1Df}})\to (N_{\textup{1Df}}, \tilde g_{\textup{1Df}})$ is a Riemannian submersion. Now, using \cref{thm:action1Dcircuitsfixedinputphase} we know that the projection $(N_{\textup{1Df}}, \tilde g_{\textup{1Df}}) \to (B'_{\textup{1Df}}, h'_{\textup{1Df}})$ is a Riemannian submersion. This way, we can compose both Riemannian submersions to obtain a Riemannian submersion $(M_{\textup{1Df}}, g_{\textup{1Df}}) \to (B'_{\textup{1Df}}, h'_{\textup{1Df}})$, and use \cref{theoremgauge1} to conclude that for every $x, y$ in an open and dense subset of $M_{\textup{1Df}}$ it holds that 
\[
\exists \lambda \in \textup{U}(1)\ \mathit{s.t.}\ \mathds{TN}_x = \lambda \mathds{TN}_y \iff [x]_{B'_{\textup{1Df}}} = [y]_{B'_{\textup{1Df}}}.
\]

\end{proof}

\section{Depth-two quantum circuits in one dimension with no fixed input}
\label{sectionnofixedinput}

Let us now study the tensor network family of depth-two quantum circuits in one dimension with no fixed input and $k$ gates in the bottom layer,
\begin{equation}
\label{eq:circuito1Dnofixed}
\input{TikzFigures/circuitosimple2.tikz},
\end{equation}
where each gate acts on two neighbouring qudits, which are assumed to be unitary vectors in $\bb{C}^d$. The tensor network manifold associated with this family is
\[
M_{\textup{1D}} := \textup{U}(d^2)^{\times 2k-1},
\]
and its \textit{natural metric} is $g_{\textup{1D}} := g_{\textup{bi}}^{\oplus 2k-1}$. 

Again, in this setting there exist two quotient tensor network manifolds associated with $M_{\textup{1D}}$ which characterise quantum circuits both up to a phase and exactly. 

\begin{theorem}
\label{thm:quotTN1D}
Let $M_{\textup{1D}}$ be the tensor network manifold associated with the tensor network family of one-dimensional, depth-two circuits with $k$ gates in the bottom layer. When endowed with its natural metric $g_{\textup{1D}}$, it defines two quotient tensor network manifolds, $(B_{\textup{1D}}, h_{\textup{1D}})$ and $(B'_{\textup{1D}}, h'_{\textup{1D}})$ such that, for every $x, y$ on a open and dense subset of $M_{\textup{1D}}$, it holds that
\[
\mathds{TN}_x = \mathds{TN}_y \iff [x]_{B_{\textup{1D}}} = [y]_{B_{\textup{1D}}},
\]
and
\[
\exists \lambda \in \textup{U}(1) \text{ s.t. } \mathds{TN}_x = \lambda \mathds{TN}_y \iff [x]_{B'_{\textup{1D}}} = [y]_{B'_{\textup{1D}}}.
\]
\end{theorem}

\subsection{Characterising the gauge}

Studying the gauge freedom of one-dimensional, depth-two circuits with no fixed input can be easily reduced to studying that of depth-two circuits with a fixed input.  

\begin{theorem}
\label{theoremgaugenotfixed}
Consider the tensor network family of depth-two quantum circuits in one dimension with no fixed input and $k$ gates in the bottom layer, as shown in \cref{eq:circuito1Dnofixed}. Consider two tensor networks of this family, representing the circuits $U$ and $V$,
\[
U \equiv (A_1, \dotsc, A_k, U_1, \dotsc, U_{k-1}),\quad V \equiv (B_1, \dotsc, B_k, V_1, \dotsc, V_{k-1}),
\]
where the first $k$ coordinates denote the bottom gates of the circuit, and the last $k-1$ denote the top gates of the circuit. Assume that there exists a product input state $\ket{\psi}^{\otimes 2k}$ for which the tensors $A_i\ket{\psi}^{\otimes 2}$ are full-rank in the sense of \cref{theoremgauge1}, for every $i \in \{1, \dotsc, k\}$. Further assume that there exists some phase $\theta \in [0, 2\pi)$ for which $U = e^{i\theta }V$. Then, for every $j \in \{1, \dots, k-1\}$ there exist some unitary matrices $X_j, Y_j \in \textup{U}(d)$ for which 
\[
V_j \propto U_j (X^\dagger_j \otimes Y^\dagger_j),
\]
and 
\begin{align*}
B_1 &\propto (\mathds{1} \otimes X_1) A_1,\\
B_j &\propto (\tilde Y_{j-1} \otimes X_j) A_j,\ \forall j \in \{2, \dotsc, k-1\},\\
B_{k} &\propto (Y_{k-1} \otimes \mathds{1})A_k,    
\end{align*}
Furthermore, if $\theta = 0$, the above equations hold with equality. 
\end{theorem}
\begin{proof}
The fact that the gates in the top layer of the circuit are related by the gauge freedom described above follows automatically from \cref{theoremgauge1} (cf. \cref{sym1}), by considering the input state $\ket{\psi}^{\otimes 2k}$. Thus, we know that the bottom layer of the circuits are related as follows:
\begin{equation}
\label{eq:equalitybottomlayer1D}
\input{TikzFigures/fig10.tikz},
\end{equation}
where $A_1, \dotsc, A_k, B_1, \dotsc, B_k \in \textup{U}(d^2)$, $X_1, \dotsc, X_{k-1}, Y_1, \dotsc, Y_{k-1} \in \textup{U}(d)$. \cref{eq:equalitybottomlayer1D} implies that there exist some phases $\theta_j \in [0, 2\pi)$, $j \in \{1, \dotsc, k\}$ for which 
\begin{align}
\label{eq:equalitiesABs}
\begin{split}
(\mathds{1} \otimes X_1) A_1 &= e^{i\theta_1} B_1,\\
(Y_{j-1} \otimes X_{j}) A_j &= e^{i\theta_j} B_j,\quad \textup{for } j \in \{2, \dotsc, k-1\},\\
(Y_{k-1} \otimes \mathds{1}) A_k &= e^{i\theta_k} B_k.
\end{split}
\end{align}
This concludes the proof whenever $\theta \neq 0$. In the case when $\theta = 0$, we showed in \cref{theoremgauge1} that \cref{sym1} holds exactly and not up to a phase, and so \cref{eq:equalitybottomlayer1D} holds exactly as well. In this case, the equalities shown in \cref{eq:equalitiesABs} hold with the extra condition that $\sum_{j = 1}^k \theta_j \in 2\pi \bb{Z}$. Now, the result follows by \textit{absorbing} the phases in the matrices $X_i, Y_i$. Indeed, it suffices to take 
\[
\tilde{X}_j = e^{-i\sum_{m = 1}^j\theta_m} X_j,\quad \tilde{Y}_j = e^{i\sum_{m = 1}^j \theta_m} Y_j,
\]
which allows to conclude that
\begin{align*}
(\mathds{1} \otimes \tilde X_1) A_1 &= B_1,\\
(\tilde Y_{j-1} \otimes \tilde X_{j}) A_j &= B_j,\quad \textup{for } j \in \{2, \dotsc, k-1\},\\
(\tilde Y_{k-1} \otimes \mathds{1}) A_k &= B_k,
\end{align*}
finishing the proof. 
\end{proof}

\subsection{Quotient tensor network manifold structure}
\label{freeactionsourquotients}

Again, the gauge freedom identified in the previous subsection motivates the definition of two group actions that describe it. Recall that the tensor network manifold associated with depth-two quantum circuit in one dimension with $k$ gates in the bottom layer is
\[
M_{\textup{1D}} = \textup{U}(d^2)^{\times 2k-1},
\]
where each qudit is assumed to be a unit vector in $\bb{C}^d$. To describe a tensor network of this family up to a phase, it suffices to specify its gates up to a phase, and so we can consider the manifold
\[
N_{\textup{1D}} := \textup{PU}(d^2)^{\times 2k-1},
\]
which we endow with its \textit{natural metric} $\tilde g_{\textup{1D}} := g_{\textup{pu}}^{\oplus 2k-1}$. 

\begin{definition}[Gauge action for 1D circuits with no fixed input]
\label{def:gauge1D}
We define the action $\beta_1$ of $\textup{U}(d)^{\times 2(k-1)}$ on $M_{\textup{1D}}$ as follows,
\[
\beta_1: \textup{U}(d)^{\times 2(k-1)} \times M_{\textup{1D}} \to M_{\textup{1D}},\quad (X, (U, V)) \mapsto (XU, VX^\dagger),
\]
where $X \in \textup{U}(d)^{\times 2(k-1)}$, $U \in \textup{U}(d^2)^{\times k}$ represents the bottom layer of the circuit and $V \in \textup{U}(d^2)^{\times k-1}$ represents the top layer of the circuit. We write 
\[
X := (X_1, \dotsc, X_{2(k-1)}),\ U := (U_1, \dotsc, U_k),\ V := (V_1, \dotsc, V_{k-1}),
\]
and 
\begin{align*}
XU &:= ((\mathds{1} \otimes X_1) U_1, (X_2 \otimes X_3) U_2, \dotsc, (X_{2(k-2)} \otimes X_{2(k-1)-1}) U_{k-1} , (X_{2(k-1)} \otimes \mathds{1}) U_k),\\
VX^\dagger &:= (V_1 (X^\dagger_1 \otimes X^\dagger_2), \dotsc, V_{k-1}(X^\dagger_{2(k-1)-1} \otimes X^\dagger_{2(k-1)})).
\end{align*}
\end{definition}

\begin{definition}[Gauge action for 1D circuits with no fixed input, up to phase]
\label{def:gauge1DP}
We define the action $\beta_2$ of $\textup{PU}(d)^{\times 2(k-1)}$ on $N_{\textup{1D}}$ as follows,
\begin{align*}
\beta_2: \textup{PU}(d)^{\times 2(k-1)} \times N_{\textup{1D}} \to N_{\textup{1D}},\quad (X, (U, V)) &\mapsto (XU, VX^\dagger),
\end{align*}
where $X \in \textup{PU}(d)^{\times 2(k-1)}$, $U \in \textup{PU}(d^2)^{\times k}$ and $V \in \textup{PU}(d^2)^{\times k-1}$. We furthermore write
\[
X = ([X_1], \dotsc, [X_{2(k-1)}]),\ U = ([U_1], \dotsc, [U_k]),\ V = ([V_1], \dotsc, [V_{k-1}]),
\]
and 
\begin{align*}
XU &:= ([(\mathds{1} \otimes X_1) U_1], [(X_2 \otimes X_3) U_2], \dotsc, [(X_{2(k-2)} \otimes X_{2(k-1)-1}) U_{k-1}], [(X_{2(k-1)} \otimes \mathds{1}) U_k]),\\
VX^\dagger &:= ([V_1 (X^\dagger_1 \otimes X^\dagger_2)], \dotsc, [V_{k-1}(X^\dagger_{2(k-1)-1} \otimes X^\dagger_{2(k-1)})]).
\end{align*}
\end{definition}

The following two results guarantee the existence of a smooth manifold structure and a Riemannian metric on the quotient spaces associated with $\beta_1$ and $\beta_2$, which make the projection maps Riemannian submersions. We omit the proofs, as they are analogous to those of \cref{circuitswithphase}. For a visual representation of the setting considered, see \cref{fig:diag2}. 
\begin{figure}
    \centering
    \input{TikzFigures/diag2.tikz}
    \caption{Visual description of the setting considered in \cref{freeactionsourquotients} and \cref{thm:quotTN1D}. }
    \label{fig:diag2}
\end{figure}
\begin{proposition}
\label{circuitswithphasenoinput}
Let $M_{\textup{1D}}$ be endowed with its natural metric $g_{\textup{1D}}$. Consider the action $\beta_1$ of $\textup{U}(d)^{\times 2(k-1)}$ on $M_{\textup{1D}}$ given in \cref{def:gauge1D}. There is a smooth manifold structure and a Riemannian metric $h_{\textup{1D}}$ on $B_{\textup{1D}} := M_{\textup{1D}} / \textup{U}(d)^{\times 2(k-1)}$ for which the projection $(M_{\textup{1D}}, g_{\textup{1D}}) \to (B_{\textup{1D}}, h_{\textup{1D}})$ is a Riemannian submersion. 
\end{proposition}

\begin{proposition}
\label{circuitswithphasenoinputproj}
Let $N_{\textup{1D}}$ be endowed with its natural metric $\tilde g_{\textup{1D}}$. Consider the action $\beta_2$ of $\textup{PU}(d)^{\times 2(k-1)}$ on $N_{\textup{1D}}$ given in \cref{def:gauge1DP}. There is a smooth manifold structure and a Riemannian metric $h'_{\textup{1D}}$ on $B'_{\textup{1D}} := N_{\textup{1D}} / \textup{PU}(d)^{\times 2(k-1)}$ for which the projection $(N_{\textup{1D}}, \tilde g_{\textup{1D}}) \to (B'_{\textup{1D}}, h'_{\textup{1D}})$ is a Riemannian submersion. 
\end{proposition}

\subsection{Proof of Theorem \ref{thm:quotTN1D}}

To prove \cref{thm:quotTN1D} we can follow the same reasoning as in \cref{sec:proofThm1}. 

\begin{proof}[Proof of \cref{thm:quotTN1D}]
Again, note that the full-rank assumption made in \cref{theoremgaugenotfixed} holds in an open and dense subset of $M_{\textup{1D}}$. Using the gauge characterisation result shown in that theorem, along with \cref{circuitswithphasenoinput,circuitswithphasenoinputproj} we obtain two Riemannian submersions, the projections $(M_{\textup{1D}}, g_{\textup{1D}}) \to (B_{\textup{1D}},h_{\textup{1D}})$ and $(N_{\textup{1D}},\tilde g_{\textup{1D}}) \to (B'_{\textup{1D}}, h'_{\textup{1D}})$. Lastly, since the projection $(M_{\textup{1D}}, g_{\textup{1D}}) \to (N_{\textup{1D}},\tilde g_{\textup{1D}})$ is a Riemannian submersion by construction (cf. \cref{sec:manifolds}), we obtain a Riemannian submersion $(M_{\textup{1D}}, g_{\textup{1D}}) \to (B'_{\textup{1D}}, h'_{\textup{1D}})$ which allows us to conclude the proof. 
\end{proof}

\section{Depth-two, two-dimensional circuits with a fixed input}
\label{sec:twodimcircuitsfixed}

We now focus on the tensor network family of depth-two, two-dimensional quantum circuits on a $2k \times 2k$ square lattice with a fixed input,
\begin{equation}
\label{estado2D}
\input{TikzFigures/2D1-1.tikz}.
\end{equation}
We assume that the two layers of unitaries of the circuit---the bottom layer in blue and the top layer in pink---consist of possibly different gates $U_{(i, j)}, V_{(l, m)} \in \textup{U}(d^4)$ for $(i, j) \in \{1, \dotsc, k\}^2$ and $(l, m) \in \{1, \dotsc, k-1\}^2$. The (fixed) input of the circuit can be assumed to be $\ket{0}^{\otimes (2k)^2} \in (\bb{C}^d)^{(2k)^2}$. Therefore each tensor network 
of the family consists of two types of tensors:
\[
\input{TikzFigures/2D1-2.tikz}.
\]
and so we can see each tensor network of the family as a tuple
\[
(\Lambda_{(1, 1)}, \dotsc, \Lambda_{(k, k)}, V_{(1,1)}, \dotsc, V_{(k-1,k-1)}),
\]
where $\Lambda_{(i, j)} := U_{(i, j)}\ket{0}^{\otimes 4}$ is an element of $\bb{S}^{2d^4-1}$ for every $(i, j) \in \{1, \dotsc, k\}^2$ and $V_{(l, m)}$ is an element of $\textup{U}(d^4)$ for every $(l, m) \in \{1, \dotsc, k-1\}^2$. The tensor network manifold associated with this family is 
\[
M_{\textup{2Df}} := (\bb{S}^{2d^4-1})^{\times k^2} \times \textup{U}(d^4)^{\times (k-1)^2},
\]
and its \textit{natural metric} is $g_{\textup{2Df}} := g_{\textup{round}}^{\oplus k^2} \oplus g_{\textup{bi}}^{\oplus (k-1)^2}$. 

In this case, we can obtain a quotient tensor network manifold that characterises the output state of the circuit up to a phase. 
\begin{theorem}
\label{thm:quotTN2DFI}
Let $M_{\textup{2Df}}$ be the tensor network manifold associated with the tensor network family of two-dimensional, depth-two circuits with a fixed input on a $2k \times 2k$ square lattice. When endowed with its natural metric $g_{\textup{2Df}}$, it defines a quotient tensor network manifold $(B_{\textup{2Df}}, h_{\textup{2Df}})$ such that, for every $x, y$ on a open and dense subset of $M_{\textup{2Df}}$, it holds that
\[
\exists \lambda \in \textup{U}(1) \text{ s.t. } \mathds{TN}_x = \lambda \mathds{TN}_y \iff [x]_{B_{\textup{2Df}}} = [y]_{B_{\textup{2Df}}}.
\]
\end{theorem}

As in the case of \cref{thm:quotTN1DFI,thm:quotTN1D}, the proof of this result will follow from the results in \cref{sec:characterising2Df,sec:quotTN2df} below. We will omit the details in this case. Also note that in this case we will not obtain a quotient manifold that characterises the output state of the circuit exactly. The reason for this will be made clearer in \cref{sec:quotTN2df}. 

\subsection{Characterising the gauge}
\label{sec:characterising2Df}

In this case, we characterise the gauge only up to a phase. 

\begin{theorem}
\label{simetria2D}
Consider the tensor network family of depth-two, two-dimensional quantum circuits on a $2k \times 2k$ square lattice with a fixed input, as in \cref{estado2D}. Consider two tensor networks of the family representing the states $\ket{U}$ and $\ket{V}$. Assume that the tensor network $(\Lambda_{(1, 1)}, \dotsc, \Lambda_{(k, k)}, V_{(1,1)}, \dotsc, V_{(k-1,k-1)}))$ that represents $\ket{U}$ is such that, given a cut between any two rows or columns of bottom states, every state $\Lambda_{(i, j)}$ adjacent to the cut is such that there exists an operator acting on the sites marked in green of $\Lambda_{(i, j)}$,
\begin{equation}
\label{figcolores}
\input{TikzFigures/2D9.tikz}    
\end{equation}
and such that it generates any element in the Hilbert space corresponding to the site(s) marked in blue in $\Lambda_{(i, j)}$. Assume that there exists some $\theta \in [0, 2\pi)$ such that $\ket{U} = e^{i\theta} \ket{V}$. Then, the tensors of both tensor networks are related up to a phase by the group action given by the multiplication by unitaries and their inverses over the contracting legs. 
\end{theorem}

Although we can adapt \cref{simetria2D} for the case when $\theta = 0$ to conclude that the tensors are related strictly and not just up to a phase by the gauge described, we will not be able to study this gauge freedom as a \textit{free} action on $M_{\textup{2Df}}$.  

Note that the full-rank assumption of \cref{simetria2D} is always fulfilled when considering a \textit{general} state. Indeed, the assumption can be rewritten in the following ways:
\begin{remark}
\label{rmk:fullrankrmk}
Consider a cut between any two rows or columns of bottom states of the lattice, as in \cref{figcolores}. Let $\Lambda_{(i, j)}$ be any state adjacent to the cut, and let $\mathcal{H}_A$ and $\mathcal{H}_B$ be the Hilbert spaces corresponding to the sites of $\Lambda_{(i, j)}$ marked in green and blue in \cref{figcolores}, respectively. Seeing $\Lambda_{(i, j)}$ as a state in $\mathcal{H}_A \otimes \mathcal{H}_B$, the assumption above is equivalent to assuming that for every $\ket{\psi} \in \mathcal{H}_B$ there exists some $\bra{\phi} \in \mathcal{H}_A^*$ such that $\bra{\phi}\Lambda_{(i, j)} = \ket{\psi}$, which in turn is equivalent to assuming that $\Lambda_{(i, j)}$, when seen as a map from $\mathcal{H}_A$ to $\mathcal{H}_B$, is surjective. 

This assumption can also be described in terms of the Schmidt rank of the bottom states; we are assuming that the bottom states have full Schmidt rank when seen as states in $\mathcal{H}_A \otimes \mathcal{H}_B$. 
\end{remark}

Before proving \cref{simetria2D}, let us first state an auxiliary result. 
\begin{lemma}[{\cite[Lemma 7]{perez2010characterizing}}]
\label{productstoproducts}
Let $Y$ be an invertible matrix acting on $n$ copies of a Hilbert space, and such that it takes products to products. Then, $Y$ is of the form $P_\pi (Y_1 \otimes \dotsm \otimes Y_n)$, where $P_\pi$ is a permutation of the Hilbert spaces. 
\end{lemma}

Let us now prove \cref{simetria2D}. For readability, all the figures used below correspond to an $8 \times 8$ system. Nevertheless, note that the proof is valid for any system size. 

\begin{proof}[Proof of \cref{simetria2D}]
Consider two tensor networks representing the same state up to a phase, 
\begin{equation}
\label{eq:initialequality2D}
\input{TikzFigures/2D2.tikz}.
\end{equation}
Although the states on the bottom layer of the circuits have not been labelled for readability, they could be pairwise distinct.

Let us fix a cut between two columns of bottom states---represented with a dotted red line in \cref{eq2dcolumna}. Multiplying by $V_i^\dagger$ on both sides of the equality for every $i$ not in the cut, and by $\tilde V_i^\dagger$ for every $i$ along the cut we conclude that
\begin{equation}
\label{eq2dcolumna}
\input{TikzFigures/2D3.tikz}.    
\end{equation}
Recall that, by assumption, we can generate any input product state (with respect to the cut) on which the top layer on the left-hand side of \cref{eq2dcolumna} acts, by acting on the states shown below in green,
\[
\input{TikzFigures/2D4.tikz}.
\]
Also, from \cref{eq2dcolumna}, we know that any output state obtained by acting on the states in green will be proportional to a product state of the form 
\[
\input{TikzFigures/2D4bis.tikz},
\]
where the green elements in the above equation represent unitary gates. This allows us to conclude that the top layer on the left-hand side of \cref{eq2dcolumna} maps product states to product states, and so by \cref{productstoproducts} we can conclude that 
\[
\input{TikzFigures/2D5.tikz},
\]
where $P_\pi$ is either the SWAP operator or the identity, an the two vertical rectangles on the right-hand side represent two unitary matrices. Moreover, since every bottom state on the left-hand side of \cref{eq2dcolumna} adjacent to the cut is \textit{full-rank} in the sense \cref{rmk:fullrankrmk}, it follows that
\[
\input{TikzFigures/2D5-2.tikz},
\]
and so $P_\pi$ cannot be the SWAP operator. This allows us to conclude that
\begin{equation}
\label{eq:aux1}
\input{TikzFigures/2D5-3.tikz}.    
\end{equation}
Now, since the left-hand side of this equation is a product operator in the vertical direction as well, we can apply \cref{productstoproducts} to conclude that the right-hand side of \cref{eq:aux1} is also a product of unitaries in the vertical direction, and so 
\[
\input{TikzFigures/2D6.tikz}.
\]

Following the same reasoning considering cuts between any two rows or columns of the lattice, we can conclude that for every $I \in \{1, \dots, k-1\}^2$ there exist some unitary matrices $A_I, B_I, C_I, D_I \in \textup{U}(d^2)$ such that
\[
\input{TikzFigures/2D7.tikz},
\]
This implies that for every $I \in \{1, \dots, k-1\}^2$ there exist some unitary matrices $X_I, Y_I, Z_I, W_I \in \textup{U}(d)$ such that 
\[
\input{TikzFigures/2D8.tikz}.
\]
Indeed, we can assume that the matrices are unitary by \cref{productounitarias}. Recall that the above equality is only up to a phase, that is, for every $(j, l) \in \{1, \dotsc, k-1\}^2$ there exist some $\varphi_{(j, l)} \in [0, 2\pi)$ and $X_{(j, l)}, Y_{(j, l)}, W_{(j, l)}, Z_{(j, l)} \in \textup{U}(d)$ such that 
\[
V_{(j, l)} = e^{i\varphi_{(j, l)}}\tilde{V}_{(j, l)} (W_{(j, l)} \otimes X_{(j, l)} \otimes Y_{(j, l)} \otimes Z_{(j, l)}).
\]
Inserting this expression into \cref{eq:initialequality2D} we can conclude that for every $(j, l) \in \{1, \dotsc, k\}^2$, there exists some phase $\phi_{(j,l)} \in [0, 2\pi)$ such that the bottom state at the $(j,l)$-th site of $\ket{U}$, is related to that of $\ket{V}$, as
\[
\Sigma_{(j, l)} =  e^{i\phi_{(j,l)}}(Z_{(j-1, l-1)} \otimes Y_{(j-1, l)} \otimes X_{(j, l-1)} \otimes W_{(j, l)}) \Lambda_{(j, l)},
\]
where by definition 
\begin{align*}
Z_{(0, l)} = Z_{(l, 0)} = Y_{(0, l)} = X_{(l, 0)} := \mathds{1},\\
Y_{(l, k)} = X_{(k, l)} = W_{(k, l)} = W_{(l, k)} := \mathds{1},
\end{align*}
for every $l \in \{1, \dotsc, k\}$. 
\end{proof}

\subsection{Quotient tensor network manifold structure}
\label{sec:quotTN2df}

As we saw earlier, the tensor network manifold associated with the tensor network family of two-dimensional, depth-two quantum circuits on a $2k \times 2k$ lattice with a fixed input is 
\[
M_{\textup{2Df}} = (\bb{S}^{2d^4-1})^{\times k^2} \times \textup{U}(d^4)^{\times (k-1)^2},
\]
where each qudit is assumed to be a unit vector in $\bb{C}^d$. Since we are only interested in the description of the circuit up to a global phase, we define the manifold 
\[
N_{\textup{2Df}} := (\bb{CP}^{d^4-1})^{\times k^2} \times \textup{PU}(d^4)^{\times (k-1)^2},
\]
which we endow with its \textit{natural metric} $g_{\textup{2Df}} := g_{\textup{FS}}^{\oplus k^2} \oplus g_{\textup{pu}}^{\oplus (k-1)^2}$. Using this phase-agnostic description of the tensor networks, we may define the group action of $\textup{PU}(d)^{4(k-1)^2}$ on $N_{\textup{2Df}}$ that describes the gauge freedom.

\begin{definition}[Gauge action for 2D circuits with a fixed input, up to phase]
\label{def:gauge2DFI}
We define the action of $\textup{PU}(d)^{4(k-1)^2}$ on $N_{\textup{2Df}}$ as 
\begin{align}
\textup{PU}(d)^{4(k-1)^2} \times N_{\textup{2Df}} \to N_{\textup{2Df}},\quad (X, (\Lambda, V)) \mapsto (X\Lambda, VX^\dagger), 
\end{align}
where $X \in \textup{PU}(d)^{4(k-1)^2}$, $\Lambda \in (\bb{CP}^{d^4-1})^{\times k^2}$ and $V \in \textup{PU}(d^4)^{\times (k-1)^2}$. Moreover,
\begin{align*}
X &:= ([W_{(1,1)}], \dotsc, [W_{(k-1,k-1)}], [X_{(1,1)}], \dotsc, [X_{(k-1,k-1)}],
\\&\quad [Y_{(1,1)}], \dotsc, [Y_{(k-1,k-1)}], [Z_{(1,1)}], \dotsc, [Z_{(k-1,k-1)}]),\\
\Lambda &:= ([\Lambda_{(1,1)}], \dotsc, [\Lambda_{(k,k)}]),\\
V &:= ([V_{(1,1)}], \dotsc, [V_{(k-1, k-1)}]),
\end{align*}
and 
\begin{align*}
X\Lambda &:= ([(\mathds{1} \otimes \mathds{1} \otimes \mathds{1} \otimes W_{(1, 1)}) \Lambda_{(1, 1)}], \dotsc, 
\\&\quad [(Z_{(l-1, m-1)} \otimes Y_{(l-1, m)} \otimes X_{(l, m-1)} \otimes W_{(l, m)}) \Lambda_{(l, m)}], \dotsc, 
\\&\quad [(Z_{(k-1, k-1)} \otimes \mathds{1} \otimes \mathds{1} \otimes \mathds{1}) \Lambda_{(k, k)}]),\\
VX^\dagger &:= ([V_{(1,1)} (W^\dagger_{(1,1)} \otimes X^\dagger_{(1,1)} \otimes Y^\dagger_{(1,1)} \otimes Z^\dagger_{(1,1)})], \dotsc,
\\&\quad [V_{k-1}(W^\dagger_{(k-1,k-1)} \otimes X^\dagger_{(k-1,k-1)} \otimes Y^\dagger_{(k-1,k-1)} \otimes Z^\dagger_{(k-1,k-1)})]),
\end{align*}
with
\begin{align*}
Z_{(0, l)} = Z_{(l, 0)} = Y_{(0, l)} = X_{(l, 0)} := \mathds{1},\\
Y_{(l, k)} = X_{(k, l)} = W_{(k, l)} = W_{(l, k)} := \mathds{1},
\end{align*}
for every $l \in \{0, \dotsc, k\}$. 
\end{definition}

Using this action, we can construct a quotient tensor network manifold that characterises the output state of a two-dimensional quantum state with a fixed input up to a global phase. For a visual representation of the setting considered, see \cref{fig:diag3}. On the other hand, as we mentioned in the introduction of this section, the analogous of the above action on $M_{\textup{2Df}}$ fails to be free, and so the smooth manifold structure of its corresponding quotient space is not guaranteed. 
\begin{figure}
    \centering
    \input{TikzFigures/diag3.tikz}
    \caption{Visual representation of the setting considered in \cref{sec:quotTN2df}.}
    \label{fig:diag3}
\end{figure}

\begin{proposition}
Let $N_{\textup{2Df}}$ be endowed with its natural metric $\tilde g_{\textup{2Df}}$. Consider the action of $\textup{PU}(d)^{4(k-1)^2}$ on $N_{\textup{2Df}}$ given in \cref{def:gauge2DFI}. There is a smooth manifold structure and a Riemannian metric $h_{\textup{2Df}}$ on $B_{\textup{2Df}} := N_{\textup{2Df}}/\textup{PU}(d)^{4(k-1)^2}$ for which the projection map $(N_{\textup{2Df}}, \tilde g_{\textup{2Df}}) \to (B_{\textup{2Df}}, h_{\textup{2Df}})$ is a Riemannian submersion. 
\end{proposition}
We omit the proof, as it is analogous to that of \cref{thm:action1Dcircuitsfixedinputphase}.

\section{Depth-two, two-dimensional circuits with no fixed input}

The results obtained in the previous section can be easily adapted for the tensor network family of depth-two, two-dimensional circuits on a $2k \times 2k$ lattice with no fixed input. Since the adaptations are straightforward, we only present the main result and omit the proof. The tensor network manifold associated with this family is
\[
M_{\textup{2D}} := \textup{U}(d^4)^{k^2 + (k-1)^2},
\]
which we endow with its natural metric $g_{\textup{2D}} := g_{\textup{bi}}^{\oplus k^2 + (k-1)^2}$. 

\begin{theorem}
\label{thm:quotTN2D}
Let $M_{\textup{2D}}$ be the tensor network manifold associated with the tensor network family of two-dimensional, depth-two circuits on a $2k \times 2k$ lattice. When endowed with its natural metric $g_{\textup{2D}}$, it defines a quotient tensor network manifold $(B_{\textup{2D}}, h_{\textup{2D}})$ such that, for every $x, y$ on a open and dense subset of $M_{\textup{2D}}$, it holds that
\[
\exists \lambda \in \textup{U}(1) \text{ s.t. } \mathds{TN}_x = \lambda \mathds{TN}_y \iff [x]_{B_{\textup{2D}}} = [y]_{B_{\textup{2D}}}.
\]
\end{theorem}

In order to define the quotient tensor network manifold $B_{\textup{2D}}$, we first define the manifold
\[
N_{\textup{2D}} := \textup{PU}(d^4)^{k^2 + (k-1)^2}, 
\]
endowed with its \textit{natural metric} $g_{\textup{2D}} := g_{\textup{pu}}^{\oplus k^2 + (k-1)^2}$, which allows us to describe a tensor network of the family up to a phase. Then, we consider the action of $\textup{PU}(d)^{4(k-1)^2}$ on $N_{\textup{2D}}$ that describes the gauge freedom, and define $B_{\textup{2D}}$ as $N_{\textup{2D}} / \textup{PU}(d)^{4(k-1)^2}$. See \cref{fig:diag4} for a visual representation of the setting considered. 
\begin{figure}
    \centering
    \input{TikzFigures/diag4.tikz}
    \caption{Visual representation of the setting considered in \cref{thm:quotTN2D}}
    \label{fig:diag4}
\end{figure}

\section{Matrix Product States}
\label{sec:sectionMPS}

Let us now consider the tensor network family of length-$k$ injective MPSs in canonical form with bond dimension $D$ and physical dimension $d$. For our purpose, we consider open boundary conditions, and we promote the boundary legs of the MPSs to physical legs. Whilst this assumption is needed in our derivation, its effect on the bulk of the MPSs should be negligible when considering sufficiently large values of $k$. Every tensor network of the family is of the form
\[
\input{TikzFigures/mpsstate.tikz},
\]
where each tensor $U_i$ is a linear isometry, i.e. a $dD \times D$ matrix such that $U_i^\dagger U_i = \mathds{1}_D$, and $\Lambda$ is some unit vector in $\bb{C}^{D^2}$, i.e. $\Lambda \in \bb{S}^{2D^2-1}$. Thus, the tensor network manifold associated with this family is
\[
M_{\textup{mps}} := \bb{S}^{2D^2-1} \times \textup{V}_D(\bb{C}^{dD})^{\times k},
\]
which we endow with its \textit{natural metric} $g_{\textup{mps}} := g_{\textup{round}} \oplus g_{\textup{st}}^{\oplus k}$ (cf. \cref{sec:manifolds}).

In this setting, we can obtain a quotient tensor network manifold what allows us to characterise the state represented by an MPS uniquely and a second quotient tensor network manifold that allows us to do the same up to a phase. As we mentioned in the introduction, the results obtained in this section can be understood as a reformulation of the results obtained in \cite[Section III]{haegeman2014geometry} in the setting of real manifolds rather than complex manifolds, and considering open boundary conditions. 

\begin{theorem}
\label{thm:quotTNMPS}
Let $M_{\textup{mps}}$ be the tensor network manifold associated with the tensor network family of injective MPS of length $k$ with bond dimension $D$ and physical dimension $d$ in canonical form. When endowed with its natural metric $g_{\textup{mps}}$, it defines two quotient tensor network manifolds, $(B_{\textup{mps}}, h_{\textup{mps}})$ and $(B'_{\textup{mps}},  h'_{\textup{mps}})$ such that, for every $x, y$ on a open and dense subset of $M_{\textup{mps}}$, it holds that
\[
\mathds{TN}_x = \mathds{TN}_y \iff [x]_{B_{\textup{mps}}} = [y]_{B_{\textup{mps}}}.
\]
and
\[
\exists \lambda \in \textup{U}(1) \text{ s.t. } \mathds{TN}_x = \lambda \mathds{TN}_y \iff [x]_{B'_{\textup{mps}}} = [y]_{B'_{\textup{mps}}}.
\]
\end{theorem}

As in previous cases, the proof of this result follows from the next two subsections and we omit the details.

\subsection{Characterising the gauge}

Let us first characterise the gauge freedom of the tensor networks in this family. 
\begin{theorem}
\label{gaugeMPSthm}
Consider the tensor network family of injective MPS of length $k$ with bond dimension $D$ and physical dimension $d$ in canonical form. Consider two tensor networks of the family representing the states $\ket{U}$ and $\ket{V}$. Let $(\Lambda, U_1, \dotsc, U_n)$ denote the tensor network associated with $\ket{U}$, and assume that the isometries $U_i$ are surjective when seen as linear maps from the Hilbert spaces associated with their leftmost legs to the Hilbert space associated with their rightmost leg,
\begin{equation}
\label{eqcut}
\input{TikzFigures/mps2sitefig7.tikz}, 
\end{equation}
i.e. when seen as a linear map $\bb{C}^{D} \times \bb{C}^d \to \bb{C}^{D}$. Further assume that $\Lambda$ is full-rank when seen as a map from $\bb{C}^D$ to $\bb{C}^D$. Assume that there exists some $\theta \in [0, 2\pi)$ for which $\ket{U} = e^{i\theta} \ket{V}$. Then, the tensors of both tensor networks are related up to a phase by the multiplication by unitaries and their inverses over the contracting edges. If $\theta = 0$, the tensors are related, strictly and not up to a phase, by the same gauge. 
\end{theorem}

Before we prove this result, let us introduce an auxiliary lemma that relates two purifications of the same reduced density matrix (see \cite[Section 2.5]{nielsen2010quantum}).
\begin{lemma}
\label{purificationsym}
Let $\ket{AE_1}$ and $\ket{AE_2}$ be two purifications of a state $\rho^A$ to a composite system $AE$. Then there exists a unitary transformation $U_E$ on $E$ such that
\[
\ket{AE_1} = (\mathds{1}_A \otimes U_E)\ket{AE_2}.
\]
\end{lemma}

\begin{proof}[Proof of \cref{gaugeMPSthm}]
Assume first that we have two length-one MPS representations of the same state up to a global phase,
\begin{equation}
\label{eqdepth1mps}
\input{TikzFigures/mps1sitefig1.tikz},
\end{equation}
where the MPS on the right-hand side of the equality is \textit{full-rank} in the sense of the assumptions of the theorem. If we consider the partial trace over the two rightmost legs at both sides of the equality above, we obtain 
\[
\input{TikzFigures/mps1sitefig2.tikz}.    
\]
Thus, using \cref{purificationsym}, we know that there exists some unitary matrix $X$ such that
\[
\input{TikzFigures/mps1sitefig3.tikz}.
\]
Inserting this identity in \cref{eqdepth1mps} we obtain
\begin{equation*}
\input{TikzFigures/mps1sitefig4.tikz}.
\end{equation*}
Now, $\Sigma$ is full-rank by assumption, and so we can invert it to conclude that
\begin{equation*}
\input{TikzFigures/mps1sitefig5.tikz},
\end{equation*}
finishing the proof. 

Let us now consider two length-two MPS representations of the same state up to a phase, 
\begin{equation}
\label{eqdepth2mps}
\input{TikzFigures/mps2sitefig1.tikz},
\end{equation}
where by assumption the MPS on the right-hand side satisfies the assumptions of the theorem. Considering the partial trace over the three rightmost legs we reach the same conclusion as in the one-site case; i.e. $\Lambda = e^{i\varphi} \Sigma X$ for some unitary matrix $X$ and some phase $\varphi$. Inserting this into \cref{eqdepth2mps} we obtain
\begin{equation}
\label{eq:depth2mps2}
\input{TikzFigures/mps2sitefig2.tikz}.
\end{equation}
If we now consider the partial trace in both states over the two rightmost legs, we obtain
\begin{equation*}
\input{TikzFigures/mps2sitefig3.tikz}.
\end{equation*}
Applying \cref{purificationsym} and using the fact that $\Sigma$ is full-rank, we conclude that there exists some unitary matrix $Y$ and some phase $\phi$ such that $U_1 X= e^{i\phi}(\mathds{1} \otimes Y)V_1 $, i.e. 
\begin{equation*}
\input{TikzFigures/mps2sitefig4.tikz}.
\end{equation*}
Again, inserting the above expression in \cref{eq:depth2mps2} and multiplying both sides by $\Sigma^{-1}$ we conclude that
\begin{equation*}
\input{TikzFigures/mps2sitefig5.tikz}.  
\end{equation*}
By the assumption that $V_1$ is full-rank when seen as a map from its leftmost legs to its rightmost leg (cf. \cref{eqcut}), we can conclude that 
\begin{equation*}
\input{TikzFigures/mps2sitefig6.tikz},
\end{equation*}
finishing the proof. 

In the general case, if one considers two length-$k$, $k > 2$, MPS representations of the same state up to a phase, it is now easy to see that the unique freedom is given by the multiplication by unitary matrices and their inverses over the contracting legs. This result follows from considering partial traces in the same manner as we have done above, in a \textit{bottom-up} fashion.  Using the fact that the isometries are full-rank (cf. \cref{eqcut}) we can reduce the size of the MPS by one in each step.

Furthermore, if $\theta = 0$, the above proof holds identically with equality. 
\end{proof}

\subsection{Quotient tensor network manifold structure}
\label{sectionfreeactionMPS}

As we saw earlier, the tensor network manifold associated with the tensor network family of length-$k$ injective MPS in canonical form with bond dimension $D$ and physical dimension $d$ is
\[
M_{\textup{mps}} = \bb{S}^{2D^2-1} \times \textup{V}_D(\bb{C}^{dD})^{\times k}.
\]
If we are only interested in describing an MPS of this family up to a phase, we can use the manifold 
\[
N_{\textup{mps}} := \bb{CP}^{D^2-1} \times \textup{PV}_{D}(\bb{C}^{dD})^{\times k},
\]
which we endow with its \textit{natural metric} $\tilde g_{\textup{mps}} := g_{\textup{FS}} \oplus g_{\textup{pst}}^{\oplus k}$ (cf. \cref{sec:manifolds}).

Let us now define two actions, one on $M_{\textup{mps}}$ and the other on $N_{\textup{mps}}$, which describe the gauge freedom identified in the previous subsection. 
\begin{definition}[Gauge action on matrix product states]
\label{def:actionMPS}
We define the action $\delta_1$ of $\textup{U}(D)^{\times k}$ on $M_{\textup{mps}}$ as
\[
\textup{U}(D)^{\times k} \times M_{\textup{mps}} \to M_{\textup{mps}},\quad (X, (U, V)) \mapsto ((\mathds{1} \otimes X_1)U, XV),
\]
where 
\begin{align*}
X := (X_1, \dotsc, X_{k}) \in \textup{U}(D)^{\times k},\quad U \in \bb{S}^{2D^2-1},\quad V := (V_1, \dotsc, V_k) \in \textup{V}_{D}(\bb{C}^{dD})^{\times k},
\end{align*}
and 
\[
XV := ((\mathds{1} \otimes X_2)V_1 X^\dagger_1, (\mathds{1} \otimes X_3) V_2 X^\dagger_2, \dotsc,(\mathds{1} \otimes X_k) V_{k-1} X^\dagger_{k-1} , V_k X^\dagger_k) \in  \textup{V}_{D}(\bb{C}^{dD})^{\times k}.
\]
\end{definition}

\begin{definition}[Gauge action on matrix product states, up to phase]
\label{def:actionMPSphase}
We define the action $\delta_2$ of $\textup{PU}(D)^{\times k}$ on $N_{\textup{mps}}$ as
\begin{align*}
\textup{PU}(D)^{\times k} \times N_{\textup{mps}} \to N_{\textup{mps}},\quad (X, ([U], V)) \mapsto ([(\mathds{1} \otimes X_1)U], XV),
\end{align*}
where 
\[
X = ([X_1], \dotsc, [X_{k}]) \in \textup{PU}(D)^{\times k},\quad U \in \bb{CP}^{D^2-1},\quad  V = ([V_1], \dotsc, [V_k]) \in \textup{PV}_{D}(\bb{C}^{dD})^{\times k},
\]
and 
\[
XV := ([(\mathds{1} \otimes X_2)V_1 X^\dagger_1], [(\mathds{1} \otimes X_3) V_2 X^\dagger_2], \dotsc,[(\mathds{1} \otimes X_k) V_{k-1} X^\dagger_{k-1}] , [V_k X^\dagger_k]). 
\]
\end{definition}

Again, these two actions allow us to obtain the quotient tensor network manifolds that characterise the state represented by an MPS of this family both uniquely and up to a phase. For a visual representation of the setting considered, see \cref{fig:diag5}. 
\begin{figure}
    \centering
    \input{TikzFigures/diag5.tikz}
    \caption{Visual representation of the setting considered in \cref{thmfreeactionMPS,thmfreeactionMPS2}.}
    \label{fig:diag5}
\end{figure}

\begin{proposition}
\label{thmfreeactionMPS2}
Let $M_{\textup{mps}}$ be endowed with its natural metric $g_{\textup{mps}}$. Consider the action $\delta_1$ of $\textup{U}(D)^{\times k}$ on $M_{\textup{mps}}$ given in \cref{def:actionMPS}. There is a smooth manifold structure and a Riemannian metric $h_{\textup{mps}}$ on $B_{\textup{mps}} := M_{\textup{mps}}/ \textup{U}(D)^{\times k}$ for which the projection $(M_{\textup{mps}}, g_{\textup{mps}}) \to (B_{\textup{mps}}, h_{\textup{mps}})$ is a Riemannian submersion. 
\end{proposition}
\begin{proof}
The action is smooth. To see that it is isometric, recall that the multiplication by unitary matrices is an isometry both on the sphere endowed with the round metric and on the unitary group endowed with its bi-invariant metric (cf. \cref{prop:isomunitaries,prop:isomsphere}). Thus, using \cref{projectiveIsometry} we can conclude that multiplying by a projective unitary matrix is an isometric action on $\textup{V}_D(\bb{C}^{dD})$ whenever it is endowed with $g_{\textup{st}}$. 

It remains to show that the action is free. For this, let $V \in \textup{V}_D(\bb{C}^{dD})^{\times k}$ and assume that there exists an element $X \in \textup{U}(D)^{\times k}$ such that 
\[
XV = V. 
\]
In particular
\[
V_k X_k^\dagger = V_k,
\]
which implies that $X_k = \mathds{1}$, as $V_k^\dagger V_k = \mathds{1}$. Looking at the $k-1$-th coordinate of $V$, we now see that $V_{k-1} X^\dagger_{k-1} = V_{k-1}$. Following the same reasoning, we can conclude that $X_i = \mathds{1}$ for every $i \in \{1, \dotsc, k\}$, finishing the proof. 
\end{proof}

\begin{proposition}
\label{thmfreeactionMPS}
Let $N_{\textup{mps}}$ be endowed with its natural metric $\tilde g_{\textup{mps}}$. Consider the action $\delta_2$ of $\textup{PU}(D)^{\times k}$ on $N_{\textup{mps}}$ as given in \cref{def:actionMPSphase}. There is a smooth manifold structure and a Riemannian metric $h'$ on $B'_{\textup{mps}} := N_{\textup{mps}} / \textup{PU}(D)^{\times k}$ for which the projection $(N_{\textup{mps}}, \tilde g_{\textup{mps}}) \to (B'_{\textup{mps}}, h'_{\textup{mps}})$ is a Riemannian submersion. 
\end{proposition}
\begin{proof}
The proof of this result can be easily obtained from that of the previous proposition.
\end{proof}

\section{Two-dimensional sequentially generated states}
\input{TikzFigures/preamble}
Let us now consider the tensor network family of two-dimensional sequentially generated states \cite{banuls2008sequentially} with bond dimension $D$ and physical dimension $d$ on an $n \times (k+2)$ lattice. Every tensor network of this family is constructed as follows: first, consider $n$ injective MPS of length $k$ in canonical form with bond dimension $D$ and physical dimension $d$, arranged as the rows of the lattice (see \cref{fig:peps1}). On top of the MPS rows, consider unitaries stacked along the columns of the lattice, leaving the rightmost leg of every MPS untouched (see Figure \eqref{fig:peps2}).
\begin{figure}
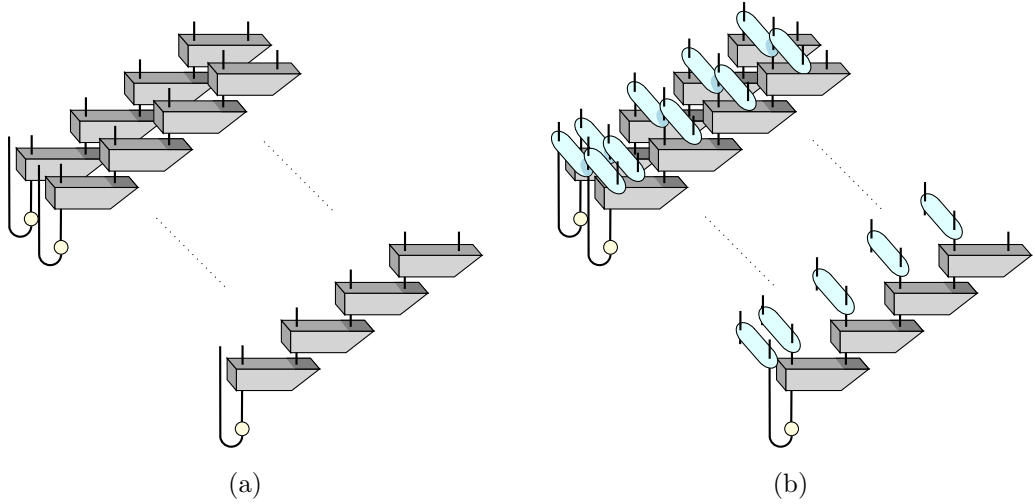

\centering
\begin{subfigure}{.48\textwidth}
  \centering
  \ctikzfig{TikzFigures/peps1}
  \caption{}
  \label{fig:peps1}
\end{subfigure}%
\begin{subfigure}{.48\textwidth}
  \centering
  \ctikzfig{TikzFigures/peps2}
  \caption{}
  \label{fig:peps2}
\end{subfigure}
\caption{Visual representation of a two-dimensional sequentially generated state. The figure on the left-hand side shows the MPSs arranged along the rows of the lattice, and the figure on the right-hand side shows the state with the unitaries stacked along the columns of the lattice.}
\label{fig:peps}
\end{figure}
In this case, we also promote the boundary legs of each MPS to physical legs. Also note that we are not assuming that the state is translationally invariant; each tensor of the tensor network can be different from the rest. The tensor network manifold associated with this family is 
\[
M_{\textup{seq}} := (\bb{S}^{2D^2-1})^{\times n} \times \textup{V}_{D}(\bb{C}^{dD})^{\times nk} \times \textup{U}(D^2)^{\times n-1} \times \textup{U}(d^2)^{\times k(n-1)},
\]
which we endow with its \textit{natural metric} $g_{\textup{seq}} := g_{\textup{round}}^{\oplus n} \oplus g_{\textup{st}}^{\oplus nk} \oplus g_{\textup{bi}}^{\oplus (k+1)(n-1)}$. 

In this case, similarly to when we considered two-dimensional circuits, we can only find a quotient tensor network manifold that characterises sequentially generated states up to a phase. The proof will follow from the results of the next two subsections. 
\begin{theorem}
\label{thm:quotSeqPEPS}
Let $M_{\textup{seq}}$ be the tensor network manifold associated with two-dimensional sequentially generated states with bond dimension $D$ and physical dimension $d \leq D^2$ on an $n \times (k+2)$ lattice. When endowed with its natural metric $g_{\textup{seq}}$, it defines a quotient tensor network manifold, $(B_{\textup{seq}}, h_{\textup{seq}})$ such that, for every $x, y$ on a open and dense subset of $M_{\textup{seq}}$, it holds that
\[
\exists \lambda \in \textup{U}(1) \text{ s.t. } \mathds{TN}_x = \lambda \mathds{TN}_y \iff [x]_{B_{\textup{seq}}} = [y]_{B_{\textup{seq}}}.
\]
\end{theorem}

\subsection{Characterising the gauge}

Let us first characterise the gauge of sequentially generated states. 

\begin{theorem}
\label{thmgauge3}
Consider the tensor network family of two-dimensional, sequentially generated states with bond dimension $D$ and physical dimension $d$ on an $n \times (k+2)$ lattice. Assume that $D^2 \geq d$. Consider two tensor networks of the family representing the states $\ket{U}$ and $\ket{V}$. Assume that every isometry of the tensor network associated with $\ket{U}$ is surjective when seen as a map from $\bb{C}^D \times \bb{C}^D$ to $\bb{C}^d$ and when seen as a map from $\bb{C}^D \times \bb{C}^d$ to $\bb{C}^D$,
\begin{equation}
\label{eq:fullrankseqpeps}
\begin{tikzpicture}[scale=1]
    \draw[shift={(0, 1.223)}, scale=0.5] (0, 0) -- (1.27, 0) -- (2.258, 1.834) -- (0, 1.834) -- cycle;
    \draw[shift={(0.844, 2.139)}, scale=0.5, <-] (0, 0) -- (0, 1.552);
    \draw[->] (0.324, 0.447) -- (0.324, 1.223);
    \draw[shift={(0.282, 2.14)}, scale=0.5, ->] (0, 0) -- (0, 1.552);
    \draw[shift={(1.87, 1.223)}, scale=0.5] (0, 0) -- (1.27, 0) -- (2.258, 1.834) -- (0, 1.834) -- cycle;
    \draw[shift={(2.714, 2.139)}, scale=0.5, ->] (0, 0) -- (0, 1.552);
    \draw[->] (2.193, 0.447) -- (2.193, 1.223);
    \draw[shift={(2.152, 2.14)}, scale=0.5, <-] (0, 0) -- (0, 1.552);
    \draw[shift={(3.669, 1.223)}, scale=0.5] (0, 0) -- (1.27, 0) -- (2.258, 1.834) -- (0, 1.834) -- cycle;
    \draw[shift={(4.513, 2.139)}, scale=0.5, <-] (0, 0) -- (0, 1.552);
    \draw[<-] (3.992, 0.447) -- (3.992, 1.223);
    \draw[shift={(3.951, 2.14)}, scale=0.5, <-] (0, 0) -- (0, 1.552);
    \node[anchor=center] at (0.32, 0.148) {$\mathbb{C}^D$};
    \node[anchor=center] at (0.92, 3.217) {$\mathbb{C}^D$};
    \node[anchor=center] at (2.19, 0.148) {$\mathbb{C}^D$};
    \node[anchor=center] at (3.989, 0.148) {$\mathbb{C}^D$};
    \node[anchor=center] at (2.79, 3.217) {$\mathbb{C}^D$};
    \node[anchor=center] at (4.589, 3.217) {$\mathbb{C}^D$};
    \node[anchor=center] at (0.285, 3.217) {$\mathbb{C}^d$};
    \node[anchor=center] at (2.155, 3.217) {$\mathbb{C}^d$};
    \node[anchor=center] at (3.954, 3.217) {$\mathbb{C}^d$};
\end{tikzpicture}.
\end{equation}
Further assume that the bottom state of every MPS is full-rank when seen as a map from $\bb{C}^D$ to $\bb{C}^D$. Let $\theta \in [0, 2\pi)$ be such that $\ket{U} = e^{i\theta}\ket{V}$. Then, the tensors of the tensor networks associated with $\ket{U}$ and $\ket{V}$ are related, up to a phase, by the multiplication by unitaries and their inverses over the contracting edges. 
\end{theorem}

\begin{proof}[Proof of \cref{thmgauge3}]
Following an analogous approach to that used for matrix product states, we analyse the gauge freedom of two-dimensional sequentially generated states by induction in the number of rows. This way, consider two two-dimensional sequentially generated states on a $2 \times (k+1)$ lattice representing the same state up to a phase, 
\begin{equation}
\label{eq:initialequalityseq}
\input{TikzFigures/peps3.tikz},
\end{equation}
where the choice of the MPS having four sites if merely for the ease of understanding the figures. If we multiply each column $i$ in both sides by $A_i^\dagger$ we obtain
\[
\input{TikzFigures/peps4.tikz}.
\]
Grouping the legs in each row except for the leftmost one in this equation we obtain
\begin{equation}
\label{eq:equalitygrouping}
\input{TikzFigures/peps5.tikz},
\end{equation}
where 
\[
C_1 := A_1^\dagger B_1.
\]
In particular, using the full-rank assumption of the isometries $V_1, \dotsc, V_{2k}$, \cref{eq:equalitygrouping} implies that for every $\ket{\phi}, \ket{\varphi} \in \bb{C}^{D}$ there exist $\ket{\alpha}, \ket{\beta} \in \bb{C}^{D}$ such that 
\[
C_1(\ket{\phi} \otimes \ket{\varphi}) = \ket{\alpha}\otimes \ket{\beta}, 
\]
i.e. $C_1$ maps products to products. By \cref{productstoproducts} we can conclude that there exist two unitary matrices $X_1, X_2 \in \textup{U}(D)$ such that
\[
C_1 = P^1_{\pi}(X_1 \otimes X_2),
\]
where $P^1_{\pi}$ is either a SWAP gate or the identity. Next, for any $i \in \{1, \dotsc, k\}$, let us choose the $i$-th isometry of the MPSs. In both rows, we group the legs that are at the right of the rightmost leg of the isometry, and the legs at the left of the bottom leg of the isometry; we obtain 
\begin{equation}
\label{eq:equalitygroupingany}
\input{TikzFigures/pepsidk.tikz},    
\end{equation}
where $C_{i+1} := A^\dagger_{i+1}B_{i+1}$. Using the full-rank assumption of the isometries shown in \cref{eq:fullrankseqpeps}, \cref{eq:equalitygroupingany} implies that for every $\ket{\phi}, \ket{\varphi} \in \bb{C}^{D}$ there exist $\ket{\alpha}, \ket{\beta} \in \bb{C}^{D}$ such that 
\[
C_{i+1}(\ket{\phi} \otimes \ket{\varphi}) = \ket{\alpha}\otimes \ket{\beta}, 
\]
and so there exist two unitaries $X_{2i+1}, X_{2(i+1)}$ such that $C_{i+1} = P^{i+1}_\pi(X_{2i+1} \otimes X_{2(i+1)})$, where $P^{i+1}_\pi$ is either the SWAP gate or the identity. This allows us to rewrite \cref{eq:initialequalityseq} as
\begin{equation}
\label{eq:identitywithpermutations}
\input{TikzFigures/peps10.tikz}.
\end{equation}
If we now take the partial trace over all legs on both sides except for the leftmost leg in each row, we see that
\[
\input{TikzFigures/lastjust1.tikz}.
\]
Thus, by \cref{purificationsym} there exists a unitary matrix $Y_1 \in \textup{U}(D^2)$ such that 
\[
\input{TikzFigures/lastjust2.tikz}
\]
Inserting this identity into \cref{eq:identitywithpermutations} and multiplying by $(P^1_{\pi})^{-1}$, $X_1^\dagger, X_2^\dagger$, and $\Sigma_1^{-1}$, $\Sigma_2^{-1}$ we conclude that
\begin{equation}
\label{eq:identitywithpermutations2}
\input{TikzFigures/lastjust3.tikz}
\end{equation}
Again, taking the partial trace over all legs on both sides except for the two leftmost legs in each row, and applying \cref{purificationsym} we conclude that there exists some unitary matrix $Y_2 \in \textup{U}(D^2)$ such that 
\[
\input{TikzFigures/lastjust4.tikz}
\]
Inserting this into \cref{eq:identitywithpermutations2}, and using the fact that $V_1$ and $V_2$ are full-rank when seen as linear maps from the Hilbert spaces corresponding to their two leftmost legs to that associated with their rightmost leg, we conclude that 
\[
\input{TikzFigures/lastjust5.tikz}
\]
Following the same reasoning, we conclude that for every $i \in \{1, \dotsc, k-1\}$ there exists some unitary matrix $Y_{i+1} \in \textup{U}(D^2)$ such that 
\begin{equation}
\label{eq:identityintermediateisoms}
\input{TikzFigures/lastjust6.tikz}
\end{equation}
and that for $i = k$, 
\[
\input{TikzFigures/lastjust7.tikz}
\]
Looking at this last identity, since $P^{k+1}_\pi$ is either the SWAP gate or the identity, we know that it must be the case that either 
\[
\input{TikzFigures/lastjust8.tikz}
\]
In either case by the full-rank assumption on the isometries, this means that $Y_k$ maps product states to product states and so by \cref{productstoproducts} it must be the composition of a product unitary with either a SWAP gate or the identity. This means that either 
\begin{equation}
\label{eq:possiblescenarios}
\input{TikzFigures/lastjust9.tikz}
\end{equation}
Note that the last three scenarios cannot happen. Indeed, let $v_1, v_2$ and $v_3$ be three vectors such that 
\[
\begin{tikzpicture}[scale=1.1]
    \filldraw[shift={(0.372, 0.66)}, scale=1.323, fill=pink] (0, 0) rectangle (0.141, 0.141);
    \filldraw[shift={(0.559, 0.847)}, scale=1.323, fill=violet] (0, 0) -- (-0.071, 0.071) -- (-0.212, 0.071) -- (-0.141, 0) -- cycle;
    \filldraw[shift={(0.372, 0.847)}, scale=1.323, fill=violet] (0, 0) -- (-0.071, 0.071) -- (-0.071, -0.071) -- (0, -0.141) -- (0, -0.141) -- cycle;
    \node[anchor=center] at (2.327, 1.674) {$\neq$};
    \filldraw[shift={(0.327, 1.458)}, xscale=1.059, yscale=1.234, fill=lightgray] (0, 0) -- (0.847, 0) -- (1.27, 0.282) -- (0, 0.282) -- (0, 0.282) -- cycle;
    \filldraw[shift={(0.327, 1.807)}, xscale=1.059, yscale=1.234, fill=darkgray] (0, 0) -- (-0.141, 0.141) -- (-0.141, -0.141) -- (0, -0.282) -- (0, 0) -- cycle;
    \filldraw[shift={(0.177, 1.981)}, xscale=1.059, yscale=1.234, fill=darkgray] (0, 0) -- (1.27, 0) -- (1.411, -0.141) -- (0.141, -0.141) -- (0.141, -0.141) -- cycle;
    \filldraw[shift={(0.426, 1.876)}, scale=1.222, heavier, fill=lightgray] (0, 0) -- (-0.001, 0.451);
    \filldraw[shift={(1.298, 1.876)}, xscale=0.882, yscale=0.741, heavier, fill=lightgray] (0, 0) -- (0, 0.423);
    \filldraw[shift={(0.426, 1.461)}, xscale=0.837, yscale=0.726, heavier, fill=lightyellow] (0, 0) -- (-0.004, -0.737) -- (-0.004, -0.729);
    \node[anchor=center] at (0.882, 1.634) {$V_{2k}$};
    \filldraw[shift={(0.422, 0.346)}, xscale=0.882, yscale=0.741, heavier, fill=lightgray] (0, 0) -- (0, 0.423);
    \node[circle, fill=black, inner sep=0pt, minimum size=4pt] at (0.422, 0.346) {};
    \node[circle, fill=black, inner sep=0pt, minimum size=4pt] at (0.425, 2.427) {};
    \node[circle, fill=black, inner sep=0pt, minimum size=4pt] at (1.298, 2.19) {};
    \node[anchor=center] at (0.164, 2.553) {$v_1$};
    \node[anchor=center] at (1.091, 2.392) {$v_2$};
    \node[anchor=center] at (0.424, 0.102) {$v_3$};
    \node[anchor=center] at (2.954, 1.674) {$0$};
\end{tikzpicture},
\]
which exist due to the full-rank assumption on $V_{2k}$. Then, if any of the three last scenarios shown in \cref{eq:possiblescenarios} were possible, we could use the vectors $v_1,v_2$ and $v_3$ to conclude that
\[
\input{TikzFigures/lastjust10.tikz}
\]
where we have used arrows to explicit that we are seeing each tensor network as a linear map from the Hilbert spaces associated with the inbound edges to that associated with the outbound edge. Whilst the right-hand side of the above identities is always a surjective linear map by assumption, the left-hand side cannot be, yielding a contradiction. Thus, it must be the case that $Y_k = Y_k^1 \otimes Y_k^2$ for some unitary matrices $Y_k^1, Y_k^2 \in \textup{U}(D)$, and that $P^{k+1}_\pi = \mathds{1}$. 

Inserting this into \cref{eq:identityintermediateisoms} for $i = k-1$ we can continue the same reasoning to conclude that $P^k_\pi$ is also the identity matrix and that $Y_{k-1}$ is of the form $Y_{k-1}^1 \otimes Y^2_{k-1}$ for some unitary matrices $Y_{k-1}^1,Y^2_{k-1} \in \textup{U}(D)$. In fact, one can follow the same argument to conclude that $P^i_\pi = \mathds{1}$ for every $i \in \{1, \dotsc, k+1\}$ and that $Y_i = Y_i^1 \otimes Y_i^2$ for some unitary matrices $Y_i^1, Y^2_i \in \textup{U}(D)$, for every $i \in \{1, \dotsc, k\}$, finishing the proof. 

In the more general case, given two sequentially generated states on an $n \times (k+1)$ lattice with $n > 2$ rows, we can group the first $n-1$ rows and proceed by induction.
\end{proof}

\subsection{Quotient tensor network manifold structure}
\label{sectionfreeactiontwodseqstates}

Let us now describe the gauge identified in the previous subsection as a group action. Recall that the tensor network manifold associated with the tensor network family of two-dimensional, sequentially generated states with bond dimension $D$ and physical dimension $d$ on an $n \times (k+2)$ lattice is
\[
M_{\textup{seq}} = (\bb{S}^{2D^2-1})^{\times n} \times \textup{V}_{D}(\bb{C}^{dD})^{\times nk} \times \textup{U}(D^2)^{\times n-1} \times \textup{U}(d^2)^{\times k(n-1)}.
\]
To describe the tensor networks of this family up to a phase, we can consider the manifold
\[
N_{\textup{seq}}  := (\bb{CP}^{D^2-1})^{\times n} \times \textup{PV}_{D}(\bb{C}^{dD})^{\times nk} \times \textup{PU}(D^2)^{\times n-1} \times \textup{PU}(d^2)^{\times k(n-1)},
\]
endowed with its \textit{natural metric} $\tilde g_{\textup{seq}} := g_{\textup{FS}}^{\oplus n} \oplus g_{\textup{pst}}^{\oplus nk} \oplus g_{\textup{pu}}^{\oplus (k+1)(n-1)}$.

\begin{definition}[Gauge action on sequentially generated PEPS]
\label{def:gaugeactionseqPEPS}
We define the action of $\textup{PU}(D)^{2(n-1) + kn} \times \textup{PU}(d)^{2k(n-1)}$ on $N_{\textup{seq}}$ as
\begin{align*}
\textup{PU}(D)^{2(n-1) + kn} \times \textup{PU}(d)^{2k(n-1)} \times N_{\textup{seq}} \to N_{\textup{seq}},\quad 
(X, (U, V, W)) \mapsto (XU, XV, XW),
\end{align*}
where $X \in \textup{PU}(D)^{2(n-1) + kn} \times \textup{PU}(d)^{2k(n-1)}$, and 
\begin{align*}
U &:= ([U_1], \dotsc, [U_n]) \in (\bb{CP}^{D^2-1})^{\times n},
\\V &:= (V^1, \dotsc, V^n) \in (\textup{PV}_{D}(\bb{C}^{dD})^{\times k})^{\times n},
\\W &:= (W^1, \dotsc, W^{k+1}) \in \textup{PU}(D^2)^{\times n-1} \times (\textup{PU}(d^2)^{\times (n-1)})^{\times k},
\end{align*}
where the tuple $U$ represents the bottom states of the MPSs composing the state, each $W^i$ denotes a column of unitaries and each $V^i$ denotes the set of isometries of an MPS of the state. For this reason $W$ has $(k+1)$ components and $V$ has $n$ components. Each element $W^i$ and $V^i$ is defined as 
\begin{align*}
V^i &:= ([V^i_1], \dotsc, [V^i_k]) \in \textup{PV}_{D}(\bb{C}^{dD})^{\times k},\\
W^1 &:= ([W^i_1], \dotsc, [W^i_{n-1}]) \in \textup{PU}(D^2)^{\times (n-1)},\\
W^i &:= ([W^i_1], \dotsc, [W^i_{n-1}]) \in \textup{PU}(d^2)^{\times (n-1)}, \quad \forall i \in \{2, \dotsc, k+1\}. 
\end{align*}
Furthermore, for readability, we will write
\[
X = (X^1, \dotsc, X^n, \tilde{X}^1, \dotsc, \tilde{X}^{k+1}),
\]
where
\begin{align*}
X^i &= ([X^i_1], \dotsc, [X^i_k]) \in \textup{PU}(D)^{\times k}, \quad \forall i \in \{1, \dotsc, n\},
\\\tilde{X}^1 &= ([\tilde{X}^i_1], \dotsc, [\tilde{X}^i_{2(n-1)}]) \in \textup{PU}(D)^{\times 2(n-1)},
\\\tilde{X}^i &= ([\tilde{X}^i_1], \dotsc, [\tilde{X}^i_{2(n-1)}]) \in \textup{PU}(d)^{\times 2(n-1)}, \quad \forall i \in \{2, \dotsc, k+1\}.
\end{align*}
We defined $X$ so that each $\tilde{X}^i$ acts on the $i$-th column of unitaries, represented by $W^i$. Furthermore, each $X^i$ acts on the MPS in the $i$-th row, represented by $[U_i]$ and $V^i$. This way, the action on the unitaries can be defined as
\[
XW := (\tilde{X}^1 W^1, \tilde{X}^2 W^2, \dotsc, \tilde{X}^{k+1} W^{k+1}),
\]
where for each column $\tilde{X}^i$ acts naturally (see \cref{fig:actioncolums}),
\begin{multline}
\label{eq:actioncolumns}
\tilde{X}^i W^i := ([(\mathds{1} \otimes \tilde{X}^i_3) W^i_1 ((\tilde{X}^i_1)^\dagger \otimes (\tilde{X}^i_2)^\dagger)], \dotsc,\\ [(\mathds{1} \otimes \tilde{X}^i_{2(n-2) + 1}) W^i_{n-2} ((\tilde{X}^i_{2(n-2)-1})^\dagger \otimes (\tilde{X}^i_{2(n-2)})^\dagger)], [W^i_{n-1} ((\tilde{X}^i_{2(n-2) + 1})^\dagger \otimes (\tilde{X}^i_{2(n-1)})^\dagger)]).
\end{multline}
\begin{figure}
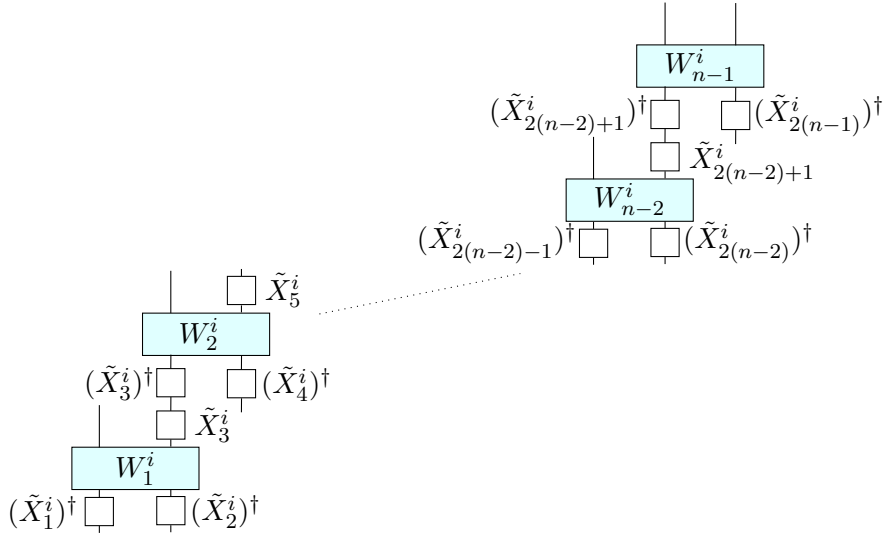

    \centering
    \ctikzfig{TikzFigures/actioncolumns}
    \caption{Visual representation of the action described in \cref{eq:actioncolumns} for the $i$-th column of unitaries.}
    \label{fig:actioncolums}
\end{figure}

The action on $U$ can be easily described as well (see \cref{fig:actionspheres}),
\begin{equation}
\label{eq:actionspheres}
XU := ([(\tilde{X}^1_{f(1)} \otimes X^1_1)U_1], \dotsc, [(\tilde{X}^1_{f(n)} \otimes X^{n}_1) U_n]),
\end{equation}
where 
\[
f(i) = \begin{cases}
1\quad &\text{if}\ i = 1,\\
2(i-1)\quad &\text{if}\ i > 1.
\end{cases}
\]
\begin{figure}
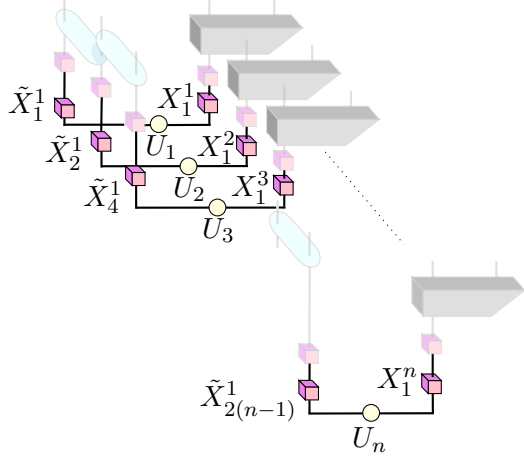

    \centering
    \ctikzfig{TikzFigures/actionspheres}
    \caption{Visual representation of the action described in \cref{eq:actionspheres}. Some faded tensors have been added for reference.}
    \label{fig:actionspheres}
\end{figure}
Note that the elements of $X$ which act on $U$ correspond to $X^i_1$, $i \in \{1, \dotsc, n\}$, which corresponds to the first element of $X$ acting on the MPS in the $i$-th row, and some elements of $\tilde{X}^1$, which also act on the first column of unitaries. 

Lastly, we define
\[
XV := (\tilde{V}^1, \tilde{V}^2, \dotsc,\tilde{V}^n),
\]
where
\begin{multline}
\label{accionV}
\tilde{V}^i := ([(\tilde{X}^2_{f(i)} \otimes X^i_2)V^i_1 (X^i_1)^\dagger],\\ [(\tilde{X}^3_{f(i)} \otimes X^i_3)V^i_2 (X^i_2)^\dagger], \dotsc, [(\tilde{X}^k_{f(i)} \otimes X^i_k) V^i_{k-1} (X^i_{k-1})^\dagger],\\ [(\tilde{X}^{k+1}_{f(i)} \otimes \mathds{1})V^i_k (X^i_k)^\dagger]),
\end{multline}
for every $i = 1, \dotsc, n$. Note that both $X^i$ and the $f(i)$-th coordinate of $\tilde X^2, \dotsc, \tilde X^{k+1}$ act on $V^i$ (see \cref{fig:actionMPS}).
\begin{figure}
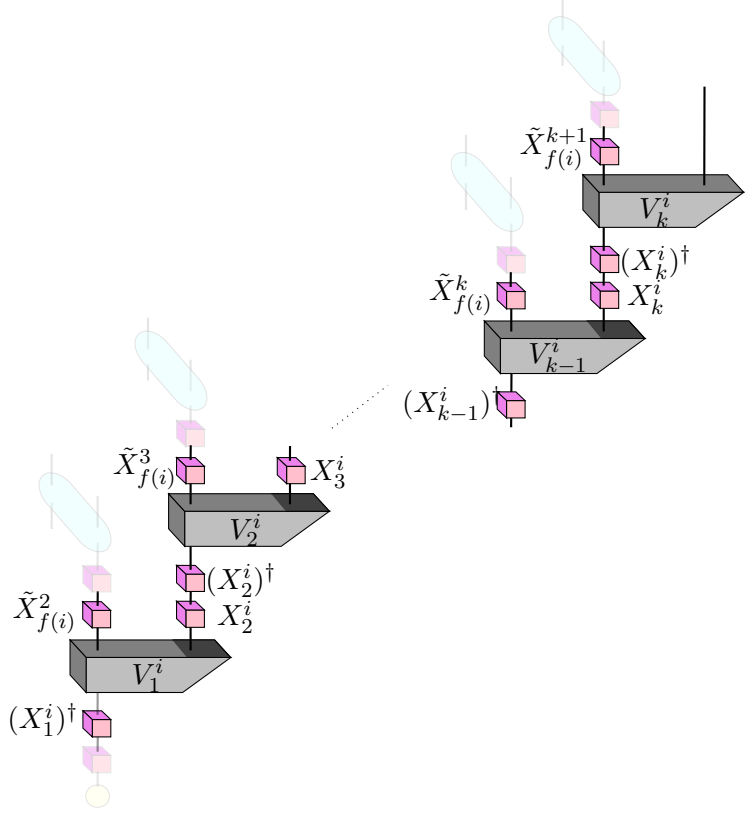

    \centering
    \ctikzfig{TikzFigures/actionMPS}
    \caption{Visual representation of the action described in \cref{accionV}, for the $i$-th MPS row. Some faded tensors have been added for reference.}
    \label{fig:actionMPS}
\end{figure}
\end{definition}

As usual, the above action allows us to define a quotient tensor network manifold such that every point in it is in one-to-one correspondence with a two-dimensional sequentially generated state up to a phase.

\begin{figure}
    \centering
    \input{TikzFigures/diag6.tikz}
    \caption{Visual representation of the setting considered in \cref{prop:theoremaisometric}.}
    \label{fig:diag6}
\end{figure}

\begin{proposition}
\label{prop:theoremaisometric}
Let $N_{\textup{seq}}$ be endowed with its natural metric $\tilde g_{\textup{seq}}$. Consider the action of $\textup{PU}(D)^{2(n-1) + kn} \times \textup{PU}(d)^{2k(n-1)}$ on $N_{\textup{seq}}$ as given in \cref{def:gaugeactionseqPEPS}. There is a smooth manifold structure and a Riemannian metric $h_{\textup{seq}}$ on $B_{\textup{seq}} := N_{\textup{seq}} / \textup{PU}(D)^{2(n-1) + kn} \times \textup{PU}(d)^{2k(n-1)}$ for which the projection $(N_{\textup{seq}}, \tilde g_{\textup{seq}}) \to (B_{\textup{seq}}, h_{\textup{seq}})$ is a Riemannian submersion. 
\end{proposition}
\begin{proof}
Again, this action is smooth and isometric. In order to show that it is free, consider some $X \in \textup{PU}(D)^{2(n-1) + kn} \times \textup{PU}(d)^{2k(n-1)}$ and assume that there exists a tuple $(U, V, W) \in N_{\textup{seq}}$ satisfying
\[
(U, V, W) = (XU, XV, XW).
\]
In particular, for every $i \in \{1, \dotsc, k+1\}$
\begin{multline}
\label{igualdadclases}
([W^i_1], \dotsc, [W^i_{n-1}]) = ([(\mathds{1} \otimes \tilde{X}^i_3) W^i_1 ((\tilde{X}^i_1)^\dagger \otimes (\tilde{X}^i_2)^\dagger)], \dotsc,\\ [(\mathds{1} \otimes \tilde{X}^i_{2(n-2) + 1}) W^i_{n-2} ((\tilde{X}^i_{2(n-2)-1})^\dagger \otimes (\tilde{X}^i_{2(n-2)})^\dagger)], [W^i_{n-1} ((\tilde{X}^i_{2(n-2) + 1})^\dagger \otimes (\tilde{X}^i_{2(n-1)})^\dagger)]).
\end{multline}
Thus, looking at the rightmost coordinate, the above equation implies that, for every $i \in \{1, \dotsc, k+1\}$
\[
[W^i_{n-1}] = [W^i_{n-1} ((\tilde{X}^i_{2(n-2) + 1})^\dagger \otimes (\tilde{X}^i_{2(n-1)})^\dagger)],
\]
which can be rewritten as 
\[
[\mathds{1}] = [((\tilde{X}^i_{2(n-2) + 1})^\dagger \otimes (\tilde{X}^i_{2(n-1)})^\dagger)].
\]
Inserting this into \cref{igualdadclases} and looking at the next rightmost coordinate, we can continue the same reasoning to conclude that $[\tilde{X}^i_j] = [\mathds{1}]$ for every $i \in \{1, \dotsc, k+1\}$, $j \in \{1, \dotsc, n-1\}$. Using these equalities and the assumption that $XV = V$, we now know that
\[
([V^i_1], \dotsc, [V^i_k]) = ([(\mathds{1} \otimes X^i_2)V^i_1 (X^i_1)^\dagger],\dotsc, [(\mathds{1} \otimes X^i_k) V^i_{k-1} (X^i_{k-1})^\dagger], [V^i_k (X^i_k)^\dagger]),
\]
for every $i \in \{1, \dotsc, n\}$. Following the same ideas as in the proof of \cref{thmfreeactionMPS}, we can conclude that $X^i_j = [\mathds{1}]$ for every $i \in \{1, \dotsc, n\}$, $j \in \{1, \dotsc, k\}$, finishing the proof.
\end{proof}

\section{Isometric PEPS}
\label{sec:IsomPEPS}

Lastly, let us study the tensor network family of isometric PEPS \cite{wei2022isometric} on an $L \times (L+1)$ lattice with physical dimension $d$ and bond dimension $D$. Every tensor network of this family is constructed by sequentially applying isometries on a two-dimensional lattice. The order in which the isometries are applied is fixed, and corresponds to the radial preparation considered in \cite{wei2022isometric}. For example, the tensor network family of isometric PEPS on a $4 \times 5$ lattice consists of those tensor networks of the form
\begin{equation}
\label{eq:IsomPEPS}
\input{TikzFigures/isompeps.tikz},
\end{equation}
where the different tensors have been given numeric labels so that they are clearly identified. 

Note that, as in the case of sequentially generated states and MPS states, we also consider open boundary conditions and promote the boundary legs of the state to physical legs. The tensor network manifold associated with this family is 
\begin{align*}
M_{\textup{iso}} &:= \bb{S}^{2D^4-1} \times \textup{V}_D(\bb{C}^{D^3})^{\times 2L-3} \times \textup{V}_{D^2}(\bb{C}^{dD^2})^{\times (L-2)^2} 
\\&\quad \times \textup{V}_{D^2}(\bb{C}^{d^2D})^{L-2} \times\textup{V}_D(\bb{C}^{dD})  \times \textup{U}(D^2)^{\times L-1},
\end{align*}
which we endow with its \textit{natural metric} $g_{\textup{iso}} := g_{\textup{round}} \oplus g_{st}^{\oplus L(L-1)} \oplus g_{\textup{bi}}^{\oplus L-1}$.

\begin{theorem}
\label{thm:mainresultIsomPeps}
Let $M_{\textup{iso}}$ be the tensor network manifold associated with isometric PEPS on a $L \times (L+1)$ lattice with bond dimension $D$ and physical dimension $d$. When endowed with its natural metric $g_{\textup{iso}}$, it defines a quotient tensor network manifold $(B_{\textup{iso}}, h_{\textup{iso}})$ such that, for every $x, y$ on an open and dense subset of $M_{\textup{iso}}$, it holds that 
\[
\exists \lambda \in \textup{U}(1)\ \text{s.t.}\ \mathds{TN}_x = \lambda \mathds{TN}_y \iff [x]_{B_{\textup{iso}}} = [x]_{B_{\textup{iso}}}. 
\]
\end{theorem}
As in all previous cases, the proof follows from the following two subsections.
\subsection{Characterising the gauge}

Let us first characterise the gauge freedom of the tensors in the family of isometric PEPS. 

\begin{notation}
We will often represent the tensors of a tensor network of the family as rectangles. To do so, we will always use the following convention:
\[
\begin{tikzpicture}[scale=1.2]
    \filldraw[fill=Cyan!40] (0.499, 0.302) -- (0.504, 0.706);
    \filldraw[fill=Cyan!40] (1.078, 0.715) -- (1.083, 1.12);
    \filldraw[fill=Cyan!40] (0, 1.476) -- (0.284, 1.713) -- (1.413, 1.139) -- (0.568, 0.43) -- (0.179, 0.624) -- (0.75, 1.099) -- (0, 1.476);
    \filldraw[fill=Cyan!40] (1.08, 1.09) -- (1.084, 1.495);
    \filldraw[fill=Cyan!40] (0.25, 1.489) -- (0.255, 1.894);
    \filldraw[fill=Cyan!40] (0.505, 0.611) -- (0.511, 1.259);
    \node[anchor=center] at (0.51, 0.116) {$i$};
    \node[anchor=center] at (1.074, 0.469) {$j$};
    \node[anchor=center] at (0.51, 1.421) {$k$};
    \node[anchor=center] at (1.109, 1.703) {$l$};
    \node[anchor=center] at (0.263, 2.021) {$m$};
    \node[anchor=center] at (1.786, 1.142) {$\equiv$};
    \draw (2.188, 1.423) rectangle (3.881, 0.858);
    \draw (2.47, 1.705) -- (2.47, 1.423);
    \draw (3.034, 1.705) -- (3.034, 1.423);
    \draw (2.47, 0.858) -- (2.47, 0.435);
    \draw (3.599, 1.705) -- (3.599, 1.423);
    \draw (3.043, 0.862) -- (3.043, 0.438);
    \node[anchor=center] at (2.475, 0.259) {$i$};
    \node[anchor=center] at (3.052, 0.249) {$j$};
    \node[anchor=center] at (2.472, 1.905) {$k$};
    \node[anchor=center] at (3.026, 1.904) {$l$};
    \node[anchor=center] at (3.609, 1.869) {$m$};
\end{tikzpicture},
\]
where we labelled the legs to make the correspondence between the two representations more explicit. 
\end{notation}

\begin{theorem}
\label{thm:gaugeIsom}
Consider the tensor network family of isometric PEPS on an $L \times (L+1)$ lattice, as shown in \cref{eq:IsomPEPS}. Consider two tensor networks of the family representing the states $\ket{U}$ and $\ket{V}$. Assume that the tensors of the tensor network associated with $\ket{U}$ satisfy the following:
\begin{enumerate}
    \item The state at the bottom left of the tensor network (with number $1$ in \cref{eq:IsomPEPS}) is full-rank when seen as a linear map $\bb{C}^{D^2} \to \bb{C}^{D} \times \bb{C}^{D}$, 
    \begin{equation}
    \label{eq:fullrankbottomstate}
    \begin{tikzpicture}[scale=1.5]
        \draw (0, 1.129) rectangle (1.693, 0.564);
        \draw [->] (0.282, 1.411) -- (0.282, 1.129);
        \draw [<-] (0.847, 1.411) -- (0.847, 1.129);
        \draw [<-] (1.411, 1.411) -- (1.411, 1.129);
        \node[anchor=center] at (0.282, 1.693) {$\mathbb{C}^{D^2}$};
        \node[anchor=center] at (0.847, 1.693) {$\mathbb{C}^D$};
        \node[anchor=center] at (1.411, 1.693) {$\mathbb{C}^D$};
    \end{tikzpicture}.
    \end{equation}

    \item The isometries in the first row of the tensor network (with number $2$ in \cref{eq:IsomPEPS}) are full-rank when seen as linear maps $\bb{C}^{D} \times \bb{C}^{D} \to \bb{C}^{D} \times \bb{C}^{D}$ from the spaces corresponding to their leftmost legs to those corresponding to their rightmost ones, 
    \begin{equation}
    \label{eq:fullrankfirstrow}
    \begin{tikzpicture}[scale=1.5]
        \draw (0, 1.278) rectangle (1.693, 0.714);
        \draw [->] (0.282, 1.56) -- (0.282, 1.278);
        \draw [<-] (0.847, 1.56) -- (0.847, 1.278);
        \draw [<-] (1.411, 1.56) -- (1.411, 1.278);
        \node[anchor=center] at (0.282, 0.149) {$\mathbb{C}^D$};
        \node[anchor=center] at (0.847, 1.843) {$\mathbb{C}^D$};
        \node[anchor=center] at (1.411, 1.843) {$\mathbb{C}^D$};
        \node[anchor=center] at (0.282, 1.843) {$\mathbb{C}^D$};
        \draw [<-] (0.282, 0.714) -- (0.282, 0.431);
    \end{tikzpicture}.
    \end{equation}

    \item The isometries in the first column of the tensor network (with number $2$ in \cref{eq:IsomPEPS}) are full-rank when seen as linear maps $\bb{C}^{D} \times \bb{C}^{D} \to \bb{C}^{D} \times \bb{C}^{D}$ from the spaces corresponding to their leftmost and bottom legs to those corresponding to their rightmost ones,
    \begin{equation}
    \label{eq:fullrankfirstcolumn}
    \begin{tikzpicture}[scale=1.5]
        \draw (0, 1.278) rectangle (1.693, 0.714);
        \draw [<-] (0.847, 0.714) -- (0.847, 0.431);
        \node[anchor=center] at (0.847, 0.149) {$\mathbb{C}^D$};
        \node[anchor=center] at (0.847, 1.843) {$\mathbb{C}^D$};
        \node[anchor=center] at (1.411, 1.843) {$\mathbb{C}^D$};
        \node[anchor=center] at (0.282, 1.843) {$\mathbb{C}^D$};
        \draw [->] (0.282, 1.56) -- (0.282, 1.278);
        \draw [<-] (0.847, 1.56) -- (0.847, 1.278);
        \draw [<-] (1.411, 1.56) -- (1.411, 1.278);
    \end{tikzpicture}.
    \end{equation}

    \item The bulk isometries and those in the last column of the tensor network (with numbers $3$ and $4$ in \cref{eq:IsomPEPS}) are full-rank when seen as linear maps $\bb{C}^{d} \times \bb{C}^{D} \times \bb{C}^{D} \to \bb{C}^{D} \times \bb{C}^{D}$ from the spaces corresponding to their leftmost and bottom legs, to those corresponding to their rightmost ones,  
    \begin{equation}
    \label{eq:fullrankbulk}
    \begin{tikzpicture}[scale=1.5]
        \draw (0, 1.278) rectangle (1.693, 0.714);
        \draw [<-] (0.847, 0.714) -- (0.847, 0.431);
        \node[anchor=center] at (0.847, 0.149) {$\mathbb{C}^D$};
        \node[anchor=center] at (0.847, 1.843) {$\mathbb{C}^D$};
        \node[anchor=center] at (1.411, 1.843) {$\mathbb{C}^D$};
        \node[anchor=center] at (0.282, 1.843) {$\mathbb{C}^d$};
        \draw [->] (0.282, 1.56) -- (0.282, 1.278);
        \draw [<-] (0.847, 1.56) -- (0.847, 1.278);
        \draw [<-] (1.411, 1.56) -- (1.411, 1.278);
        \draw [<-] (0.282, 0.714) -- (0.282, 0.431);
        \node[anchor=center] at (0.282, 0.149) {$\mathbb{C}^D$};
    \end{tikzpicture}.
    \end{equation}

    \item The leftmost isometry in the top row of the tensor network (with number $5$ in \cref{eq:IsomPEPS}) is full-rank when seen as a linear map $\bb{C}^{d} \times \bb{C}^{D} \to \bb{C}^{D}$ from the spaces corresponding to its leftmost and bottom legs, to that corresponding to its rightmost leg, 
    \begin{equation}
    \label{eq:fullranktopleft}
    \begin{tikzpicture}[scale=1.5]
        \draw (0, 1.277) rectangle (1.693, 0.712);
        \draw [->] (0.282, 1.559) -- (0.282, 1.277);
        \draw [<-] (1.411, 1.559) -- (1.411, 1.277);
        \node[anchor=center] at (1.411, 0.148) {$\mathbb{C}^D$};
        \node[anchor=center] at (1.411, 1.841) {$\mathbb{C}^D$};
        \node[anchor=center] at (0.282, 1.841) {$\mathbb{C}^d$};
        \draw [<-] (1.411, 0.712) -- (1.411, 0.43);
    \end{tikzpicture}.
    \end{equation}

    \item  The isometries $U_i$ in the bottom row and the rest of the isometries $V_i$ of the tensor network (except for the leftmost one in each row) are full-rank in the following sense
    \begin{equation}
    \label{eq:fullrankdensities}
    \begin{tikzpicture}[scale=1]
        \draw (0, 1.488) rectangle (1.693, 0.924);
        \draw (0, 3.605) rectangle (1.693, 3.041);
        \draw[<-] (0.282, 3.041) -- (0.282, 2.758);
        \draw[->] (0.282, 1.771) -- (0.282, 1.488);
        \draw (1.411, 3.041) -- (1.411, 1.488);
        \draw[<-] (0.282, 4.028) -- (0.282, 3.605);
        \draw[->] (0.282, 0.924) -- (0.282, 0.501);
        \draw (4.516, 1.488) rectangle (6.209, 0.924);
        \draw (4.516, 3.605) rectangle (6.209, 3.041);
        \draw[<-] (4.798, 3.041) -- (4.798, 2.758);
        \draw[->] (4.798, 1.771) -- (4.798, 1.488);
        \draw (5.927, 3.041) -- (5.927, 1.488);
        \draw[<-] (4.798, 4.028) -- (4.798, 3.605);
        \draw[->] (4.798, 0.924) -- (4.798, 0.501);
        \draw (0.847, 3.041) -- (0.847, 1.488);
        \draw (5.362, 3.041) -- (5.362, 1.488);
        \draw[->] (5.362, 4.028) -- (5.362, 3.605);
        \draw[<-] (5.362, 0.924) -- (5.362, 0.501);
        \node[anchor=center] at (0.353, 4.381) {$\mathbb{C}^D$};
        \node[anchor=center] at (0.353, 0.148) {$\mathbb{C}^D$};
        \node[anchor=center] at (0.353, 2.053) {$\mathbb{C}^D$};
        \node[anchor=center] at (0.353, 2.476) {$\mathbb{C}^D$};
        \node[anchor=center] at (4.868, 0.148) {$\mathbb{C}^D$};
        \node[anchor=center] at (5.433, 0.148) {$\mathbb{C}^D$};
        \node[anchor=center] at (4.868, 4.381) {$\mathbb{C}^D$};
        \node[anchor=center] at (5.433, 4.381) {$\mathbb{C}^D$};
        \node[anchor=center] at (4.868, 2.476) {$\mathbb{C}^d$};
        \node[anchor=center] at (4.868, 2.053) {$\mathbb{C}^d$};
        \node[anchor=center] at (3.104, 2.335) {and};
        \node[anchor=center] at (0.847, 1.206) {$U_i$};
        \node[anchor=center] at (0.847, 3.323) {$U_i^\dagger$};
        \node[anchor=center] at (5.362, 3.323) {$V_i^\dagger$};
        \node[anchor=center] at (5.362, 1.206) {$V_i$};
    \end{tikzpicture}
    \end{equation}
\end{enumerate}
Assume that there exists some $\theta \in [0, 2\pi)$ for which $\ket{U} = e^{i\theta} \ket{V}$. Then, the tensors of the tensor networks associated with $\ket{U}$ and $\ket{V}$ are related, up to a phase, by the multiplication by unitaries and their inverses over the contracting edges. 
\end{theorem}

Before we present the proof of \cref{thm:gaugeIsom}, let us state and prove an auxiliary result that will allow us to decompose some unitaries into a tensor product. This will be used frequently in the following pages.
\begin{lemma}
\label{lem:reducedmatrixdecomposition}
Let $U \in \textup{U}(n^2)$ be a unitary matrix and let $V, W$ be two linear maps from $\bb{C}^n$ to $\bb{C}^{n \times  k}$, 
\[
\begin{tikzpicture}[scale=1]
    \draw[<-] (0.282, 0.968) -- (0.282, 0.403);
    \draw (0, 1.532) rectangle (1.129, 0.968);
    \draw[<-] (0.282, 2.097) -- (0.282, 1.532);
    \draw[<-] (0.847, 2.097) -- (0.847, 1.532);
    \node[anchor=center] at (0.282, 0.121) {$\mathbb{C}^n$};
    \node[anchor=center] at (0.282, 2.379) {$\mathbb{C}^n$};
    \node[anchor=center] at (0.847, 2.379) {$\mathbb{C}^k$};
\end{tikzpicture}
\]
such that 
\begin{equation}
\label{igualdadisom}
\input{TikzFigures/lemmaisom1.tikz}.
\end{equation}
Assume that $V$ is such that for any matrix $M \in \textup{M}_n(\bb{C})$, there exists some $N \in \textup{M}_n(\bb{C})$ such that 
\begin{equation}
\label{eq:assumptionfullranklem}
\input{TikzFigures/lemmaisom8.tikz}.
\end{equation}
Then, $U$ is of the form $U = U_1 \otimes U_2$, for some $U_1, U_2 \in \textup{U}(n)$. 
\end{lemma}

\begin{remark}
Note that the assumption on $V$ written in \cref{eq:assumptionfullranklem} is equivalent to that shown in \cref{eq:fullrankdensities}. Also observe that the same result holds if we assume that \cref{eq:assumptionfullranklem} holds for $W$ instead of $V$. 
\end{remark}

\begin{proof}[Proof of \cref{lem:reducedmatrixdecomposition}]
First, inserting the assumption from \cref{eq:assumptionfullranklem} into \cref{igualdadisom} we conclude that for any $\ket{i},\ket{j}, \ket{k}, \ket{l} \in \bb{C}^n$, there exists some $X \in \textup{M}_n(\bb{C})$ such that
\begin{equation}
\label{igualdadisom2}
\input{TikzFigures/lemmaisom2.tikz}, 
\end{equation}
where $X$ depends on $\ket{i}$ and $\ket{j}$. Denoting 
\[
\input{TikzFigures/lemmaisom3.tikz},
\]
we can decompose $U$ as 
\[
U = \sum_{ij = 1}^n U_{ij} \otimes \ket{i}\bra{j}.    
\]
Now, \cref{igualdadisom2} implies that for every $i, j, k, l \in \{1, \dotsc, n\}$, there exists some $\alpha_{ijkl} \in \bb{C}$ such that 
\begin{equation}
\label{eq:relationUs}
U^\dagger_{ki} U_{lj} = \alpha_{ijkl}\mathds{1}. 
\end{equation}
In particular, for every $i, j \in \{1, \dotsc, n\}$, 
\[
U_{ij}^\dagger U_{ij} = \alpha_{jjii} \mathds{1}, 
\]
where $\alpha_{jjii} \geq 0$ as $U_{ij}^\dagger U_{ij}$ is Hermitian and positive-semidefinite. Therefore, for every $i, j \in \{1, \dotsc, n\}$, $U_{ij}$ is of the form 
\begin{equation}
\label{eq:expressionUunitaria}
U_{ij} = \sqrt{\alpha_{jjii}}\, V_{ij}, 
\end{equation}
for some unitary matrix $V_{ij} \in \textup{U}(n)$. Putting \cref{eq:relationUs,eq:expressionUunitaria} together we conclude that for every $i, j, k, l \in \{1, \dotsc, n\}$,
\[
\alpha_{ijkl} \mathds{1} = U_{ki}^\dagger U_{lj} = \sqrt{\alpha_{iikk} \alpha_{jjll}}\, V^\dagger_{ki} V_{lj}, 
\]
and so for every $i, j, k, l \in \{1, \dotsc, n\}$ it holds that
\[
\alpha_{ijkl} V_{ki} = \sqrt{\alpha_{iikk} \alpha_{jjll}}\, V_{lj},
\]
which means that the unitaries $V_{ij}$ only differ by a global phase, which allows us to rewrite \cref{eq:expressionUunitaria} as 
\[
U_{ij} = \beta_{ij}V, 
\]
for every $i, j \in \{1, \dotsc, n\}$, where $V \in \textup{U}(n)$ is some fixed unitary and $\beta_{ij} \in \bb{C}$, $|\beta_{ij}| = 1$. Thus,
\[
U = \sum_{ij = 1}^n U_{ij} \otimes \ket{i}\bra{j}=  V \otimes \sum_{ij = 1}^n \beta_{ij}\ket{i}\bra{j}.
\]
Lastly, using \cref{productounitarias} we conclude that $\sum_{ij = 1}^n \beta_{ij}\ket{i}\bra{j}$ is also unitary, finishing the proof. 
\end{proof}

\begin{proof}[Proof of \cref{thm:gaugeIsom}]
Consider two tensor networks of the family representing the same state up to a phase,
\begin{equation}
\label{eq:identitybottomrows}
\input{TikzFigures/isompeps1.tikz},
\end{equation}
where we only represented the bottom two rows of the tensor networks, as the proof will follow by a \textit{bottom-up} argument, in an analogous fashion to when we studied the gauge freedom of MPS states in \cref{gaugeMPSthm}. For simplicity, the lattice in the figures only has four columns. Nevertheless, the proof is completely general. We assume that the tensor network on the right-hand side satisfies the full-rank assumptions of the statement. 

Contracting both states against their adjoints over every leg except for the leftmost leg in the first row, we conclude that
\[
\input{TikzFigures/isompeps2.tikz}.
\]
This way, we can apply \cref{purificationsym} to conclude that there exists some unitary matrix $X \in \textup{U}(D^2)$ such that
\[
\input{TikzFigures/isompeps3.tikz}.
\]
Inserting this identity into \cref{eq:identitybottomrows} we obtain that
\begin{equation}
\label{eq:identitybottomrows2}
\input{TikzFigures/isompeps4.tikz}.
\end{equation}
Now, since we are assuming that $V_1$ is full-rank in the sense of \cref{eq:fullrankbottomstate}, we can rewrite \cref{eq:identitybottomrows2} as 
\begin{equation}
\label{eq:identitybottomrows3}
\input{TikzFigures/isompeps5.tikz}    
\end{equation}
Contracting both tensor networks against their adjoints over all legs except for the leftmost leg of $U_2$ and $V_2$ we obtain that
\begin{equation}
\label{eq:identitybottomrowsaux}
\input{TikzFigures/isompeps6.tikz},
\end{equation}
and so, from the full-rank assumption shown in \cref{eq:fullrankdensities} we can use \cref{lem:reducedmatrixdecomposition} to conclude that there exist two unitaries $X_1, X_2 \in \textup{U}(D)$ such that $X = X_1 \otimes X_2$, which allows us to rewrite \cref{eq:identitybottomrows3} as 
\begin{equation}
\label{eq:identitybottomrows4}
\input{TikzFigures/isompeps7.tikz}.
\end{equation}
and \cref{eq:identitybottomrowsaux} as
\[
\input{TikzFigures/isompeps8.tikz};
\]
thus, by \cref{purificationsym}, there exists some unitary matrix $Y \in \textup{U}(D^2)$ such that 
\[
\input{TikzFigures/isompeps9.tikz}.
\]
Since we are also assuming that $V_2$ is full-rank in the sense of \cref{eq:fullrankfirstrow}, we can rewrite \cref{eq:identitybottomrows4} as 
\[
\input{TikzFigures/isompeps10.tikz},
\]
which again, using the full-rank assumption from \cref{eq:fullrankdensities} and \cref{lem:reducedmatrixdecomposition} we know that $Y = Y_1 \otimes Y_2$ for some unitary matrices $Y_1, Y_2 \in \textup{U}(D)$. This allows us to conclude that 
\[
\input{TikzFigures/isompeps11.tikz}.
\]
Lastly, we contract both tensor networks against their adjoints leaving the two leftmost legs of $U_3$ and $V_3$ open, obtaining that
\[
\input{TikzFigures/isompeps12.tikz}
\]
which again by \cref{purificationsym} implies that there exists some unitary matrix $Z \in \textup{U}(D)$ such that
\[
\input{TikzFigures/isompeps13.tikz}
\]
and so, as we are assuming that $V_3$ is full-rank in the sense of \cref{eq:fullrankfirstrow} we know that
\[
\input{TikzFigures/isompeps14.tikz}.
\]
Note that in this expression, the bottom row of the tensor networks does no longer appear. This way, following an analogous reasoning in a bottom-up fashion, using the full-rank assumptions written in \cref{eq:fullrankfirstcolumn,eq:fullrankbulk,eq:fullrankdensities} we identify the same gauge freedom and eventually we reach the two upmost rows of the tensor networks, for which we know that 
\begin{equation}
\label{eq:topmostlayersisom}
\input{TikzFigures/isompeps15.tikz},
\end{equation}
for some unitary matrices $A_1, A_2, A_3 \in \textup{U}(D)$. We can proceed in the same manner as for the previous rows of the tensor networks, i.e. we contract both tensor networks against their adjoints leaving the leftmost leg of $W_1$ and $\tilde W_1$ open, obtaining
\[
\input{TikzFigures/isompeps16.tikz},
\]
and so by \cref{purificationsym} we know that there exists some unitary matrix $C \in \textup{U}(D^2)$ such that 
\[
\input{TikzFigures/isompeps17.tikz}.
\]
Again, as we are assuming that $\tilde W_1$ is full-rank in the sense of \cref{eq:fullrankfirstcolumn}, we can rewrite \cref{eq:topmostlayersisom} as
\begin{equation}
\label{eq:topmostlayersisom2}
\input{TikzFigures/isompeps18.tikz}.
\end{equation}
Contracting both tensor networks against their duals leaving the leftmost leg of $W_2$ and $\tilde W_2$ open, we conclude that
\[
\input{TikzFigures/isompeps19.tikz}
\]
and so we can use the full-rank assumption shown in \cref{eq:fullrankdensities}, and \cref{lem:reducedmatrixdecomposition}, to conclude that there exist two unitary matrices $C_1, C_2 \in \textup{U}(D)$ for which $C = C_1 \otimes C_2$, which allows us to rewrite \cref{eq:topmostlayersisom2} as
\[
\input{TikzFigures/isompeps20.tikz}.
\]
Following the same reasoning we reach the top row of the tensor networks, for which we know that
\begin{equation}
\label{eq:identitytoprow}
\input{TikzFigures/isompeps21.tikz},
\end{equation}
for some unitary matrices $D_1, E \in \textup{U}(D)$. Contracting against their adjoints and leaving the leftmost leg open, we can use \cref{purificationsym} to conclude that there exists some unitary matrix $F_1 \in \textup{U}(D)$ for which 
\[
\begin{tikzpicture}[scale=1]
    \draw (0, 1.834) rectangle (1.693, 1.27);
    \draw (1.129, 0.847) rectangle (1.693, 0.282);
    \draw (0.282, 2.117) -- (0.282, 1.834);
    \draw (1.411, 1.27) -- (1.411, 0.847);
    \draw (1.411, 0.282) -- (1.411, 0);
    \node[anchor=center] at (1.428, 0.54) {$C_1$};
    \draw (2.54, 1.834) rectangle (4.233, 1.27);
    \draw (2.822, 2.117) -- (2.822, 1.834);
    \draw (3.951, 1.27) -- (3.951, 0.847);
    \node[anchor=center] at (3.372, 1.556) {$\tilde B_1$};
    \node[anchor=center] at (0.832, 1.556) {$B_1$};
    \draw (1.411, 2.117) -- (1.411, 1.834);
    \draw (3.951, 2.117) -- (3.951, 1.834);
    \draw (3.669, 2.681) rectangle (4.233, 2.117);
    \draw (3.951, 2.963) -- (3.951, 2.681);
    \node[anchor=center] at (3.951, 2.399) {$F_1$};
    \node[anchor=center] at (2.117, 1.552) {$\propto$};
\end{tikzpicture}.
\]
In particular, as we are also assuming that $\tilde B_1$ is full-rank in the sense of \cref{eq:fullranktopleft}, we can rewrite \cref{eq:identitytoprow} as
\[
\input{TikzFigures/isompeps22.tikz}.
\]
Following an analogous argument for the rest of the tensors in the last row the result follows. Note that no full-rank assumption is needed for $B_2,\dotsc, B_L$ as they are unitary matrices.
\end{proof}

\subsection{Quotient tensor network manifold structure}

Let us now write the gauge freedom identified in the previous sub-section as a group action. Recall that the tensor network manifold associated with the tensor network family of isometric PEPS on an $L \times (L+1)$ lattice with bond dimension $D$ and physical dimension $d$ is given by 
\begin{align*}
M_{\textup{iso}} &= \bb{S}^{2D^4-1} \times \textup{V}_D(\bb{C}^{D^3})^{\times 2L-3} \times \textup{V}_{D^2}(\bb{C}^{dD^2})^{\times (L-2)^2} 
\\&\quad \times \textup{V}_{D^2}(\bb{C}^{d^2D})^{L-2} \times\textup{V}_D(\bb{C}^{dD})  \times \textup{U}(D^2)^{\times L-1}.
\end{align*}
To ease the notation, let us assume that $d = D$, which allows us to write $M_{\textup{iso}}$ as
\begin{align*}
M_{\textup{iso}} &= \bb{S}^{2d^4-1} \times \textup{V}_d(\bb{C}^{d^3})^{\times 2L-3} \times \textup{V}_{d^2}(\bb{C}^{d^3})^{\times (L-2)(L-1)} \times\textup{V}_d(\bb{C}^{d^2})  \times \textup{U}(d^2)^{\times L-1}.
\end{align*}
To describe a tensor network of this family up to a phase, we use
\[
N_{\textup{iso}} := \bb{CP}^{d^4-1} \times \textup{PV}_d(\bb{C}^{d^3})^{\times 2L-3} \times \textup{PV}_{d^2}(\bb{C}^{d^3})^{\times (L-2)(L-1)} \times\textup{PV}_d(\bb{C}^{d^2})  \times \textup{PU}(d^2)^{\times L-1},
\]
which we endow with its \textit{natural metric} $g_{\textup{iso}} := g_{\textup{FS}} \oplus g_{pst}^{\oplus L(L-1)} \oplus g_{pu}^{\oplus L-1}$.

\begin{definition}[Gauge action on isometric PEPS]
\label{def:GaugeIsomPEPS}
We define the action of $\textup{PU}(d)^{\times 2(L^2-L)}$ on $N_{\textup{iso}}$ as 
\[
\textup{PU}(d)^{\times 2(L^2-L)} \times N_{\textup{iso}} \to N_{\textup{iso}},\quad
(X, (U, V, W)) \mapsto (XU, XV, XW),
\]
where $U$ and $W$ denote the bottom and top rows of the PEPS, respectively, and $V$ denotes the \textit{bulk} of the PEPS. 
This way, 
\begin{align*}
U &:= ([\Lambda], [U_1], \dotsc, [U_{L-1}]) \in \bb{CP}^{d^4-1} \times \textup{PV}_d(\bb{C}^{d^3})^{\times L-1},\\
V &:= (V^1, \dotsc, V^{L-2}),\\
W &:= ([W_1], \dotsc, [W_{L}]) \in \textup{PV}_{d}(\bb{C}^{d^2}) \times \textup{PU}(d^2)^{\times L-1},
\end{align*}
where each $V^i$ denotes a row of the \textit{bulk} of the PEPS state, and so 
\[
V^i := ([V^i_1],\dotsc, [V^i_{L}]) \in \textup{PV}_{d}(\bb{C}^{d^3}) \times \textup{PV}_{d^2}(\bb{C}^{d^3})^{\times L-1},
\]
for every $i \in \{1, \dotsc, L-2\}$. We define $X$ as
\[
X := (X^1, \dotsc, X^L),
\]
where each $X^i$ acts on a row of the PEPS. Thus, $X^i \in \textup{PU}(d)^{\times 2L-1}$ for every $i \in \{1, \dotsc, L-1\}$, and $X^L \in \textup{PU}(d)^{\times L-1}$. More explicitly, $X^1$ acts on $U$, $X^i$ acts on $V^{i-1}$ for every $i \in \{2, \dotsc, L-1\}$ and $X^L$ acts on $W$. The action on the first row (see \cref{fig:actionfirstrowPEPS}) can be written as 
\begin{align*}
XU &:=  X^1 U 
\\&:=([(\mathds{1}\otimes X^1_1 \otimes X^1_2)\Lambda], 
\\&\quad[(\mathds{1} \otimes X^1_3 \otimes X^1_4)U_1(X^1_1)^\dagger],\dotsc, [(\mathds{1} \otimes X^1_{2(L-1)-1} \otimes X^1_{2(L-1)})U_{L-2}(X^1_{2(L-2)-1})^\dagger],
\\&\quad[(\mathds{1} \otimes \mathds{1} \otimes X^1_{2L-1}) U_{L-1} (X^1_{2(L-1)-1})^\dagger ]) \in \bb{CP}^{d^4-1} \times \textup{PV}_d(\bb{C}^{d^3})^{\times L-1}. 
\end{align*}

In the \textit{bulk} of the PEPS, represented by $V$, $X$ acts as
\[
XV := (XV^1, \dotsc, XV^{L-2}), 
\]
where for every row $i \in \{2, \dotsc, L-1\}$ represented by $V^{i-1}$ we define (see \cref{fig:actionbulkPEPS})
\begin{align*}
X V^{i-1} &:= ([(\mathds{1}\otimes X^i_1 \otimes X^i_2)V^{i-1}_1(X^{i-1}_2)^\dagger],
\\&\quad
[(\mathds{1} \otimes X^i_3 \otimes X^i_4)V^{i-1}_2(X^i_1 \otimes X^{i-1}_4)^\dagger],\dotsc,
\\&\quad [(\mathds{1} \otimes X^i_{2(L-1)-1} \otimes X^i_{2(L-1)})V^{i-1}_{L-1}(X^i_{2(L-2)-1} \otimes X^{i-1}_{2(L-1)})^\dagger],
\\&\quad[(\mathds{1} \otimes \mathds{1} \otimes X^i_{2L-1}) V^{i-1}_L (X^{i}_{2(L-1)-1} \otimes X^{i-1}_{2L-1})^\dagger]) \in  \textup{PV}_{d}(\bb{C}^{d^3}) \times \textup{PV}_{d^2}(\bb{C}^{d^3})^{\times L-1}.
\end{align*} 

Lastly, $X$ acts on the last row of the PEPS (see \cref{fig:actionlastrowPEPS}) as 
\begin{align*}
XW &:=([(\mathds{1} \otimes X^L_1) W_1(X^{L-1}_2)^\dagger], 
\\&\quad [(\mathds{1} \otimes X^L_2) W_2 (X^L_1 \otimes X^{L-1}_{4})^\dagger], \dotsc, [(\mathds{1} \otimes X^L_{L-1})W_{L-1}(X^L_{L-2} \otimes X^{L-1}_{2(L-1)})^\dagger],
\\&\quad [W_L(X^L_{L-1} \otimes X^{L-1}_{2L-1})^\dagger]) \in \textup{PV}_{d}(\bb{C}^{d^2}) \times \textup{PU}(d^2)^{\times L-1}.
\end{align*}

This action allows us to obtain the quotient tensor network manifold for which every point is in one-to-one correspondence with a state represented by an isometric PEPS up to a phase. See \cref{fig:diag7} for a visual representation of the setting considered. 

\begin{proposition}
\label{prop:quotTNlast}
Let $N_{\textup{iso}}$ be endowed with its natural metric $\tilde g_{\textup{iso}}$. Consider the action of $\textup{PU}(d)^{\times 2(L^2-L)}$ on $N_{\textup{iso}}$ as given in \cref{def:GaugeIsomPEPS}. There is a smooth manifold structure and a Riemannian metric $h_{\textup{iso}}$ on $B_{\textup{iso}} := N_{\textup{iso}}/\textup{PU}(d)^{\times 2(L^2-L)}$ for which the projection $(N_{\textup{iso}}, \tilde g_{\textup{iso}}) \to (B_{\textup{iso}}, h_{\textup{iso}})$ is a Riemannian submersion. 
\end{proposition}
\begin{proof}
The action is smooth and isometric as $N_{\textup{iso}}$ is endowed with its natural metric. It only remains to see that it is free. To do so, assume that there exists some $X \in \textup{PU}(d)^{\times 2(L^2-L)}$ and some $(U, V, W) \in N_{\textup{iso}}$ for which 
\[
(U, V, W) = (XU, XV, XW).
\]
Then, in particular, 
\[
W = XW, 
\]
and so 
\begin{align*}
[W_1] &= [(\mathds{1} \otimes X^L_1) W_1(X^{L-1}_2)^\dagger],\\
[W_i] &= [(\mathds{1} \otimes X^L_i) W_i (X^L_{i-1} \otimes X^{L-1}_{2i})^\dagger],\quad \forall i \in \{2, \dotsc, L-1\}\\
[W_L] &= [W_L(X^L_{L-1} \otimes X^{L-1}_{2L-1})^\dagger]).
\end{align*}
From the last equality, it holds that 
\[
[X^L_{L-1} \otimes X^{L-1}_{2L-1}] = [\mathds{1}],
\]
which implies that
\[
[X^L_{L-1}] = [\mathds{1}],\quad\textup{and}\quad [X^{L-1}_{2L-1}] = [\mathds{1}]. 
\]
Consequently, 
\[
[W_{L-1}] = [W_{L-1} (X^L_{L-2} \otimes X^{L-1}_{2(L-1)})^\dagger],
\]
which means that 
\[
[X^L_{L-1}] = [\mathds{1}],\quad\textup{and}\quad [X^{L-1}_{2(L-1)}] = [\mathds{1}].
\]
Continuing this reasoning, we conclude that $[X^{L-1}_{2L-1}] = [\mathds{1}]$ and 
\[
[X^L_i] = [\mathds{1}],\quad \textup{and}\quad [X^{L-1}_{2i}] = [\mathds{1}]
\]
for every $i \in \{1, \dotsc, L-1\}$. Looking now at $X^{L-1} V^{L-2}$ we know that
\begin{align*}
[V^{L-2}_1] &= [(\mathds{1}\otimes X^{L-1}_1 \otimes \mathds{1})V^{L-2}_1(X^{L-2}_2)^\dagger],\\
[V^{L-2}_i] &= [(\mathds{1} \otimes X^{L-1}_{2i-1} \otimes \mathds{1})V^{L-2}_i(X^{L-1}_{2i-1} \otimes X^{L-2}_{2i})^\dagger], \quad \forall i \in \{2, \dotsc, L-1\},\\
[V^{L-2}_L] &= [V^{L-2}_L (X^{L-1}_{2(L-1)-1} \otimes X^{L-2}_{2L-1})^\dagger]
\end{align*}
from which we deduce that
\[
[X^{L-2}_{2i}] = [\mathds{1}], \quad \textup{and}\quad [X^{L-1}_{2i-1}] = [\mathds{1}],
\]
for every $i \in \{1, \dotsc, L-1\}$. Repeating the same reasoning by \textit{moving downwards} along each row, we find that $X^i$ consists of the unit element for every $i \in \{2, \dotsc, L\}$, and it only remains to study $X^1$. In this case, we know that 
\begin{align*}
[\Lambda] &= [(\mathds{1}\otimes X^1_1 \otimes \mathds{1})\Lambda],\\
[U_i] &= [(\mathds{1} \otimes X^1_{2(i+1)-1} \otimes \mathds{1})U_i(X^1_{2i-1})^\dagger],\\
[U_{L-1}] &= [U_{L-1} (X^1_{2(L-1)-1})^\dagger ].
\end{align*}
Thus, the same argument as before shows that the action is free. The result then follows from \cref{existenceRiemannianSubmersionMetrics}. 
\end{proof}

\begin{figure}
    \centering
    \input{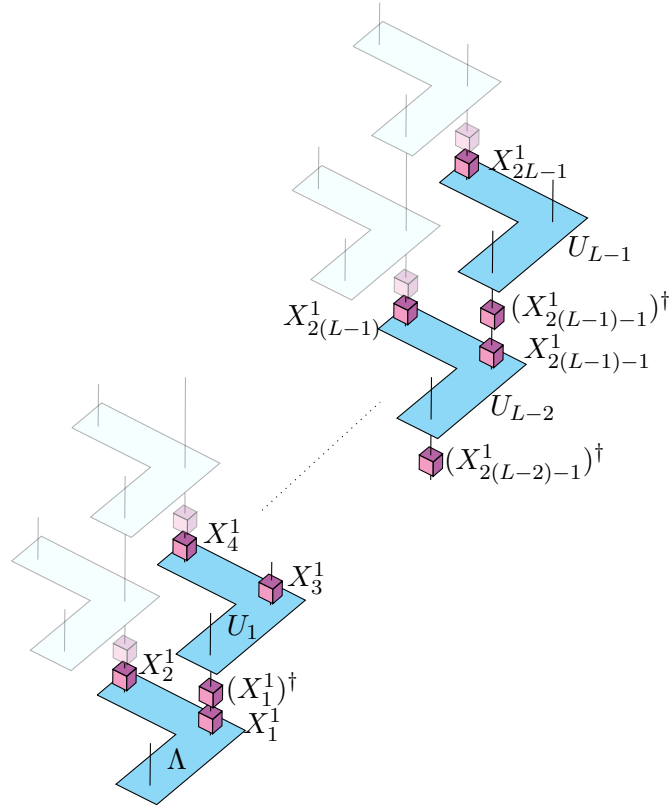}
    \caption{Visual description of the action on the first row of the PEPS, $XU$.}
    \label{fig:actionfirstrowPEPS}
\end{figure}
\begin{figure}
    \centering
    \input{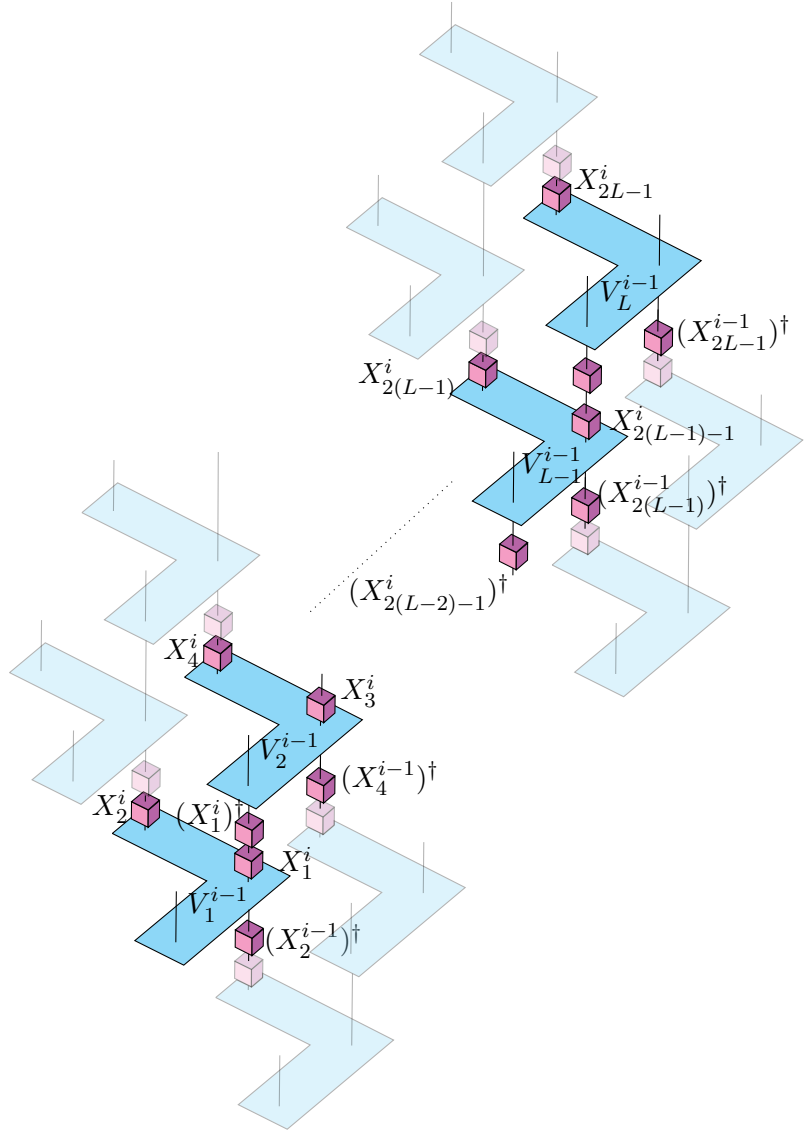}
    \caption{Visual description of the action on the $i$-th row of the PEPS, $XV^{i-1}$. }
    \label{fig:actionbulkPEPS}
\end{figure}
\begin{figure}
    \centering
    \input{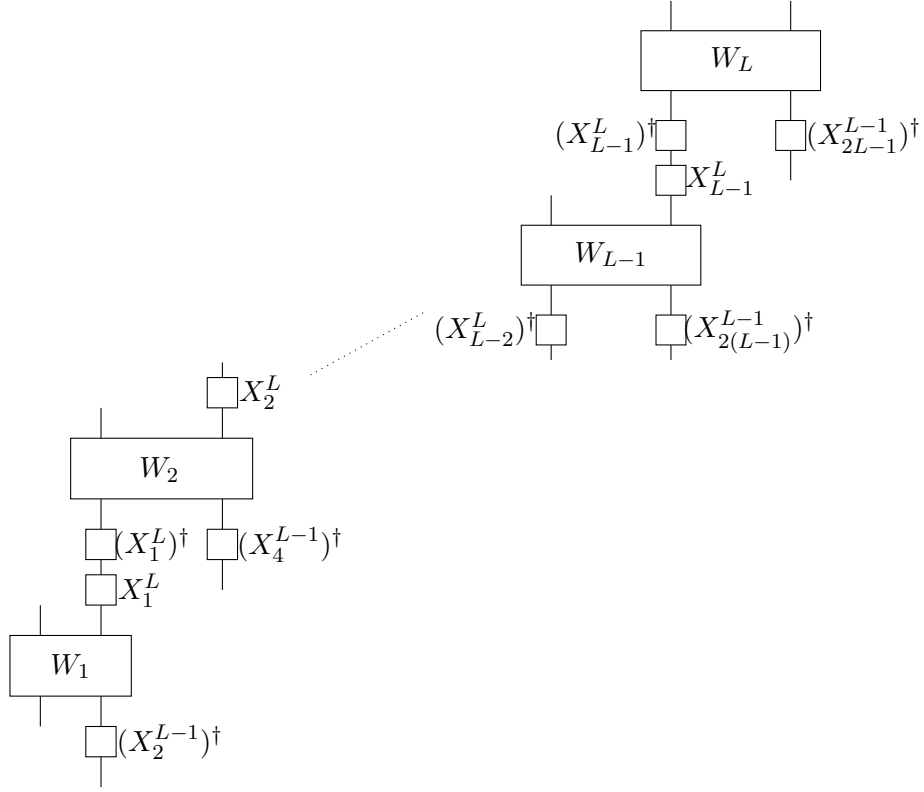}
    \caption{Visual description of the action on the last row of the PEPS, $XW$.}
    \label{fig:actionlastrowPEPS}
\end{figure}    
\end{definition}
\begin{figure}
    \centering
    \input{TikzFigures/diag7.tikz}
    \caption{Visual representation of the setting considered in \cref{prop:quotTNlast}.}
    \label{fig:diag7}
\end{figure}

%% file: Chapters/Geometry.tex
\section{Background in Riemannian geometry}
\label{SectionExamples}

In this section, we will introduce some auxiliary results and definitions that are used in the main text. First, we will give a brief introduction to Riemannian submersions and Lie groups. We will later present the manifolds used to construct the tensor network manifolds studied in this work. Lastly, we will give some results that guarantee that the actions considered in the main text are isometric. We assume previous knowledge on Riemannian manifolds. We refer to \cite{lee2018introductionRiemannian} for an introduction to this topic. 

\subsection{Riemannian submersions}
\label{sec:riemanniansubmersion}

Let us begin with a brief introduction to Riemannian submersions. For an in-depth discussion on this topic, we refer to classic references such as \cite{lee2018introductionRiemannian} or O'Neill's articles \cite{oneil1967submersions,oneill1966fundamental}. Let us first define a \textit{submersion} between Riemannian manifolds. 

\begin{definition}[Submersion]
Let $M$ and $B$ be two smooth manifolds, and let $F : M \rightarrow B$ be a smooth map. We say that $F$ is a \textup{submersion} if its differential $dF|_p$ is surjective for all $p \in M$. $M$ and $B$ are often referred to as the \textup{total} and \textup{base} spaces, respectively. 
\end{definition}

Given a submersion, the tangent space of the total space at each point can be split into its \textit{vertical} and \textit{horizontal} components. 

\begin{definition}[Vertical and horizontal tangent spaces]
\label{verticalhorizontaltangent}
Given a submersion $\pi: M \rightarrow B$, we define the \textup{vertical tangent space} at $p \in M$ as $V_p M := \ker d\pi|_p$. Note that $V_p M$ is a vector space of dimension $(n-m)$, where $n$ is the dimension of $M$ and $m$ is the dimension of $B$. This way, if $M$ is a Riemannian manifold, we can define the \textup{horizontal tangent space} as $H_p M := (V_p M)^{\bot}$. In particular,
\[
T_p M = V_p M \oplus H_p M,
\]
for every $p \in M$. 
\end{definition}

\textit{Riemannian submersions} are submersions which \textit{preserve} the metric between the total and the base space.

\begin{definition}[Riemannian Submersion]
Let $(M, g)$ and $(B, h)$ be two Riemannian manifolds, and let $\pi: M \rightarrow B$ be a submersion. We say that $\pi$ is a \textup{Riemannian submersion} if for every $p \in M$, $d\pi|_p$ maps $H_p M$ isometrically onto $T_{\pi(p)} B$. That is, for each $p \in M$ and $X, Y \in H_p M$, 
\[
\langle X, Y \rangle_g = \langle d\pi|_p X,d\pi|_p Y \rangle_h.
\]
\end{definition}

One of the most widely studied methods for constructing Riemannian submersions is by considering quotient maps with respect to a group action defined on a Riemannian manifold. This is the setting considered in the main text. Let us begin by defining group actions on manifolds. 

\begin{definition}[Group action on a manifold]
Let $G$ be a group and $M$ be a manifold. A \textup{left action} of $G$ on $M$ is a map $G \times M \to M$, usually denoted as $(x, p) \mapsto x\cdot p$ verifying the following: 
\begin{itemize}
    \item $x_2 \cdot (x_1 \cdot p) = (x_2x_1)\cdot p$ for every $x_1, x_2 \in G$, $p \in M$.
    \item $e \cdot p = p$ for every $p \in M$, $e$ being the unit of $G$. 
\end{itemize}
Right actions are defined analogously. 
\end{definition}

Although group actions often give rise to Riemannian submersions, not every group action does. Some sufficient conditions for a group action to give rise to a Riemannian submersion are for the action to be \textit{smooth, free, proper} and \textit{isometric}. 

\begin{definition}[Smooth action]
The action of a Lie group $G$ on a manifold $M$ is \textup{smooth} if the map $G \times M \rightarrow M$ given by $(x, p)\mapsto x \cdot p$ is smooth. 
\end{definition}

\begin{definition}[Free action]
A group action of $G$ on a manifold $M$ is \textup{free} if for every $p \in M$, 
\[
\{x \in G : x p = p\} = \{e\},
\]
where $e$ denotes the unit of $G$. 
\end{definition}

\begin{definition}[Proper action]
A group action of $G$ on a manifold $M$ is said to be \textup{proper} if the map 
\begin{align*}
G \times M &\rightarrow M \times M\\
(x, p) &\mapsto (x\cdot p, p)    
\end{align*} 
is proper, i.e., the preimage of every compact set is compact.
\end{definition}

In particular, every smooth action of a compact Lie group is always proper (cf. \cite[Corollary C.16]{lee2018introductionRiemannian}). 

\begin{definition}[Isometric action]
A group action of $G$ on a Riemannian manifold $(M, g)$ is said to be \textup{isometric} if, for every $x \in G$, the map 
\begin{align*}
\alpha_x : M &\to M\\
p &\mapsto x \cdot p
\end{align*}
is an isometry on $(M, g)$ (i.e. a diffeomorphism such that for every point $p \in M$ and every $u, v \in T_p M$, $g(u, v) = g(d\alpha_p(u), d\alpha_p(v))$ (cf. \cref{def:isometry})). 
\end{definition}

As we mentioned earlier, every group action satisfying the above properties induces a Riemannian submersion. 

\begin{proposition}[{\cite[Corollary 2.29]{lee2018introductionRiemannian}}]
\label{existenceRiemannianSubmersionMetrics}
Let $(M, g)$ be a Riemannian manifold, and let $G$ be a Lie group acting smoothly, freely, properly and isometrically on $M$. Then, the orbit space $M/G$ has a unique smooth manifold structure and Riemannian metric such that $\pi : M \rightarrow M/G$ is a Riemannian submersion.
\end{proposition}

\subsection{Lie groups}
\label{sec:SecLieGroups}
We now introduce some of the most elementary definitions and properties of Lie groups. Although Lie groups can be both studied from the point of view of algebra and differential geometry, we will focus on the latter approach. For a more general and detailed introduction to Lie groups, see for example \cite{gallier2020differential}.

\begin{definition}[Lie group]
A \textup{Lie group} is a group $G$ such that it can be endowed with a smooth (real) manifold structure in such a way that both the group operation and the inverse map are smooth. 
\end{definition}

We consider two Lie groups in this text. The first is the unitary group
\[
\textup{U}(n) := \{ U \in \textup{M}_{n}(\bb{C}) : U^\dagger U = \mathds{1}\},
\]
endowed with its \textit{bi-invariant} metric, which is given by 
\[
g_{\textup{bi}}(X, Y) := \Tr(Y^\dagger X),
\]
for every $U \in \textup{U}(n)$ and any $X, Y \in T_U \textup{U}(n)$. 

The second Lie group considered is the projective unitary group $\textup{PU}(n)$, 
\[
\textup{PU}(n) := \{[U] : U \in \textup{U}(n)\},
\]
where $[U] := \{\lambda U : \lambda \in \textup{U}(1)\}$ for every $U \in \textup{U}(n)$. While it is well-known that $\textup{U}(n)$ is a compact Lie group, to the best of our knowledge the group $\textup{PU}(n)$ is not as studied in the literature. $\textup{PU}(n)$ is indeed a group with respect to the product defined as 
\[
[U] \cdot [V] := [UV],
\]
for every $[U], [V] \in \textup{PU}(n)$. In order to ensure that it is a \textit{Lie group}, let us first define the concept of \textit{normal subgroup} of a group. 

\begin{definition}[Normal subgroup]
Let $G$ be a group and $H$ be a subgroup of $G$. We say that $H$ is a \textup{normal subgroup} of $G$ if 
\[
xH = Hx
\]
for every $x \in G$. 
\end{definition}

The quotient of a Lie group by one of its closed normal subgroups results in a Lie group. 
\begin{theorem}[{\cite[Theorem 21.26]{lee2003smooth}}]
\label{quotientmanifoldLiegroups}
Suppose $G$ is a Lie group and $K \subseteq G$ is a closed normal subgroup of $G$. Then $G/K$ is a Lie group.
\end{theorem}

From the above result, it is straightforward to conclude that $\textup{PU}(n)$ is a Lie group, as it arises as the quotient of $\textup{U}(n)$ by its---normal and closed---subgroup
\[
\{\lambda \mathds{1} : \lambda \in \textup{U}(1)\},
\]
where $\mathds{1}$ denotes the identity matrix of size $n$. 

Lastly, observe that the tensor product operation presented in the introduction can be extended to the projected unitary group. 
\begin{remark}
The tensor product map
\begin{align*}
\textup{U}(n) \times \textup{U}(n) &\to \textup{U}(n^2)\\
(X, Y) &\mapsto X \otimes Y,
\end{align*}
which corresponds to the restriction of a bi-linear function to the set $\textup{U}(n) \times \textup{U}(n)$, induces a well-defined map on $\textup{PU}(n) \times \textup{PU}(n)$,
\begin{align*}
\textup{PU}(n) \times \textup{PU}(n) &\to \textup{PU}(n^2)\\
([X], [Y]) &\mapsto [X \otimes Y].
\end{align*}
\end{remark}

\subsection{Manifolds of interest}
\label{sec:manifolds}

Let us now give a short description of the different manifolds that are considered in this work. Every tensor of the tensor network families considered in this work can be seen as a point in the sphere, a unitary matrix or a linear isometry. Thus, the tensor network manifolds studied are obtained by considering the Cartesian product of these three manifolds. 

We already saw above that the unitary group is a Lie group, and that the Riemannian metric that we consider for it is the bi-invariant metric. We also consider the sphere $\bb{S}^{2n-1}$ which we understand as a real manifold embedded in $\bb{C}^n \simeq \bb{R}^{2n}$. The metric considered for the sphere is known as the round metric, $g_{\textup{round}}$, which is the one induced by the Euclidean metric on $\bb{R}^{2n}$

Linear isometries can be seen as complex matrices $U$ of size $n \times k$, with $1 \leq k \leq n$, such that $U^\dagger U = \mathds{1}$. The set of linear isometries from $\bb{C}^k$ to $\bb{C}^n$, 
\[
\textup{V}_k(\bb{C}^n) := \{X \in \textup{M}_{n \times k} : X^\dagger X = \mathds{1}\},    
\]
can be endowed with a smooth manifold structure (cf. \cite{autenried2014sub}). The manifold of linear isometries is known as the Stiefel manifold and it can also be understood as the quotient of $\textup{U}(n)$ by the group action of $\textup{U}(n - k)$ defined as
\begin{equation}
\label{eq:actiondefstiefel}
\textup{U}(n-k) \times \textup{U}(n) \to \textup{U}(n),\quad 
(X, U) \mapsto U\begin{pmatrix}
\mathds{1} & 0\\
0 & X
\end{pmatrix},
\end{equation}
where $\mathds{1}$ is the identity of size $k$ (cf. \cite{autenried2014sub}). We consider the Stiefel manifold endowed with the metric $g_{\textup{st}}$ for which the projection $(\textup{U}(n), g_{\textup{bi}}) \to (\textup{V}_{k}(\bb{C}^{n}), g_{\textup{st}})$ is a Riemannian submersion.

As we mentioned earlier, the tensor network manifolds associated with the tensor network families studied in \crefrange{onedimcircuitsfixed}{sec:IsomPEPS} are defined as the product of the above manifolds. The tensor network manifolds are always assumed to be endowed with the product metric, and each component is assumed to be endowed with its corresponding metric ($g_{\textup{bi}}$ for the unitary group, $g_{\textup{round}}$ round for the sphere and $g_{\textup{st}}$ for the Stiefel manifold). We call this product metric the \textit{natural metric} of the tensor network manifold. 

In the main text we also considered quotient tensor network manifolds which allowed to describe a tensor network up to a phase. In order to do so, we considered the quotient of the original tensor network manifold by the group action given by the multiplication by a complex phase on each factor of the product. 

The manifold that arises when considering the quotient of $\textup{U}(n)$ by the multiplication by a complex phase was already studied in \cref{sec:SecLieGroups}, and corresponds to the \textit{projective unitary group} $\textup{PU}(n) (\simeq \textup{U}(n)/\textup{U}(1))$, 
\[
\textup{PU}(n) = \{[U] : U \in \textup{U}(n)\},
\]
where $[U] = \{\lambda U : \lambda \in \textup{U}(1)\}$. We consider $\textup{PU}(n)$ endowed with the metric $g_{\textup{pu}}$ for which the projection $(\textup{U}(n), g_{\textup{bi}}) \to (\textup{PU}(n), g_{\textup{pu}})$ is a Riemannian submersion.

If we consider the quotient of the sphere $\bb{S}^{2n-1}$ by the multiplication by a complex phase we obtain the \textit{complex projective space} $\bb{CP}^{n-1} (\simeq \bb{S}^{2n-1} / \textup{U}(1))$, 
\[
\bb{CP}^{n-1} := \{[z] : z \in \bb{S}^{2n-1}\},
\]
where $[z] := \{\lambda z : \lambda \in \textup{U}(1)\}$. We endow $\bb{CP}^{n-1}$ with the metric $g_{\textup{FS}}$ for which the projection $(\bb{S}^{2k+1}, g_{\textup{round}}) \to (\bb{CP}^{k}, g_{\textup{FS}})$ is a Riemannian submersion. This metric is known as the Fubini-Study metric \cite{lee2018introductionRiemannian}.

Lastly, when we consider the quotient of the Stiefel manifold by the multiplication by a complex phase we obtain the \textit{projective Stiefel manifold} $\textup{PV}_k(\bb{C}^{n}) (\simeq \textup{V}_k(\bb{C}^{n}) / \textup{U}(1))$ \cite{gondhali2013vector}, 
\[
\textup{PV}_k(\bb{C}^{n}) := \{[X] : X \in \textup{M}_{n \times k}(\bb{C}),\ X^\dagger X = \mathds{1}\},
\]
where $[X] = \{\lambda X: \lambda \in \textup{U}(1)\}$. Using the definition of the Stiefel manifold from \cref{eq:actiondefstiefel}, we can also rewrite the above description of $\textup{PV}_k(\bb{C}^n)$ as
\[
\textup{PV}_k(\bb{C}^{n}) = \{[U] : U \in \textup{U}(n)\},
\]
where 
\[
[U] = \{\lambda U \begin{pmatrix}
\mathds{1} & 0\\
0 & V
\end{pmatrix} : \lambda \in \textup{U}(1), V \in \textup{U}(n-k)\}.
\]
We consider the projective complex Stiefel manifold endowed with the metric $g_{\textup{pst}}$ for which the projection $(\textup{U}(n), g_{\textup{bi}}) \to (\textup{PV}_k(\bb{C}^{n}), g_{\textup{pst}})$ is a Riemannian submersion.

Again, the quotient tensor network manifolds that arise as the Cartesian product of the projective unitary group, the complex projective space and the projective Stiefel manifold are always endowed with the product of the metrics considered above, also denoted as the \textit{natural metric}. 

\subsection{Isometric actions}
\label{sec:isometries}

Lastly, let us prove some auxiliary results that will allow us to guarantee that the actions considered in the main text are isometric. We begin by defining an isometry between Riemannian manifolds. 

\begin{definition}[Isometry]
\label{def:isometry}
Let $(M, g)$ and $(N, h)$ be two Riemannian manifolds. We say that $f : M \to N$ is an isometry if $f$ is a diffeomorphism and $df_p$ is an isometry between $T_p M$ and $T_{f(p)} N$ for every $p \in M$, i.e. 
\[
g(u, v) = h(df(u), df(v)),
\]
for every $p \in M$ and every $u, v \in T_p M$. 
\end{definition}

Next, let us show that left and right multiplication in the unitary group endowed with its bi-invariant metric are isometric actions. 

\begin{proposition}[Matrix multiplication is an isometric action on $\textup{U}(n)$]
\label{prop:isomunitaries}
Let $(\textup{U}(n), g_{\textup{bi}})$ be the unitary group endowed with its bi-invariant metric. Let $\textup{U}(n)$ act on itself by left (resp. right) multiplication. Then, for every $U \in \textup{U}(n)$ the map
\[
L_U : (\textup{U}(n), g_{\textup{bi}}) \to (\textup{U}(n), g_{\textup{bi}}),\quad p \mapsto Up ,
\]
(resp. $p \mapsto pU$) is an isometry. 
\end{proposition}
\begin{proof}
Let us prove the result for the left multiplication. The proof for the right multiplication is an analogous.  

For every $p \in \textup{U}(n)$ and every $v \in T_p \textup{U}(n)$, let $\gamma_{p, v}(t)$ be the unique geodesic in $\textup{U}(n)$ such that $\gamma_{p, v}(0) = p$ and $\gamma_{p, v}'(0) = v$. Then, we know that 
\[
dL_U|_p(v) = \left.\frac{d}{dt}\right|_{t = 0} L_U(\gamma_{p, v}(t)) = \left.\frac{d}{dt}\right|_{t = 0}U\gamma_{p, v}(t) = U\gamma'_{p, v}(0) = Uv,
\] 
and so $dL_U|_p : T_p \textup{U}(n) \to T_{Up} \textup{U}(n)$ acts by left-multiplication for every $U \in G$ and every $p \in \textup{U}(n)$. Therefore, since $U$ is unitary, 
\[
g_{\textup{bi}}(dL_U(v), dL_U(w)) = g_{\textup{bi}}(Uv, Uw) = \Tr(w^\dagger U^\dagger Uv) = \Tr(w^\dagger v) = g_{\textup{bi}}(v, w),
\]
for every $p \in \textup{U}(n)$ and every $v, w \in T_p \textup{U}(n)$, concluding the proof. 
\end{proof}

We will now prove that left multiplication on the sphere by a unitary matrix is also an isometric action, whenever the sphere is endowed with the round metric. To do so, let us first study the round metric on the sphere $\bb{S}^{2n-1}$, when understood as a subset of $\bb{C}^n$. 

\begin{lemma}[Round metric on the sphere]
Let $(\bb{S}^{2n-1}, g_\textup{round})$ be the sphere endowed with the round metric. On the one hand, we can see the round sphere embedded in $\bb{R}^{2n}$, i.e. every $x \in \bb{S}^{2n-1}$ can be written as
\[
x = (x_1, y_1, \dotsc, x_n, y_n)^T \in \bb{R}^{2n}, 
\]
such that $x^T x = 1$. On the other hand, we can see it as embedded in the complex space $\bb{C}^n$, i.e. we can write every $\tilde x \in \bb{S}^{2n-1}$ as
\[
\tilde x = (\tilde x_1 + i \tilde y_1, \dotsc, \tilde x_n + i\tilde y_n)^T \in \bb{C}^{n}, 
\]
such that $\tilde x^\dagger \tilde x = 1$. When considering this second picture, $g_\textup{round}$ can be written as 
\[
g_\textup{round}(v, w) = \textup{Re}(w^\dagger v),
\]
for every $p \in \bb{S}^{2n-1}$ and every $v, w \in T_p \bb{S}^{2n-1}$.
\end{lemma}
\begin{proof}
When seeing the sphere as embedded in $\bb{R}^{2n-1}$, every tangent vector $v$ at a point $x$, $v \in T_x \bb{S}^{2n-1}$, can be understood as an element in $\bb{R}^{2n}$ verifying $v^T x = 0$. Recall that the round metric is that induced by the Euclidean metric via the inclusion map from the sphere to $\bb{R}^{2n}$. For this reason, for every $x \in \bb{S}^{2n-1}$, the inner product between every two tangent vectors $v, w \in T_x \bb{S}^{2n-1}$ is given by
\[
g_\textup{round}(v, w) = w^T v. 
\]
In particular, if we write $v$ and $w$ in coordinates as 
\begin{align*}
v &= (a_1, b_1, \dotsc, a_n, b_n)^T \in \bb{R}^{2n},\\
w &= (c_1, d_1, \dotsc, c_n, d_n)^T \in \bb{R}^{2n},
\end{align*}
we can rewrite their inner product as 
\[
g_\textup{round}(u, v) = \sum_{i = 1}^n  a_i c_i + \sum_{i = 1}^n b_i d_i.
\]
Rewriting this expression using the complex space embedding picture, both $v$ and $w$ can be written as
\begin{align*}
\tilde v &= (a_1 + ib_1, \dotsc, a_n + ib_n)^T \in \bb{C}^{n}, \\
\tilde w &= (c_1 + id_1, \dotsc, c_n + id_n)^T \in \bb{C}^n,
\end{align*}
and their inner product can be written as
\[
g_\textup{round}(\tilde v, \tilde w) = \textup{Re}(\tilde w^\dagger \tilde v) =  \sum_{i = 1}^n  a_i c_i + \sum_{i = 1}^n b_i d_i,
\]
which coincides with that shown above. 
\end{proof}

\begin{proposition}[Left multiplication by a unitary matrix is an isometric action on the sphere]
\label{prop:isomsphere}
Let $(\bb{S}^{2n-1}, g_\textup{round})$ be the sphere endowed with the round metric. Let $\textup{U}(n)$ act on $(\bb{S}^{2n-1}, g_\textup{round})$ by left multiplication, seeing the sphere as embedded in $\bb{C}^n$. Then, for every $U \in \textup{U}(n)$ the map
\[
L_U : (\bb{S}^{2n-1}, g_\textup{round}) \to (\bb{S}^{2n-1}, g_\textup{round}),\quad p \mapsto Up ,
\]
is an isometry. 
\end{proposition}
\begin{proof}
Following the same reasoning as in the proof of \cref{prop:isomunitaries} we conclude that for every $p \in \bb{S}^{2n-1}$, $dL_U|_p : T_p\,\bb{S}^{2n-1} \to T_{Up}\,\bb{S}^{2n-1}$ also acts by left multiplication, i.e. $dL_U|_p(v) =  Uv$ for every $p \in \bb{S}^{2n-1}$ and any $v \in T_p \bb{S}^{2n-1}$. Moreover, 
\[
g_\textup{round}(Uv, Uw) = \textup{Re}(w^\dagger U^\dagger Uv) = \textup{Re}(w^\dagger v) = g_\textup{round}(v, w),
\]
finishing the proof. 
\end{proof}

To finish, we show a result that relates isometries and Riemannian submersions.
\begin{proposition}
\label{projectiveIsometry}
Let $\pi: (M, g) \to (B, \tilde{g})$ be a Riemannian submersion. Consider an isometry $f: M \to M$ for which there exists some function $\tilde{f}: B \to B$ fulfilling $\pi \circ f = \tilde{f} \circ \pi$. Then $\tilde{f}$ is also an isometry.
\end{proposition}
\begin{proof}
Let $x \in M$ and consider two horizontal vectors $u, v \in H_x M$. Since $\pi$ is a Riemannian submersion and $f$ is an isometry on $(M, g)$ we know that 
\[
\tilde{g}(d\pi(u), d\pi(v)) = g(u, v) = g(df(u), df(v)).
\]
Next, since $\pi \circ f = \tilde{f} \circ \pi$, it follows that $df$ preserves the horizontality of vectors. Indeed, for any vertical vector $w \in V_x M$, it follows that 
\[
d\pi(df(w)) = d\tilde f(d\pi(w)) = d\tilde f(0)= 0. 
\]
Thus, using again that $\pi$ is a Riemannian submersion, we know that
\[
g(df(u), df(v)) = \tilde{g}(d\pi(df(u)), d\pi(df(v))) . 
\]
Lastly, since by assumption $\pi \circ f = \tilde{f} \circ \pi$, we conclude that 
\[
\tilde{g}(d\pi(df(u)), d\pi(df(v))) = \tilde{g}(d\tilde{f}(d\pi(u)), d\tilde{f}(d\pi(v))),
\]
finishing the proof.
\end{proof}

The above result, along with \cref{prop:isomunitaries,prop:isomsphere} allow us to conclude that all of the group actions considered in the main text are isometric.